\documentclass[journal]{IEEEtran}
  \usepackage[pdftex]{graphicx}
\usepackage{amsmath}
\ifCLASSOPTIONcompsoc
  \usepackage[caption=false,font=normalsize,labelfont=sf,textfont=sf]{subfig}
\else
  \usepackage[caption=false,font=footnotesize]{subfig}
\fi
\usepackage{fixltx2e}
\usepackage{amssymb}
\usepackage{xcolor}
\usepackage[subnum]{cases}
\usepackage{mathtools}
\usepackage{enumitem}
\usepackage{comment}
\usepackage{cite}
\usepackage{amsthm}

\newtheorem{defn}{Definition}
\newtheorem{prop}{Proposition}
\newtheorem{thm}{Theorem}
\newtheorem{lemma}{Lemma}
\newtheorem{cor}{Corollary}
\newtheorem{rem}{Remark}
\newtheorem{assum}{Assumption}

\DeclareMathOperator{\tr}{tr}
\DeclareMathOperator{\diag}{diag}
\DeclareMathOperator{\bdiag}{blockdiag}
\DeclareMathOperator{\circulant}{circ}
\DeclareMathOperator{\R}{Re}
\DeclareMathOperator{\I}{Im}
\DeclareMathOperator{\spn}{span}

\newcommand{\blue}[1]{{\color{blue}#1}}
\DeclareMathOperator{\spec}{spec}

\renewcommand{\baselinestretch}{0.99}

\AtBeginDocument{\colorlet{defaultColor}{.}}
\newcommand{\defaultColor}[1]{{\color{defaultColor}#1}}

\begin{document}
\title{Performance of Single and Double-Integrator Networks over Directed Graphs}

\author{H.~Giray~Oral,\IEEEmembership{}
        Enrique~Mallada\IEEEmembership{}
        and~Dennice~F.~Gayme~\IEEEmembership{}% <-this % stops a space
\thanks{H. G. Oral and D. F. Gayme are with the Dept. of Mechanical Engineering, E. Mallada is with the Dept. of Electrical and Computer Engineering at The Johns Hopkins University, Baltimore, MD, USA, 21218,
        {\texttt{giray@jhu.edu, dennice@jhu.edu, mallada@jhu.edu}}. Partial support by the NSF (ECCS 1230788, CNS 1544771, EPCN 1711188) and ARO W911NF-17-1-0092 is gratefully acknowledged.}% <-this % stops a space
}


\markboth{}%
{Shell \MakeLowercase{\textit{et al.}}: }



\maketitle


\begin{abstract}

This paper 
%studies
\defaultColor{provides a framework to evaluate}
%\blue{evaluates} 
the performance of single and double integrator networks over arbitrary directed graphs. Adopting vehicular network terminology, we consider quadratic performance metrics defined by the \mbox{$\mathcal{L}_2$-norm} of position and velocity based response \defaultColor{functions given} 
%to impulses
\defaultColor{impulsive inputs to} 
%on 
each vehicle. We 
\defaultColor{exploit the spectral properties of weighted graph Laplacians and output performance matrices to derive}
%provide 
a novel method of computing the closed-form solutions for this general class of \defaultColor{performance} metrics, 
%by exploiting spectral properties of weighted graph Laplacians and output performance matrices. 
\defaultColor{which include $\mathcal{H}_2$-norm based quantities as special cases.}
%Such metrics include, as a particular case, the $\mathcal{H}_2$-norm. 
We then explore the effect of the interplay between network properties  (such as edge directionality and connectivity) and the control strategy on the overall network performance.
More precisely, for 
%graphs with 
systems whose interconnection is described by graphs with
%diagonalizable 
\defaultColor{normal}
Laplacian $L$, we characterize the role of directionality by comparing \defaultColor{their} performance with \defaultColor{that of} their undirected counterparts, represented by the Hermitian part of $L$. We show that, for single-integrator networks, directed and undirected graphs perform identically. However, for double-integrator networks, graph directionality -expressed by 
%Laplacian eigenvalues 
\defaultColor{the eigenvalues of $L$}
with nonzero imaginary part- can significantly degrade performance. Interestingly, in many cases, well-designed feedback can also exploit directionality to mitigate degradation or even
improve the performance to 
exceed that
%the performance 
of the undirected case.
%We finally focus on the system coherence metric 
\defaultColor{Finally we focus on a system coherence metric
%defined
%-i.e.,
%in terms of the 
%-represented by
-aggregate deviation from the state average-}
%- 
to investigate the relationship between performance and degree of connectivity, leading to somewhat surprising findings. For 
%instance,
\defaultColor{example,} 
%we show that} 
increasing the number of neighbors on a \mbox{$\omega$-nearest} neighbor directed graph does not necessarily improve performance. 
%Along the same line, 
\defaultColor{Similarly,}
we demonstrate equivalence in performance between all-to-one and all-to-all communication graphs.

\iffalse
%old abstracts begin

%old abstracts end
\fi

\end{abstract}

\begin{IEEEkeywords}
$\mathcal{L}_2$, $\mathcal{H}_2$ norm, directed graph, performance.
%spatially invariant system.
\end{IEEEkeywords}

% \fi
\IEEEpeerreviewmaketitle

%INTRODUCTION
\section{Introduction}

%INTRODUCTION
\IEEEPARstart{C}{onditions} for reaching consensus -achieving a synchronized steady state- \defaultColor{have been} widely studied for networked dynamical systems, \defaultColor{see e.g.} \cite{RenAtkins2005, YuChen2010, ZhuTian2009}.
A related and equally important question is how the system performs 
%as a whole 
in its effort to restore and/or maintain synchrony in the face of disturbances. This synchronization performance can be interpreted as a 
%metric for 
measure of
%the system's 
efficiency or robustness, and has been evaluated, for example, in terms of the lack of coherence or the degree of disorder in first order (single-integrator) \cite{SiamiMotee2016, YoungLeonard2010, BamiehJovanovic2012, DezfulianMotee2018, SarkarDahleh2018, MaElia2015, LinJovanovic2012_2} and second order (double-integrator) \cite{BamiehJovanovic2012, GrunbergGayme2018, teglingbamieh2019, PatesRantzer2017, HaoBarooah2013, LinJovanovic2012, OralMalladaGayme2017 
} 
consensus networks. 
Robustness metrics in power systems (e.g. transient real power losses or phase/frequency incoherency) have been assessed in transmission and inverter-based networks \cite{BamiehGayme2013, TeglingBamiehGayme2015, 
TeglingGayme2015,  m2016cdc, jpm2017cdc, wjzmdd2018tac,  wdj2016, sppmd2016, oralgaymeecc2019}.
%renewable energy integrated power networks 
% 
%as well as in microgrids 
% 
%In addition, 
Controllers that have been proposed to improve these types of performance include
%Dynamic controllers
dynamic feedback
 \cite{m2016cdc,jpm2017cdc, wjzmdd2018tac, 
teglingbamieh2019
} 
and optimization based approaches \cite{wdj2016, sppmd2016}. 
%have been proposed for 
%performance improvement.
%\blue{
%improving these types of performance metrics.}

%In many contexts, networked dynamical systems can be modeled as single or double integrator systems that interact over a graph through coupling functions, and are subjected to distributed disturbances. 
%
A widely utilized approach to quantify performance in systems subjected to distributed disturbances is to select a system output such that the desired metric is defined through the input-output $\mathcal{H}_2$ norm of the system. Certain $\mathcal{H}_2$ based performance metrics for systems whose underlying graphs are undirected can be obtained in closed form, e.g. \cite{BamiehJovanovic2012, GrunbergGayme2018, TeglingBamiehGayme2015, 
%SjodinGayme2014, 
sppmd2016, jpm2017cdc}. 
%Similar
\defaultColor{Related} 
performance metrics have also been evaluated in terms of the effective resistance of undirected graphs \cite{GrunbergGayme2018, TylooJacquod2019, pattersonYiZhang2019}, which allows for efficient computational approaches \cite{EllensKooij2011}. 
The notion of effective resistance has been extended to directed graphs \cite{YoungLeonard2016, YoungLeonard2016_2},
%yet
however 
its application to performance analysis remains an open question.
%has also been considered in an effort to potentially facilitate the performance assessment. 

%Most
Much 
of the existing literature on 
\defaultColor{evaluating the performance in systems with}
directed interconnection topologies considers restrictive scenarios on the graph topology 
%(especially,
(e.g. 
spatially invariant 
\cite{
%YoungLeonard2010, DezfulianMotee2018, OralMalladaGayme2017, 
oralgaymeacc2019} 
%or
and 
nearest-neighbor type interactions \cite{LinJovanovic2012}; or systems with normal Laplacian matrices
\cite{YoungLeonard2010, DezfulianMotee2018, OralMalladaGayme2017
%oralgaymeacc2019
})  
with closed-form solutions obtained \defaultColor{only} for specific metrics (full state \cite{SiamiMotee2017}, degrees of disorder \cite{BamiehJovanovic2012}, etc.). 
%were considered for a limited class of performance measures that quantify macroscopic or microscopic degrees of disorder \cite{BamiehJovanovic2012}.
%Performance bounds were obtained for more general quadratic measures and single-integrator systems defined over directed graphs that emit strongly connected and balanced weighted Laplacians \cite{SiamiMotee2017}. 
%However, whether one can obtain closed-form solutions for more general performance measures in the case of arbitrary digraphs remains an open question. 
%More recently, 
%More 
\defaultColor{Closed-form expressions for more}
general quadratic performance metrics of double-integrator networks over undirected graphs 
%have 
%also 
%been 
formulated in terms of the $\mathcal{L}_2$ norm of the system output 
%and closed-form solutions 
have \defaultColor{also} been obtained \cite{PaganiniMallada2017, PaganiniMallada2019, ColettaJacquod2019}.
%
%\iffalse
\defaultColor{An extension to directed graphs with diagonalizable Laplacian matrices 
%has followed, 
was provided for $\mathcal{H}_2$ based 
%quadratic 
metrics \cite{jmg2018},
however a 
%complete 
precise
understanding of the role that the underlying network architecture plays is still lacking.}
%\fi
%
%
%However,
%Therefore,
\defaultColor{Although the results described above represent progress into a wide range of special cases,} 
%despite all of the above stated advances, 
%there is still a lack of 
a unified treatment of general performance metrics over \emph{arbitrary} directed graphs has yet to be developed. This paper aims to lay the foundations for such 
%study 
a framework
via the following contributions:
%Therefore in this work, we would like to bridge the gap between the $\mathcal{L}_2$ norm based performance measures of systems emitting undirected graphs and performance assessment of systems defined over directed graphs. Furthermore, we would like to investigate a broader class of directed interaction topologies.

%\subsubsection*{Contributions}
%\blue{
%The contributions of this paper are manifold:}
%
\begin{enumerate}[wide, labelwidth=!, labelindent=0pt]
\item We provide a novel unifying approach 
%method of computing 
to compute
a general class of quadratic performance metrics for single and double integrator systems defined over directed graphs that have at least one globally reachable node.
\item 
%Using
We use 
the closed-form solutions resulting from this 
%unifying 
approach to demonstrate that overall network performance is determined by an interaction between network topological characteristics (e.g. edge directionality and connectivity) and the control strategy. In particular, we show that
\begin{enumerate}
\item The effect of edge directionality on performance can be characterized by the respective spectral structures of 
%communication 
\defaultColor{Laplacian}
and output \defaultColor{matrices}, 
%which facilitates judicious feedback.
which needs to be accounted for in judicious feedback design.
\item While performance is sensitive to the degree of connectivity in directed graphs, \defaultColor{the} relationship is not monotonic.
\end{enumerate} 
%by focusing on 
%the subclass of directed graphs emitting diagonalizable Laplacians.
\end{enumerate}
%As the first main contribution of this work, 
%
%For single and double integrator systems defined over directed graphs with a globally reachable node, 
%We provide closed-form solutions for a general class of quadratic 
%\noindent
%A more precise summary of the contributions is provided next. 

We develop the framework outlined above by formulating
%We formulate 
the performance metrics through the $\mathcal{L}_2$ norm of the system response due to distributed impulse disturbances.   Adopting the terminology from vehicular networks, the metrics are 
%either position or velocity based, 
%i.e. the system output is a linear combination of 
defined in terms of either the position or the velocity states of agents. Our novel method of computing these metrics in closed-form stems from exploiting the spectral properties of weighted graph Laplacians and output performance matrices. These newly derived closed-form expressions pave the way for our analytical findings. 
%The measures are also arbitrary provided that the respective consensus mode of position or velocity (the mode in which the states are in agreement) is unobservable from the output, which is an assumption that is necessary to obtain a finite output norm.
%Similar measures were considered in \cite{PaganiniMallada2017} and \cite{ } for undirected graphs whereas this result provides a generalization to directed graphs.
%
%\item
%
%We demonstrate that overall network performance is determined by an interaction between network topological characteristics (e.g. edge directionality and connectivity) and control strategy, 
By focusing on 
the subclass of directed graphs emitting diagonalizable Laplacians, we first show that the closed-form solutions for the performance metrics simplify for this family of graphs, allowing for the investigation of the interplay between the network topology and control strategy.
%(including graphs emitting normal Laplacians) 
%and proceed with further analysis.
%in order to gain further analytical insight on 
%how directed communication affects 
%system performance. 

%provide a way to directly characterize system performance in terms of the directed paths in the graph.
%the interplay between directed paths in the graph and system performance.  
%Secondly, for either position or velocity based performance measures of spatially invariant systems, we investigate the effect of the directionality of relative feedback.
%Ultimately we aim to understand the performance related advantages and shortcomings of directed communication by studying commonly occuring examples of such topologies.

%\vspace{-1pt}
%\setlength{\parskip}{0cm}

%\begin{enumerate}
%\item
%
%We study 
Using
systems with normal Laplacian matrices
%as their spectral properties provide a way to characterize the role of edge directionality in system performance.
as an example,
we present a comparative analysis between 
%a spatially invariant and directed communication structure and its `symmetrized' counterpart, which uses a bi-directional relative state measurement for any given uni-directional measurement in the former, 
directed graphs and their undirected counterparts represented by the Hermitian part of the graph Laplacian. 
%the class of systems whose directed graphs emit normal Laplacians and their `symmetrized' counterparts 
%in which a reverse path exists for any given directed path.
%defined by the Hermitian part of the original Laplacian.
In this setting, we show that directed graphs and their undirected counterparts provide identical performance for single integrator networks. In the case of double-integrator networks, 
%however,
we demonstrate that 
the presence of observable Laplacian eigenvalues with nonzero imaginary part (i.e. \defaultColor{the} observability of modes associated with directed paths) can significantly degrade both position and velocity based performance compared to the undirected topology. Nevertheless, this degradation can be eliminated for velocity-based metrics using absolute position feedback. On the other hand, for the case of position-based metrics a proper combination of relative position and velocity feedback can, not only 
%alleviate
\defaultColor{mitigate} 
this degradation, but also 
%improve beyond 
%
lead to improved performance over
\defaultColor{systems with} the undirected topology. 
%For measures that are position based, directed communication can yield better performance with well-tuned control gains. 

%\vspace{-12pt}
%\setlength{\parskip}{0cm}

%\item
We 
%also 
\defaultColor{then}
investigate the role of the degree of connectivity in system performance.   
%The first one uses distributed state measurements with respect to neighbors (a circulant graph topology), while the second using a leader-follower structure (an imploding star graph topology).
We first focus on the class of systems that we term
%has what we call 
%distributed relative feedback, 
$\omega$-nearest neighbor networks,
which have a 
%spatially invariant 
cyclic
and directed communication structure with each agent admitting uniformly weighted uni-directional state measurements 
%relative to a prescribed number of its 
from $\omega$
consecutive neighbors. 
%For general quadratic measures, 
For the special case of the metric quantifying the aggregate state deviation from the average,
we show that performance does not monotonically improve by increasing 
$\omega$.
%the number of 
%uni-directional links and the range of communication between agents. 
\iffalse
\blue{
%consecutive 
%
nearest neighbors
with which each agent communicates through uni-directional links.}
\fi
%in distributed relative feedback.
%For example, if each agent communicates with all of its neighbors, the feedback is called fully distributed. 
We also 
%study
show an equivalence between 
%two fundamentally different relative feedback schemes. 
%another class of systems that 
%has what 
%we call 
%centralized relative feedback, 
uniformly weighted, directed
\mbox{all-to-one} (imploding star) and all-to-all  networks
%where uni-directional state measurements are uniformly weighted and relative to a single designated `leader' that does not receive any relative feedback.
\iffalse
and show, \fi 
%In all cases, the total communication weight is chosen to be equal to the respective network size so as to provide a normalization.  
%This is due to the fact that removing one link from each bi-directional communication pair in the fully distributed feedback, i.e. using the maximum number of uni-directional state measurements can perform identically to, or even better than the fully distributed feedback. 
for the same performance metric. %quantifying the aggregate state deviation from the average, 
%we demonstrate 
\iffalse
an equivalence between %centralized
them. 
\fi
%\mbox{all-to-one}
%and
%two classes of systems with uni-directional feedback interconnections.
%fully distributed 
%\mbox{all-to-all}
%(each agent communicates with all of its neighbors) 
%relative feedback interconnections.
%networks. 
%we also show that the fully distributed and the centralized relative feedback schemes provide the same level of performance. 
%These somewhat surprising results suggest that 
%
%\end{enumerate}
%\end{enumerate}

%\vspace{-2pt}
%\setlength{\parskip}{0cm}

The remainder of the paper is organized as follows. In Section \ref{sysModel}, the system models and the performance metrics are introduced. In Section \ref{BlockDiag}, 
%a procedure for the block-diagonalization of the closed-loop dynamics is presented and the stability of the input-output system is discussed. 
we block-diagonalize the closed-loop dynamics and discuss the stability of the input-output system, facilitating the analysis in the following sections.
In Section \ref{ArbitraryL}, the closed-form solution for the general class of output $\mathcal{L}_2$ norm based performance metrics is provided for both single and double-integrator networks over arbitrary directed graphs that have at least one globally reachable node. In Section \ref{diagL}, we show that the performance computation simplifies for the case of the diagonalizable weighted graph Laplacian matrices. In Section \ref{NormalL}, the role of communication directionality is studied through systems with normal graph Laplacians. 
%are studied.
%in order to gain further insights on the effect of directed communication on performance.
In Section \ref{CentDistFeed}, 
%centralized and distributed relative feedback schemes 
\mbox{all-to-one} and $\omega$-nearest neighbor networks
are analyzed. Section \ref{conclusion} concludes the paper.

%\IEEEpubidadjcol

%\setlength{\parskip}{0cm}
\vspace{-4pt}

%*** PROBLEM DESCRIPTION ***
\section{System Models and Performance Metrics} \label{sysModel}

\subsection{Single and Double-Integrator Networks}

Consider $n$ dynamical systems that communicate over a weighted digraph $\mathcal{G} = \{\mathcal{N}, \mathcal{E}, \mathcal{W}\}$ that have at least one globally reachable node. Here, $\mathcal{N} = \{1,...,n\}$ is the set of nodes and $\mathcal{E} = \{(i,j) \mid i,j \in \mathcal{N}, \, i \neq j\} $ is the set of edges with weights $\mathcal{W} = \{ b_{ij} > 0 \mid (i,j) \in \mathcal{E} \}$. 
%The edge weight associated with $(i,j) \in \mathcal{E}$ is denoted by $b_{ij} > 0$, and 
In the following $b_{ij} = 0$ if and only if $(i,j) \notin \mathcal{E}$.

We consider two types of nodal dynamics.  Single integrator systems of the form
\begin{equation} \label{firstorder} \nonumber
\dot{x}_i = -\sum_{j=1}^{n}b_{ij}(x_i - x_j) + w_i,
\end{equation}
at each $i \in \mathcal{N}$, where $w_i$ denotes the disturbance to the $i^{th}$ agent. This results in the well-known consensus network
\begin{equation} \label{firstmatrix}
\bold{\dot{x}} = -L \bold{x} + \bold{w},
\end{equation}
where $L$ denotes the weighted graph Laplacian matrix given by $[L]_{ii} = \sum_{j=1}^{n}b_{ij}$, and $[L]_{ij} = -b_{ij}$ if $i \neq j$, $\forall i, j \in \mathcal{N}$.
The second type of system is governed by double-integrator dynamics of the form
\begin{equation} \label{secondorder} \nonumber
\ddot{x}_i + k_d \dot{x}_i + k_p x_i = -u_i + w_i,
\end{equation}
where $u_i =  \gamma_p \sum_{j=1}^{n}b_{ij}(x_i - x_j) + \gamma_d \sum_{j=1}^{n}b_{ij}(\dot{x}_i - \dot{x}_j)$ $\forall i \in \mathcal{N}$. Here, $k_p, k_d, \gamma_p, \gamma_d \geq 0$, and $w_i$ denotes the disturbance to the $i^{th}$ system. Defining $\bold{v} := \bold{\dot{x}}$, 
%the network of 
%\blue{
the
double-integrator
dynamics can be expressed
%is given 
in matrix form 
%by
as
%}
%
\begin{equation} \label{secondmatrix}
\begin{bmatrix}
\bold{\dot{x}} \\ \bold{\dot{v}}
\end{bmatrix} = 
\begin{bmatrix}
0 & I \\ 
-k_p I - \gamma_p L & -k_d I - \gamma_d L
\end{bmatrix}
\begin{bmatrix}
\bold{x} \\ \bold{v}
\end{bmatrix} + 
\begin{bmatrix}
0 \\ I
\end{bmatrix} \bold{w}.
\end{equation}
%
%where $\bold{v} := \bold{\dot{x}}$.

%Next, \eqref{firstmatrix} and \eqref{secondmatrix} will be represented in frequency domain, which allows for a block-diagonalization procedure such that the closed-loop dynamics can be decoupled from the eigenvectors of the weighted graph Laplacian.

%{\color{blue}
A necessary condition for \eqref{secondmatrix} to reach consensus without disturbance ($\bold{w} = 0$) is that at least one of $(k_p,\gamma_p)$ and at least one of $(k_d,\gamma_d)$ are non-zero \cite[Lemma 3]{OralMalladaGayme2017}. 
%Therefore we impose the following.
\defaultColor{To ensure that this condition is met, 
we impose 
the following 
%standing 
assumption
throughout the paper.}
\begin{assum} \label{assum_gain}
System \eqref{secondmatrix} has feedback in both state variables (position and velocity), i.e. at least one of $(k_p,\gamma_p)$ and at least one of $(k_d,\gamma_d)$ are non-zero.
\end{assum}
%
%Next, we introduce the performance metrics. 
%that are going to be studied.

%Subsection 2
\subsection{Performance Metrics}

Performance metrics 
%\blue{
that are
%} 
quadratic in the state 
%\blue{
variables
%} 
are widely used in control synthesis problems, especially paired with $\mathcal{H}_2$ or $\mathcal{H}_\infty$ criteria. In this work 
%\blue{
we focus
%} 
on the analysis of such metrics through a more general output norm based approach in order to gain 
%\blue{
insight into
%} 
how directed communication 
%\blue{
affects
%} 
performance.

%For the single-integrator network \eqref{firstmatrix}, 
%\blue{
For $C \in \mathbb{R}^{q \times n}$, the performance output 
\begin{equation} \label{posOutput}
\bold{y} = C \bold {x}
\end{equation}
will be used to quantify the performance 
%metrics 
of the single-integrator network \eqref{firstmatrix} and the double-integrator network \eqref{secondmatrix} for metrics
related to the position state $\bold{x}$. 
%where $C \in \mathbb{R}^{q \times n}$. 
For the double-integrator network \eqref{secondmatrix}, 
%performance metrics related to both position and velocity are relevant. 
%Therefore, for systems of this type, 
the performance output
\begin{equation} \label{velOutput}
\bold{y} = C \bold {v},
\end{equation}
which quantifies performance metrics related to the velocity state $\bold{v}$, will also be considered.
%}
%in addition to \eqref{posOutput}.}

We are interested in performance metrics of the form
\begin{equation} \label{perf_x}
P = \lVert \bold{y} \rVert_{\mathcal{L}_2}^2 = \int_{0}^{\infty} \bold{y}(t)^* \bold{y}(t) dt
\end{equation}
%
%and
%
%\begin{equation} \label{perf_v}
%P_{\bold{v}} = \lVert \bold{y_v} \rVert_{\mathcal{L}_2}^2 = \int_{0}^{\infty} \bold{y_v}(t)^*\bold{y_v}(t) dt
%\end{equation}
%
for an impulse input 
\begin{equation} \label{distDelta}
\bold{w}(t) = \bold{w}_0 \delta (t)
\end{equation}
with an arbitrary direction vector $\bold{w}_0 \in \mathbb{R}^{n}$. Similar metrics appear in \cite{PaganiniMallada2017} for networks over undirected graphs. 
%
%Denoting the impulse response function from $\bold{w}(t)$ to $\bold{x}(t)$ by $T_{\bold{x}\bold{w}} (t)$, and the impulse response function from $\bold{w}(t)$ to $\bold{v}(t)$ by $T_{\bold{v}\bold{w}} (t)$; we define
Denoting the impulse response function from $\bold{w}(t)$ to $\bold{y}(t)$ by $T (t)$, the performance output can be written as
%we can write
%
%\begin{equation} \label{Hfunc}
%H_{\bold{x}}(t) := C_{\bold{x}} T_{\bold{x}\bold{w}}(t) 
%\ \text{and} \
%H_{\bold{v}}(t) := C_{\bold{v}} T_{\bold{v}\bold{w}}(t).
%\end{equation} 
%

%Using the definitions in \eqref{Hfunc} along with \eqref{posOutput} and \eqref{velOutput} yields 
%
\begin{equation} \label{posOut2}
\bold{y} (t) =
\int_{0}^{t} T (t-\tau) \bold{w}(\tau) d \tau.
\end{equation}
%
%and
%
%\begin{equation} \label{velOut2}
%\bold{y_v} (t) =
%\int_{0}^{t} H_{\bold{v}}(t-\tau) \bold{w}(\tau) d \tau.
%\end{equation}
%
Substitution of \eqref{distDelta} and \eqref{posOut2} into \eqref{perf_x} gives
\begin{equation} \label{perf_x_2}
P = \int_{0}^{\infty} \bold{w}_0^* T (t)^* T (t) \bold{w}_0 dt.
\end{equation}
%
%and
%
%\begin{equation} \label{perf_v_2}
%P_{\bold{v}} = \int_{0}^{\infty} \bold{w}_0^* H_{\bold{v}}(t)^* H_{\bold{v}}(t) \bold{w}_0 dt.
%\end{equation}
%
Therefore, \eqref{perf_x_2} is finite if and only if 
$T (t)$ 
%and
%$H_{\bold{v}}(t)$ in \eqref{Hfunc}
is input-output (IO) stable. We will later discuss conditions that guarantee the IO stability of $T (t)$. 

%arxiv version
%\iffalse
We now show that 
%\fi
%
%short version
\iffalse
%\blue{
As stated next,
%}
\fi
%
for a special case of the impulse input \eqref{distDelta}, the performance metric \eqref{perf_x_2} can be computed using the $\mathcal{H}_2$ norm of $T(t)$.
%\blue{
%\red{
%arxiv version
%\iffalse
Although this connection is standard in the literature \cite{BamiehGayme2013}, for completeness we provide a short proof below. 
%\fi
%
%short version
\iffalse
This relationship, which is standard in the literature \cite{BamiehGayme2013}, will be used in the upcoming sections. 
\fi
%
%arxiv version
%\iffalse
This relationship will be used in the upcoming sections.
%\fi
%}
%

\begin{prop}
\label{Prop_H2L2}
%\blue{
%\red{
Consider a general MIMO system $G(t)$ from $\bold{w}$ to $\bold{y}$ 
%denoted by $G(t)$
. Assume
%and 
a random impulse input \eqref{distDelta}
%$\bold{w}(t) ~=~\delta (t) \mathrm{w}_0 $ 
with $E \left[\bold{w}_0 \bold{w}_0^* \right] = I$ and zero initial condition. Then
%$$
$
\lVert G \rVert_{\mathcal{H}_2}^2
=
E
\left[  
\lVert \bold{y} \rVert_{\mathcal{L}_2}^2
\right].
$ 
%$$ 
%}
\end{prop}
%
%arxiv version
%\iffalse
\begin{proof}
\defaultColor{
%Without loss of generality 
Assuming 
%that  $\bold{x}(0) = 0$,
zero initial condition,}
the output is given by
$
\bold{y}(t) = C e^{At} B \bold{w}_0.
$
Then
\begin{align*}
E 
\left[  
\lVert \bold{y}(t) \rVert_{\mathcal{L}_2}^2
\right] \
&= \ E
\left[ 
\mathrm{tr} \!
\int_0^{\infty} \! \!
C e^{At} B \bold{w}_0
\bold{w}_0^* B^* e^{A^* t} C^*
dt
\right]
\\
\ &= \ 
\mathrm{tr} \!
\int_0^{\infty} \! \!
C e^{At} B B^* e^{A^* t} C^*
dt 
\ = \ \lVert G \rVert_{\mathcal{H}_2}^2.
\ 
%\QEDhere
\! \! \! \! \!
\mbox{\qedhere}
%\tag*{\qedhere}
\end{align*}
%
%\iffalse
%\blue{
%We refer the reader to \cite{OralMalladaGayme2019_arxiv} for a proof.}
%\fi
\end{proof}
%
%\fi
%

%The performance measure \eqref{perf_x} is related to the $\mathcal{H}_2$ norm through one of the standard interpretations [references]. If we denote $P^{(i)}$ as the performance measure corresponding to $\bold{w}_0 = \bold{e}_i$, where $\bold{e}_i \in \mathbb{R}^n$ has a $1$ at the $i^{th}$ entry and zeros elsewhere, then 
%
%\begin{equation} \nonumber
%\lVert T \rVert_{\mathcal{H}_2}^{2} = \sum_{i=1}^{n} P^{(i)}.
%\end{equation}
%
%and a similar argument is valid for $H_{\bold{v}}$.

%Now we proceed to study the performance measure \eqref{perf_x_2} for systems defined over arbitrary digraphs with a globally reachable node.
%

%SECTION 3
\section{Block-diagonalization of the Closed-loop Dynamics} \label{BlockDiag}
\begin{figure}[t!]
\centering
\includegraphics[width=2.5in]{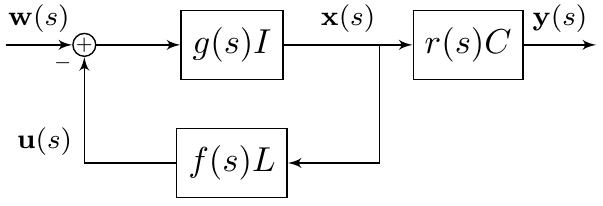}
\vspace{-0.3cm}
\caption{Block diagram of the closed-loop system $T (s)$ from the disturbance input $\bold{w} (s)$ to the performance output $\bold{y} (s)$ and the closed-loop system $H_{\bold{x}\bold{w}}(s)$ from $\bold{w} (s)$ to the position state $\bold{x} (s)$. The performance output $\bold{y} (s)$ is given by \eqref{posOutput} if $r(s) = 1$ and by \eqref{velOutput} if $r(s) = s$. }
\label{block1}
\vspace{-0.5cm}
\end{figure}

In this section, 
%we begin with a general frequency domain analysis that encompasses the single and double-integrator networks \eqref{firstmatrix} and \eqref{secondmatrix}. The approach taken here is based on the analysis in [reference]. In the following, 
we express the dynamics given in \eqref{firstmatrix} and \eqref{secondmatrix} in the frequency domain using an approach based on \cite{PaganiniMallada2017}. The framework, 
%\blue{
denoted
%} 
in 
%we model a linear networked system 
Figure \ref{block1}, describes identical systems $g(s)$ receiving feedback that depends on an arbitrary transfer function $f(s)$ and the weighted graph Laplacian $L$ emitted by the network interconnection. 
%\blue{
Assuming 
%without loss of generality 
that $\bold{x}(0) = \bold{v}(0) = 0$
(we consider perturbations to the equilibrium),
%}, 
the closed-loop system from the input $\bold{w}$ to the position state $\bold{x}$ is given by

%Taking the Laplace transform of the closed-loop system \eqref{secondmatrix} and assuming without loss of generality that $\bold{x}(0) = \bold{v}(0) = 0$ gives
%
%\begin{equation} \label{laplace1}
%\left[ (m s^2 + k_d s + k_p) I + (\gamma_p + s \gamma_d) L \right] \bold{x}(s) = \bold{w}(s).
%\end{equation}
%
%
\begin{equation*} \label{laplace1}
\left[ (g(s)^{-1} I + f(s) L \right] \bold{x}(s) = \bold{w}(s),
\end{equation*}
%
%We denote the uniform open-loop nodal dynamics by
%
%\begin{equation} \nonumber
%g(s) = \frac{1}{m s^2 + k_d s + k_p}
%\end{equation}
%
which leads to
\begin{align} \label{laplace2}
%\blue{
\bold{x}(s) = \left[ (I + g(s) f(s) L \right]^{-1} g(s) \bold{w}(s) 
%\nonumber \\ 
=: H_{\bold{x}\bold{w}}(s) \bold{w}(s), 
%}
\end{align}
where $H_{\bold{x}\bold{w}}(s)$ denotes the transfer function from the input $\bold{w}$ to the position state $\bold{x}$. 
%Equation \eqref{laplace2} can be equivalently represented by the block diagram given in Figure \ref{block1}.

$L$ can be decomposed as $L = R J R^{-1}$, where $R \in \mathbb{C}^{n \times n}$ is invertible and $J \in \mathbb{C}^{n \times n}$ is in Jordan Canonical Form (JCF). This decomposition transforms \eqref{laplace2} into
\begin{equation} \label{laplace3}
\bold{x}(s) = R \left[ (I + g(s) f(s) J \right]^{-1} g(s) R^{-1} \bold{w}(s), \nonumber
\end{equation}
%
%which is
as 
shown by the block diagram in Figure \ref{block2}. 
Defining
%Letting 
$\bold{\tilde{x}} := R^{-1} \bold{x}$ and $\bold{\tilde{w}} := R^{-1} \bold{w}$, the transfer function from $\bold{\tilde{w}}$ to $\bold{\tilde{x}}$ is
\begin{equation} \label{T_tilde_pos}
H_{\bold{\tilde{x}} \bold{\tilde{w}}} (s) = \left[ (I + g(s) f(s) J \right]^{-1} g(s),
\end{equation}
where the following relationship holds
\begin{equation} \label{TxwTrans}
H_{\bold{x} \bold{w}} = R H_{\bold{\tilde{x}} \bold{\tilde{w}}} R^{-1}.
\end{equation}
\begin{figure}[t!]
\centering
\includegraphics[width=3.3in]{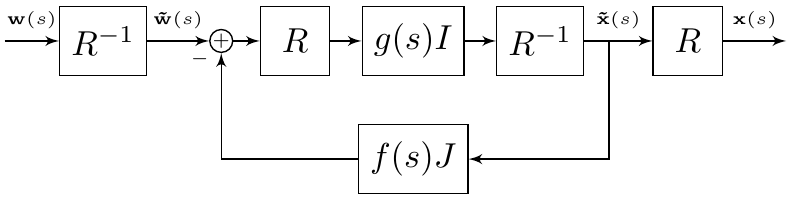}
\vspace{-0.3cm}
\caption{Application of a change of basis given by the Jordan decomposition $L = R J R^{-1}$ to the closed-loop system $H_{\bold{x}\bold{w}}(s)$. The feedback loop gives the closed-loop system $H_{\bold{\tilde{x}}\bold{\tilde{w}}}(s)$.}
\label{block2}
\vspace{-0.5cm}
\end{figure}
$J$ is composed of Jordan blocks $J_k$ associated with the eigenvalues $\lambda_k \in \mathbb{C}$ of $L$ for $k = 1, \dots, m$:
\begin{equation} \label{jordanL}
\iffalse
J =
\begin{bmatrix} 
J_1 & & \\
& \ddots & \\
& & J_m
\end{bmatrix},
\fi
%\blue{
J = \bdiag{\left( J_k \right)_{1 \leq k \leq m}},
%}
\end{equation}
where $J_k \in \mathbb{C}^{n_k \times n_k}$ and $\sum_{k=1}^{m} n_k = n$. Since $L$ is a Laplacian matrix, $L \bold{1} = \bold{0}$ with $\bold{1}$ denoting the vector of all ones therefore $J_1 = \lambda_1 = 0$. Also 
%$\lambda_k \neq 0$ 
%\blue{
$\R \left[ \lambda_k \right] > 0$
for $k = 2, \dots, m$ due to the fact that $\mathcal{G}$ has a globally reachable node 
%\cite{bullo2016lectures}. 
\cite[Theorem 7.4]{bullo2016lectures}.
So 
\eqref{T_tilde_pos} 
can be written as
%}
%
\begin{equation} \label{block_diag}
\iffalse
H_{\bold{\tilde{x}} \bold{\tilde{w}}} (s) = 
\begin{bmatrix}
H_{\bold{\tilde{x}}_1 \bold{\tilde{w}}_1} (s) & & \\
& \ddots & \\
& & H_{\bold{\tilde{x}}_m \bold{\tilde{w}}_m} (s)
\end{bmatrix},
\fi
%
%\blue{
H_{\bold{\tilde{x}} \bold{\tilde{w}}} (s) = \bdiag{\left( H_{\bold{\tilde{x}}_k \bold{\tilde{w}}_k} (s) \right)_{1 \leq k \leq m}}, 
%}
\end{equation}
where 
\begin{equation} \label{T_tilde_i}
H_{\bold{\tilde{x}}_k \bold{\tilde{w}}_k}(s) = \left[ (I + g(s) f(s) J_k \right]^{-1} g(s).
\end{equation}
%for $k = 1, \dots, m$. 
Here, the vectors $\bold{\tilde{x}}_k = [ \tilde{x}_{d_k + 1}, \dots, \tilde{x}_{d_k+n_k}]^\intercal$ and $\bold{\tilde{w}}_k ~=~ [ \tilde{w}_{d_k +1}, \dots, \tilde{w}_{d_k+n_k} ]^\intercal$ respectively denote the position state and the input 
%\blue{
to
%} 
the associated subsystem, with $d_1 = 0$ and $d_k = \sum_{i=1}^{k-1} n_i$ for $k=2,\dots,m$. An equivalent representation of the transfer function in
%\eqref{block_diag} and 
\eqref{T_tilde_i} is given by the block diagram in Figure \ref{fig3}.
\begin{figure}[b] 
\vspace{-0.5cm}
\centering
\includegraphics[width=1.9in]{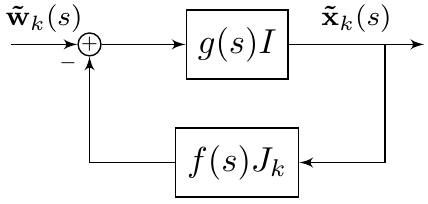}
%\caption{}
%\label{block3}
%\subfloat{
%\includegraphics[width=2.3in]{block4}}
%\caption{}
%\label{block4}
\vspace{-0.3cm}
\caption{Block diagram of each subsystem $H_{\bold{\tilde{x}}_k \bold{\tilde{w}}_k}$ for $k=1,\dots,m$.}
\label{fig3}
%\vspace{-0.5cm}
\end{figure}
The following lemma describes the form of 
%the closed-loop transfer function blocks given by \eqref{T_tilde_i}.
%these 
the
%\blue{
transfer function in \eqref{T_tilde_i}
%} 
which will be 
%\blue{
used
%} 
to compute the performance metric \eqref{perf_x_2} 
%later on.
%\blue{
in what follows.
%}
%Lemma1
\begin{lemma} \label{lemma1}
$H_{\bold{\tilde{x}}_k \bold{\tilde{w}}_k}(s)$ in \eqref{T_tilde_i} is an upper triangular Toeplitz matrix given by
\begin{equation} \nonumber
H_{\bold{\tilde{x}}_k \bold{\tilde{w}}_k}(s) = \frac{1}{ f(s) }
\begin{bmatrix}
h_k(s) & \dots & (-1)^{n_k - 1} h_k(s)^{n_k} \\
& \ddots & \vdots \\
& & h_k(s)
\end{bmatrix},
\end{equation}
where 
%
%\begin{equation} \nonumber
$
h_k(s) = \frac{g(s) f(s)}{1 + \lambda_k g(s) f(s) }
$.
%\end{equation}
%
\end{lemma}
\begin{proof}
Using \eqref{T_tilde_i} and the definition of $J_k$
%
%\begin{equation} \nonumber
%H_{\bold{\tilde{x}}_k \bold{\tilde{w}}_k}(s) = g(s)
%\begin{bmatrix}
%1+\hat{h}(s) \lambda_k & \hat{h}(s) & & \\
%& \ddots & \ddots & \\
%& &  \ddots & \hat{h}(s) \\
%& & &  1+\hat{h}(s) \lambda_k
%\end{bmatrix}^{-1},
%\end{equation}
%
%
\begin{equation} \nonumber
H_{\bold{\tilde{x}}_k \bold{\tilde{w}}_k}(s) =
\begin{bmatrix}
\frac{1+ \lambda_k g(s) f(s) } {g(s)} & f(s) & & \\
& \ddots & \ddots & \\
& &  \ddots & f(s) \\
& & &  \frac{1+ \lambda_k g(s) f(s)}{g(s)}
\end{bmatrix}^{-1},
\end{equation}
where factoring out $g(s) f(s)$ gives
\begin{equation} \label{T_tilde_prf}
H_{\bold{\tilde{x}}_k \bold{\tilde{w}}_k}(s) = \frac{1}{f(s)}
\begin{bmatrix}
h_k(s)^{-1} & 1 & & \\
& \ddots & \ddots & \\
& &  \ddots & 1 \\
& & &  h_k(s)^{-1}
\end{bmatrix}^{-1}.
\end{equation}
%
%The matrix in \eqref{T_tilde_prf} is in JCF. 
Using the inverse of 
%\blue{
the
%} 
JCF in \eqref{T_tilde_prf}
%which is inverted 
%[reference] 
%we get 
yields
the 
%desired 
result.
\end{proof}
%
%\blue{
\begin{rem}
The form of the closed-loop transfer function in Lemma \ref{lemma1} holds for arbitrary open-loop and feedback transfer functions $g(s)$ and $f(s)$, and therefore applies to a general class of networked dynamical systems. 
\end{rem}
%}
%In the following corollaries, 
We next 
%present the respective results for
%\blue{
apply Lemma \ref{lemma1} to
%} 
the special cases of the single and double-integrator networks.

%Proposition1
%\begin{prop} \label{prop1}
%\end{proof}
%

%Corollary 1
\begin{cor} \label{cor1}
Consider the single-integrator network \eqref{firstmatrix}. Then, $H_{\bold{\tilde{x}}_k \bold{\tilde{w}}_k}(s)$ in \eqref{T_tilde_i} is an upper triangular Toeplitz matrix 
%given by
%
\begin{equation} \nonumber
H_{\bold{\tilde{x}}_k \bold{\tilde{w}}_k}(s) =
\begin{bmatrix}
h_k(s) & \dots & (-1)^{n_k - 1} h_k(s)^{n_k} \\
& \ddots & \vdots \\
& & h_k(s)
\end{bmatrix},
\end{equation}
where 
%
%\begin{equation} \nonumber
$
h_k(s) = \frac{1}{s + \lambda_k}
$.
%\end{equation}
%
\end{cor}
\begin{proof}
Taking the Laplace transform of \eqref{firstmatrix}
leads to 
%it is observed that 
$g(s) = \frac{1}{s}$ and $f(s) = 1$. Evaluating the result of Lemma \ref{lemma1} at these values gives the desired result.
\end{proof}
%

%Corollary 2
\begin{cor} \label{cor2}
Consider the double-integrator network \eqref{secondmatrix}. Then, $H_{\bold{\tilde{x}}_k \bold{\tilde{w}}_k}(s)$ in \eqref{T_tilde_i} is an upper triangular Toeplitz matrix 
%given by
%
\begin{equation} \nonumber
H_{\bold{\tilde{x}}_k \bold{\tilde{w}}_k}(s) = \frac{1}{\gamma_p + s \gamma_d}
\begin{bmatrix}
h_k(s) & \dots & (-1)^{n_k - 1} h_k(s)^{n_k} \\
& \ddots & \vdots \\
& & h_k(s)
\end{bmatrix},
\end{equation}
where 
%
%\begin{equation} \nonumber
$
h_k(s) = \frac{\gamma_p + s \gamma_d}{s^2 + (k_d + \gamma_d \lambda_k) s + k_p + \gamma_p  \lambda_k}
$.
%\end{equation}

\end{cor}

\begin{proof}
Taking the Laplace transform of \eqref{secondmatrix}
leads to 
%it is observed that 
$g(s) ~=~  \frac{1}{s^2 + k_d s + k_p}$ and $f(s) ~=~ \gamma_p + s \gamma_d$. Evaluating the result of Lemma \ref{lemma1} at these values gives the desired result.
\end{proof}
%

%Lastly, we note that 
%one can obtain the relationship between 
The transfer function
%$H_{\bold{v}\bold{w}}(s)$
from the input $\bold{w}$ to
the velocity state $\bold{v}$ is given by 
$H_{\bold{v}\bold{w}}(s)
:=  s H_{\bold{x}\bold{w}}(s)
$
since
%setting $m = k_p = \gamma_d = 0$ and $k_d = \gamma_p = 1$ results in the block-diagonalization procedure for the consensus network. 
%Also, one can obtain the relationship between the velocity state $\bold{v}(s)$ and the input $\bold{w}(s)$ through
%
%\begin{equation} \nonumber
\mbox{
$
\bold{v}(s) = s \bold{x}(s) = s H_{\bold{x}\bold{w}}(s) \bold{w}(s). 
%=: H_{\bold{v}\bold{w}}(s) \bold{w}(s ).
$}
%\end{equation}
%
Therefore, the closed-loop transfer function 
%\blue{
$T (s)$ from the input $\bold{w}$ to the output $\bold{y}$
%}
can be written as
\begin{equation} \label{H_func}
T (s) = C r(s) H_{\bold{x}\bold{w}}(s),
\end{equation}
using the notation in Figure \ref{block1} 
%where we 
and specifying $r(s)$ such that
%
%\begin{equation} \label{H_cases}
\begin{numcases}
{T (s) =} 
C H_{\bold{x}\bold{w}}(s), & $r(s) = 1$ \label{H_cases_1} \\
C H_{\bold{v}\bold{w}}(s), & $r(s) = s$ \label{H_cases_2}.
\end{numcases}
%\end{equation}
% 
%
The cases \eqref{H_cases_1} and \eqref{H_cases_2} correspond to the 
%performance 
outputs \eqref{posOutput} and \eqref{velOutput}, respectively. 
%Before studying
%\blue{
We next provide necessary and sufficient conditions for  
the input-output stability of \eqref{H_cases_1} and \eqref{H_cases_2},
which ensure the finiteness of  
the performance metric \eqref{perf_x_2}. 
%is investigated next.
%}

%SUBSECTION IO STABILITY
\subsection{Input-Output Stability}

In this subsection we state necessary and sufficient conditions for the input-output stability of \eqref{H_cases_1} and \eqref{H_cases_2}. 
%using results from the literature.
The following assumption will be imposed throughout the paper to 
%establish conditions 
eliminate the 
%well-known 
unstable consensus mode of the Laplacian from the performance output.
%in order to guarantee input-output stability of \eqref{H_cases_1} and \eqref{H_cases_2}.

\begin{assum} \label{assum1}
The output matrix $C$ satisfies $C \bold{1}= \bold{0}$.
\end{assum}

%\blue{
First, we apply the change of basis in \eqref{TxwTrans} to the closed-loop system \eqref{H_func}.
%}
Since $L \bold{1} = \bold{0}$, we can apply the partitioning 
\begin{equation} \label{RandQ}
R = \begin{bmatrix} \alpha \bold{1} & \tilde{R} \end{bmatrix} \
\text{and} \
R^{-1} = \begin{bmatrix} \bold{q}_1 & \tilde{Q}^*\end{bmatrix}^*,
\end{equation} 
where $\alpha \in \mathbb{C}$, $\bold{q}_1^* \in \mathbb{C}^{1\times n} $ is the left eigenvector of $\lambda_1=0$, $\tilde{R} \in \mathbb{C}^{n \times n-1}$ and $\tilde{Q} \in \mathbb{C}^{n-1 \times n}$. Substituting \eqref{TxwTrans}, \eqref{block_diag} and \eqref{RandQ} 
%the definitions of $T$ in 
%$H_{\bold{x}\bold{w}}$ in 
%\eqref{TxwTrans} and 
%$H_{\bold{\tilde{x}} \bold{\tilde{w}} }$ in 
%\eqref{block_diag} 
into \eqref{H_func} we obtain
\begin{align} \nonumber
T (s) &= C
\left( \alpha r(s) H_{\bold{\tilde{x}}_1 \bold{\tilde{w}}_1} (s) \bold{1} \bold{q}_1^* +  
\tilde{R} \, \tilde{H} (s) \tilde{Q} \right) \\
%&= C_{\bold{x}} \tilde{R} \nonumber
%\left[ \begin{smallmatrix} 
%T_{\bold{\tilde{x}}_2 \bold{\tilde{w}}_2} & & \\
%& \ddots & \\
%& & T_{\bold{\tilde{x}}_m \bold{\tilde{w}}_m} 
%\end{smallmatrix} \right]
%\tilde{Q} \\
&= C \tilde{R} \, \tilde{H}(s) \tilde{Q} \label{H_x_hat},
\end{align}
%
%where
where 
%for $k = 2, \dots, m$ 
%
\iffalse
\begin{equation} \label{T_tilde}
\! 
\tilde{H}(s) 
\! = \! \!
\left[ \! \begin{smallmatrix} 
\tilde{H}_2 (s) & & \\
& \ddots & \\
& & \tilde{H}_m (s) \!
\end{smallmatrix} \right]
\! \! \!:= \!
r(s) \! \!
\left[ \begin{smallmatrix} 
H_{\bold{\tilde{x}}_2 \bold{\tilde{w}}_2} (s) & & \\
& \ddots & \\
& & H_{\bold{\tilde{x}}_m \bold{\tilde{w}}_m} (s)
\end{smallmatrix} \right] \! \!, \! \! \! \! \! \! \! \!
\end{equation} 
\fi
%
%\iffalse
\begin{equation} \label{T_tilde}
\! 
\tilde{H}(s) 
\! = \!
\bdiag{ ( \tilde{H}_k (s) ) } 
\! := \!
r(s)
\bdiag{ \left(
H_{\bold{\tilde{x}}_k \bold{\tilde{w}}_k} (s) \right) }, \! \! \! \! \! \! \! \!
\end{equation} 
%\fi
%\blue{
for $k = 2, \dots, m$
%}
and we have used Assumption \ref{assum1} and the fact that $H_{\bold{\tilde{x}}_1 \bold{\tilde{w}}_1} (s)$ is a scalar. 
We can partition $\tilde{R}$ in \eqref{RandQ} as 
\begin{equation} \label{Rtilde_part}
\tilde{R} = \begin{bmatrix} \tilde{R}_2 & \dots & \tilde{R}_m \end{bmatrix},
\end{equation}
which is in a form that conforms to \eqref{jordanL}. Then the columns of $\tilde{R}_k ~\in~ \mathbb{C}^{n \times n_k}$ are the right generalized eigenvectors associated with the Jordan block $J_k$ in \eqref{jordanL} for $k = 2, \dots, m$. This partitioning leads to the following useful definition.

\begin{defn} \label{obsv_indices}
The set of observable indices $\mathcal{N}_{obsv}$ is given by
\begin{equation} \label{obsv_ind_eqn}
\mathcal{N}_{obsv} =
\left\{
k \in \{ 2, \dots, m \}
\mid
C \tilde{R}_k \neq 0 
%C \boldsymbol{\theta}_k \neq 0
\right\}.
\end{equation}
\iffalse
\begin{equation*}
\mathcal{N}_{obsv} =
\left\{
k \in \{ 2, \dots, n \}
\mid
C \bold{r}_k \neq 0 
%C \boldsymbol{\theta}_k \neq 0
\right\},
\end{equation*}
%where
recalling that 
$\bold{r}_k$ 
%$\boldsymbol{\theta}_k$ 
denote the right generalized eigenvectors of 
%$M$ 
$L$ as defined in \eqref{RandQ}.
%in \eqref{matrix_M}.
\fi
\end{defn}

We now state the stability conditions. 
%The following proposition concerns 
We begin with
the system $T$ in 
%provides a necessary and sufficient condition for the input-output stability of 
\eqref{H_cases_1} for the single-integrator network \eqref{firstmatrix}.

%Prosposition 1
\begin{prop} \label{prop4}
Consider the single-integrator network \eqref{firstmatrix}. The system $T$ in \eqref{H_cases_1} is input-output stable if and only if Assumption \ref{assum1} holds 
%\blue{
\cite[Theorem 7.4]{bullo2016lectures}.
%}
\end{prop}
%
%\begin{proof} 
%\end{proof}
%

As 
%the next proposition shows 
we show next
for 
%when 
the double-integrator network \eqref{secondmatrix}, 
%is considered, 
%Assumption \ref{assum1} can also be imposed 
stability of the observable modes
%in order to derive a 
is necessary and sufficient for the input-output stability of the system $T$ given by \eqref{H_cases_1} or \eqref{H_cases_2}. For simplicity, we assume $L$ to be diagonalizable; the result can be extended by relaxing this assumption.
 
%Proposition 2
%{\color{blue}
%\red{
\begin{prop} \label{prop3}
Consider the double-integrator network \eqref{secondmatrix} and suppose that $L$ is diagonalizable and assumptions \ref{assum_gain} and \ref{assum1} hold.
The system $T$ given by \eqref{H_cases_1} or \eqref{H_cases_2} is input-output stable if and only if 
%the solutions to
\begin{equation} \label{char_eq_s}
s^2 + (k_d + \gamma_d \lambda_k) s + k_p + \gamma_p  \lambda_k = 0
\end{equation}
has solutions that satisfy
%are such that 
$\R (s) < 0$ for all
%$k = 2, \dots, m$.
$k \in \mathcal{N}_{obsv}$.
\end{prop}
\begin{proof} 
\iffalse
 
\fi
%
%arxiv version
%\iffalse
Using the block diagram in Figure \ref{fig3} and the fact that $J_k = \lambda_k$ leads to the following realization for $H_{\bold{\tilde{x}}_k \bold{\tilde{w}}_k}$
\begin{align} 
\begin{bmatrix}
\bold{\dot{\tilde{x}}}_k \\ \bold{\dot{\tilde{v}}}_k
\end{bmatrix} &= 
\underbrace{
\begin{bmatrix}
0 & 1 \\ 
-k_p - \gamma_p \lambda_k & -k_d - \gamma_d \lambda_k
\end{bmatrix}}_{\Lambda_k}
\begin{bmatrix}
\bold{\tilde{x}}_k \\ \bold{\tilde{v}}_k
\end{bmatrix}
+ 
%\underbrace{
\begin{bmatrix}
0 \\ 1
\end{bmatrix}
%}_{B_k}
 \bold{\tilde{w}}_k \nonumber \\
\bold{\tilde{y}}_k &=
%\underbrace{
\begin{bmatrix}
1 & 0 
\end{bmatrix} 
%}_{C_k}
\begin{bmatrix}
\bold{\tilde{x}}_k \\ \bold{\tilde{v}}_k
\end{bmatrix} =
 \bold{\tilde{x}}_k.
 \label{Hk_real}
\end{align}
Since $L$ is diagonalizable, the partitioning of $\tilde{R}$ in \eqref{Rtilde_part} becomes $\tilde{R} = \begin{bmatrix}
\bold{r}_2 & \dots & \bold{r}_n
\end{bmatrix}$. Using the block-diagonal form of $\tilde{H}(s)$ in \eqref{T_tilde} and the conformal partitioning $\tilde{Q} ~=~ \begin{bmatrix} \bold{q}_2 & \dots & \bold{q}_n \end{bmatrix}^*$, \eqref{H_x_hat} can be expressed in time-domain as
$$
T(t) \ = \ 
C \sum_{k=2}^{n} \bold{r}_k \tilde{H}_k (t) \bold{q}_k^* \
= C \! \sum_{k \in \mathcal{N}_{obsv}} \! \! \! \! \bold{r}_k \tilde{H}_k (t) \bold{q}_k^*.
$$
For \eqref{H_cases_1}, we can use \eqref{T_tilde} and the realization for $H_{\bold{\tilde{x}}_k \bold{\tilde{w}}_k}$ in \eqref{Hk_real} to re-write the equation above as
\begin{align*}
T(t) \ =&
%\ C \! 
\sum_{k \in \mathcal{N}_{obsv}} \! \! \! \! \begin{bmatrix}
C \bold{r}_k  & 0
\end{bmatrix}
e^{\Lambda_k t} 
\begin{bmatrix}
0 \\
\bold{q}_k^*
\end{bmatrix},
\end{align*}
which has a realization 
\begin{equation} \label{T_real_pos}
T(t) =
\left(
\begin{array}{c|c}
\left[
\begin{smallmatrix}
\ddots & & \\
& \Lambda_k & \\
& & \ddots
\end{smallmatrix} \right]
& 
\left[ \begin{smallmatrix}
\vdots \\
\left[
\begin{smallmatrix}
0 \\ \bold{q}_k^*
\end{smallmatrix}
\right]
\\ \vdots
\end{smallmatrix} \right]
\\ \hline
\left[
\begin{smallmatrix}
\dots & 
\left[
\begin{smallmatrix}
C \bold{r}_k  & 0
\end{smallmatrix}
\right]
& \dots
\end{smallmatrix} \right]
& 0
\end{array}
\right), 
\quad k \in \mathcal{N}_{obsv}.
\end{equation}
The associated observability matrix is given by
\begin{equation} \label{pos_obsv}
\mathcal{O} = 
\left[
\begin{array}{c}
\left[
\begin{smallmatrix}
\dots & 
\left[
\begin{smallmatrix}
C \bold{r}_k  & 0
\end{smallmatrix}
\right]
& \dots
\end{smallmatrix} \right]
\\
\left[
\begin{smallmatrix}
\dots & 
\left[
\begin{smallmatrix}
C \bold{r}_k  & 0
\end{smallmatrix}
\right] \Lambda_k
& \dots
\end{smallmatrix} \right]
\\
\vdots
\\
\left[
\begin{smallmatrix}
\dots & 
\left[
\begin{smallmatrix}
C \bold{r}_k  & 0
\end{smallmatrix}
\right] \Lambda_k^{2 |\mathcal{N}_{obsv}|-1}
& \dots
\end{smallmatrix} \right]
\end{array} 
\right],
\end{equation}
where $k \in \mathcal{N}_{obsv}$ and $|\mathcal{N}_{obsv}|$ denotes the cardinality of $\mathcal{N}_{obsv}$. Due to the form of $\Lambda_k$ in \eqref{Hk_real}, we can see that \mbox{$
\begin{bmatrix}
C \bold{r}_k  & 0
\end{bmatrix}
\Lambda_k =
\begin{bmatrix}
0 & C \bold{r}_k
\end{bmatrix}$}. Then the first two block-rows of \eqref{pos_obsv} imply that $\mathcal{O}$ is full rank if the vectors $C \bold{r}_k$ are linearly independent for $k \in \mathcal{N}_{obsv}$. For a proof by contradiction, assume that $C \bold{r}_k$ are linearly dependent, i.e. $\sum_{k \in \mathcal{N}_{obsv}}\alpha_k C \bold{r}_k = 0$ where $\alpha_k$ is non-zero for some $k$. This implies that $\sum_{k \in \mathcal{N}_{obsv}}\alpha_k \bold{r}_k \in \ker \{C\}$, which can be expressed as a linear combination of the vectors that span $\ker \{C\}$. Then
$$
\sum_{k \in \mathcal{N}_{obsv}}\alpha_k \bold{r}_k = -
\! \!
%\sum_{k \notin \mathcal{N}_{obsv}}
\sum_{k \in \{ 1, \dots, n  \} \setminus \mathcal{N}_{obsv}}
\! \!
\alpha_k \bold{r}_k
\Rightarrow
\sum_{k=1}^{n} \alpha_k \bold{r}_k = 0,
$$
which would contradict the fact that $R$ is invertible. Therefore, $\mathcal{O}$ in \eqref{pos_obsv} is full rank, so the realization in \eqref{T_real_pos} is observable. By a similar argument we can prove the controllability, hence the minimality of \eqref{T_real_pos}. Therefore, the poles of $T(s)$ in \eqref{H_cases_1} are given precisely by the eigenvalues of the system matrix in \eqref{T_real_pos}, which are determined by \eqref{char_eq_s}. Then $T(s)$ is input-output stable if and only if its poles are on the open left half-plane.

We now repeat the argument for \eqref{H_cases_2} which is given by 
\begin{align*}
T(t) \ =&
%\ C \! 
\sum_{k \in \mathcal{N}_{obsv}} \! \! \! \! \begin{bmatrix}
0 & C \bold{r}_k
\end{bmatrix}
e^{\Lambda_k t} 
\begin{bmatrix}
0 \\
\bold{q}_k^*
\end{bmatrix}
\end{align*}
in time-domain with a realization 
\begin{equation} \label{T_real_vel}
T(t) =
\left(
\begin{array}{c|c}
\left[
\begin{smallmatrix}
\ddots & & \\
& \Lambda_k & \\
& & \ddots
\end{smallmatrix} \right]
& 
\left[ \begin{smallmatrix}
\vdots \\
\left[
\begin{smallmatrix}
0 \\ \bold{q}_k^*
\end{smallmatrix}
\right]
\\ \vdots
\end{smallmatrix} \right]
\\ \hline
\left[
\begin{smallmatrix}
\dots & 
\left[
\begin{smallmatrix}
0 & C \bold{r}_k
\end{smallmatrix}
\right]
& \dots
\end{smallmatrix} \right]
& 0
\end{array}
\right), 
\quad k \in \mathcal{N}_{obsv}.
\end{equation}
The associated observability matrix is given by
\begin{equation} \label{vel_obsv}
\mathcal{O} = 
\left[
\begin{array}{c}
\left[
\begin{smallmatrix}
\dots & 
\left[
\begin{smallmatrix}
0 & C \bold{r}_k
\end{smallmatrix}
\right]
& \dots
\end{smallmatrix} \right]
\\
\left[
\begin{smallmatrix}
\dots & 
\left[
\begin{smallmatrix}
0 & C \bold{r}_k
\end{smallmatrix}
\right] \Lambda_k
& \dots
\end{smallmatrix} \right]
\\
\vdots
\\
\left[
\begin{smallmatrix}
\dots & 
\left[
\begin{smallmatrix}
0 & C \bold{r}_k
\end{smallmatrix}
\right] \Lambda_k^{2 |\mathcal{N}_{obsv}|-1}
& \dots
\end{smallmatrix} \right]
\end{array} 
\right],
\end{equation}
where $k \in \mathcal{N}_{obsv}$. Since $$
\begin{bmatrix}
0 & C \bold{r}_k
\end{bmatrix}
\Lambda_k =
C \bold{r}_k
\begin{bmatrix}
-k_p - \gamma_p \lambda_k & -k_d - \gamma_d \lambda_k
\end{bmatrix},$$
and assumption \ref{assum_gain} holds, \eqref{vel_obsv} is full rank and \eqref{T_real_vel} is observable, hence minimal. Therefore, the poles of $T(s)$ in \eqref{H_cases_2} are given precisely by the eigenvalues of the system matrix in \eqref{T_real_vel}, which are determined by \eqref{char_eq_s}. Then $T(s)$ is input-output stable if and only if its poles are on the open left half-plane.
%\fi
%
\iffalse
%short version
%We refer the reader to 
\defaultColor{See}
\cite{OralMalladaGayme2019_arxiv} for a proof.
\fi
%
\end{proof}
%}

%\begin{proof}
%
%\end{proof}
%

%Remark1
\begin{rem} \label{rem1}
Assumption \ref{assum1} 
%guarantees that the non-Hurwitz modes, if any, are unobservable from the performance output. 
%This assumption 
can be relaxed for specific values of $k_p$ and $k_d$ for which 
%some of these modes 
the consensus modes
become Hurwitz. If $k_p \neq 0$ and $k_d \neq 0$, the assumption can be relaxed for both \eqref{H_cases_1} and \eqref{H_cases_2} since 
%
%\begin{equation} \nonumber
$
H_{\bold{\tilde{x}}_1 \bold{\tilde{w}}_1} (s) = \frac{h_1 (s)}{f(s)} = g(s) = \frac{1}{s^2 + k_d s + k_p}
$
%\end{equation}
%
and
%
%\begin{equation} \nonumber
$
H_{\bold{\tilde{v}}_1 \bold{\tilde{w}}_1} (s) = s H_{\bold{\tilde{x}}_1 \bold{\tilde{w}}_1}(s) = \frac{s}{s^2 + k_d s + k_p}
$
%\end{equation}
%
have stable poles by the Routh-Hurwitz criterion. 

Similarly, one can relax the assumption for \eqref{H_cases_2} but 
%needs to keep it 
\defaultColor{not}
for \eqref{H_cases_1} if $k_p = 0$ and $k_d \neq 0$ since
%
%\begin{equation} \nonumber
\mbox{
$
H_{\bold{\tilde{x}}_1 \bold{\tilde{w}}_1} (s) = \frac{1}{s^2 + k_d s}
$
}
%\end{equation}
%
has a pole at $s=0$ but
%
%\begin{equation} \nonumber
\mbox{
$
H_{\bold{\tilde{v}}_1 \bold{\tilde{w}}_1} (s) = \frac{s}{s^2 + k_d s} = \frac{1}{s + k_d}
$
}
%\end{equation}
%
has a stable pole.
However for the sake of simplicity, we only consider performance metrics such that Assumption \ref{assum1} is satisfied for both \eqref{H_cases_1} and \eqref{H_cases_2}.
\end{rem}

\defaultColor{
The stability 
condition in
%given by 
Proposition \ref{prop3} can be restated as follows.}

%Proposition Stability Routh Hurwitz
\begin{prop} \label{prop_stab_rh}
Consider the double-integrator network \eqref{secondmatrix} and suppose that $L$ is diagonalizable and assumptions \ref{assum_gain} and \ref{assum1} hold.
The system $T$ given by \eqref{H_cases_1} or \eqref{H_cases_2} is input-output stable if and only if
\begin{equation} \label{stab_rh}
\alpha_k \phi_k^2 + \beta_k \xi_k \phi_k - \beta_k^2 > 0 \ \text{and} \ \phi_k > 0, \ \ 
k \in \mathcal{N}_{obsv},
%k=2,\dots,m,
\end{equation}
where $\alpha_k = k_p + \gamma_p \R[\lambda_k]$, $\phi_k = k_d + \gamma_d \R[\lambda_k]$, $\beta_k = \gamma_p \I[\lambda_k]$ and $\xi_k = \gamma_d\I[\lambda_k]$.
%Then, the system $T$ given by \eqref{H_cases_1} or \eqref{H_cases_2} is input-output stable.
\end{prop}
\begin{proof}
The result follows from applying \cite[Lemma 4]{LiHuang2010} to Proposition \ref{prop3}.
\end{proof}

%Utilizing the block-diagonalization procedure outlined in this section, the next section provides our main results. 
%of this work.

%SECTION 4
\section{Performance over Arbitrary Digraphs} \label{ArbitraryL}

In this section, 
\defaultColor{we use the block-diagonalization procedure outlined in Section \ref{BlockDiag}
to derive closed-form expressions} for the performance of the single and double-integrator networks \eqref{firstmatrix} and \eqref{secondmatrix} over arbitrary directed graphs \defaultColor{that have at least one} globally reachable node. Throughout the discussion we use both time and frequency domain representations, which simplifies the analysis.

First, we simplify \eqref{perf_x_2} using the block-diagonal form of \eqref{block_diag} and show that performance can be quantified as a linear combination of scalar integrals. These integrals can be interpreted as $\mathcal{L}_2$ scalar products of the elements of the closed-loop impulse response function matrix blocks $H_{\bold{\tilde{x}}_k \bold{\tilde{w}}_k} (t)$ and $H_{\bold{\tilde{v}}_k \bold{\tilde{w}}_k} (t)$. 
%In the following, we adopt a common notation between the two performance measures by replacing the subscripts for position and velocity with a generic one, as the same procedure applies to both. 

%%%%%%%%%%%%%%%%%%%%%%%%%%%%%%%%%%%%%
%COMMENT
\iffalse

%COMMENT
\fi
%%%%%%%%%%%%%%%%%%%%%%%%%%%%%%%%%%%

Combining \eqref{perf_x_2} and \eqref{H_x_hat}, the performance metric in \eqref{perf_x_2} can be written as
\begin{equation} \label{perf_generic}
P = \int_{0}^{\infty} \bold{w}_0^* \tilde{Q}^* \tilde{H} (t)^* \tilde{N} \tilde{H} (t) \tilde{Q} \bold{w}_0 dt,
\end{equation}
where $\tilde{N} = \tilde{R}^* C^* C \tilde{R}$ and $\tilde{H}$ is defined as in \eqref{T_tilde} with
\begin{equation} \label{TF_block}
\tilde{H}_k (s) = 
%r(s)
%T_{\bold{\tilde{x}}_i \bold{\tilde{w}}_i}(s) = 
\begin{bmatrix}
\tilde{h}_{11}^{(k)}(s) & \dots & \tilde{h}_{1,n_k}^{(k)}(s) \\
& \ddots & \vdots \\
& & \tilde{h}_{n_k,n_k}^{(k)}(s)
\end{bmatrix}
\end{equation}
for $k=2,\dots,m$. The upper triangular form of \eqref{TF_block} is given in Lemma \ref{lemma1}. Since 
\begin{equation} \label{matrix_M}
M := C^* C
\end{equation}
is a symmetric matrix, it is 
%orthogonally
unitarily 
diagonalizable, i.e. 
$$ \! M = \Theta W \Theta^*, \ \
%
\iffalse
W = \left[ \begin{smallmatrix} 
\mu_1 & & \\ 
& \ddots & \\
& & \mu_n
\end{smallmatrix} \right] \in \mathbb{R}^{n \times n}, \ \
\fi
%
%
W = \diag{ (\mu_i)_{1 \leq i \leq n} }
\in \mathbb{R}^{n \times n} \! , \ \ 
\text{and} \ \
\Theta \Theta^* = I,$$ 
therefore $\tilde{N} = \tilde{R}^* \Theta W \Theta^* \tilde{R}$. Using Assumption \ref{assum1} and assuming without loss of generality that $\mu_1 = 0$ is associated with the eigenvector $\boldsymbol{\theta}_1 = \frac{1}{\sqrt{n}}\bold{1}$, we can state $\tilde{N}$ element-wise as
\begin{equation} \label{nu_eta_kappa}
(\tilde{N})_{\eta-1, \kappa-1} = \sum_{l=2}^{n} 
\langle \boldsymbol{\theta}_l, \bold{r}_{\eta}  \rangle 
\langle \bold{r}_{\kappa}, \boldsymbol{\theta}_l  \rangle
\mu_l
=: \nu_{\eta, \kappa}
\end{equation}
for $\eta, \kappa = 2, \dots, n$, where $\langle \boldsymbol{\theta}_l, \bold{r}_{\eta}  \rangle = \bold{r}_{\eta}^*\boldsymbol{\theta}_l$, $\bold{r}_{\kappa}$ and $\boldsymbol{\theta}_l$ denote respectively the columns $\kappa$ and $l$ of $\tilde{R}$ and $\Theta$. 

Using this notation, \eqref{perf_generic} can be written in terms of the scalar products between the elements of $\tilde{H}_k (t)$, which are given by the element-wise inverse Laplace transforms of \eqref{TF_block}.

%Lemma 2
\begin{lemma} \label{lemma2}
The performance metric $P$ in \eqref{perf_generic} is given by
\begin{equation} \label{perf_generic_sol1}
P = \tr{(\Sigma_Q \Psi )},
\end{equation}
where 
\begin{equation} \label{Sigma}
\Sigma_Q = \tilde{Q} \Sigma_0 \tilde{Q}^*, \, \Sigma_0 = \bold{w}_0 \bold{w}_0^*,
\end{equation}
and the matrix $\Psi$ is partitioned as
%
%\begin{equation} \nonumber
$
%
\iffalse
\Psi = 
\begin{bmatrix}
\Psi_{22} & \dots & \Psi_{2m} \\
\vdots & \ddots & \vdots \\
 \Psi_{m2} & \dots & \Psi_{mm}
\end{bmatrix}.
\fi
%
%
\Psi = \left[ \Psi_{kl} \right]_{2 \leq k,l \leq  m}. 
$
%\end{equation}
%

Furthermore, the entry $(q, b)$ of the matrix $\Psi_{kl}$ for $k, l ~=~ 2,\dots, m$ is given by
\begin{equation} \label{psi_ik}
\left[ \Psi_{kl} \right]_{q b} = \sum_{p=1}^{q} \sum_{a=1}^{b} \nu_{d_k+p, d_l+a} 
\left \langle \tilde{h}_{a b}^{(l)}(t), \tilde{h}_{p q}^{(k)}(t) 
\right \rangle_{\mathcal{L}_2},
\end{equation}
where
\begin{equation}  \label{L2_scalar}
\left \langle \tilde{h}_{a b}^{(l)}(t), \tilde{h}_{p q}^{(k)}(t) 
\right \rangle_{\mathcal{L}_2} = 
\int_0^{\infty} \overline{ \tilde{h}_{p q}^{(k)} (t) } \tilde{h}_{a b}^{(l)}(t) dt.
\end{equation}
Here the indices $q = 1,\dots,n_k$ and $b = 1,\dots,n_l$ are determined by the Jordan block sizes $n_k$ and $n_l$. Terms of the form in
%$\nu_{\eta,\kappa}$ in
 \eqref{nu_eta_kappa} 
% multiply the scalar products \eqref{L2_scalar} 
appear in the summand of \eqref{psi_ik} and their indices take values larger than the sum of the previous Jordan block sizes, namely $d_k = \sum_{i=1}^{k-1} n_i$ and $d_l = \sum_{i=1}^{l-1} n_i$.
\end{lemma}
%
%
%Remark 2
\begin{rem} \label{rem2}
For the special case in which $L$ is diagonalizable each Jordan block is a scalar, i.e. $n_k = 1$, and \eqref{psi_ik} leads to
\begin{equation} \nonumber
\Psi_{kl} = \nu_{kl} 
\left \langle \tilde{h}^{(l)}(t), \tilde{h}^{(k)}(t) 
\right \rangle_{\mathcal{L}_2}.
\end{equation}
%
%which is also a scalar. 
Here we dropped the subscripts of $\tilde{h}_{p q}^{(k)}$ for simplicity. 
The case with diagonalizable $L$ was studied in \cite{PaganiniMallada2017, PaganiniMallada2019} and Lemma \ref{lemma2} provides a generalization to the case of arbitrary Jordan block size $n_k$ for $k=2, \dots, m$.
\end{rem}
\begin{proof}[Proof of Lemma \ref{lemma2}]
Taking the trace of both sides of \eqref{perf_generic} and using the permutation property of the trace, we have
%
%\begin{equation} \nonumber
$
P = \tr \left( \tilde{Q} \bold{w}_0 \bold{w}_0^* \tilde{Q}^*\Psi \right),
$
%\end{equation}
%
where $\Psi (t) = \int_{0}^{\infty} \tilde{H}(t)^* \tilde{N} \tilde{H} (t) dt$. Partitioning $\tilde{N}$ conformally so that its $(k,l)$ block is given by $\tilde{N}_{kl}$, one can write
\begin{equation} \label{psi_kl_full}
\Psi_{kl} = \int_{0}^{\infty}
\tilde{H}_k (t) ^ *
\tilde{N}_{kl}
\tilde{H}_l (t)  dt,
\end{equation}
for $k,l = 2,\dots,m$. Direct multiplication of the matrices in the integral argument and interchanging the order of integration with the summation gives the desired result.
\end{proof}
%
%Remark: Performance always real
\begin{rem} \label{rem_perf_real}
Since $\tilde{N} = \tilde{N} ^ *$, i.e. $\tilde{N}_{kl} = \tilde{N}_{lk} ^ *$, \eqref{psi_kl_full} leads to $\Psi_{kl} = \Psi_{lk} ^ *$, therefore $\Psi$ is Hermitian. The fact that $\Sigma_{Q}$ in \eqref{Sigma} is also Hermitian leads to 
%$$
%
$
\tr{(\Sigma_Q \Psi )} =  \tr{\left[ (\Sigma_Q \Psi)^* \right]}
= \overline{\tr{(\Sigma_Q \Psi )}},
$
%$$
which verifies that $P$ in \eqref{perf_generic_sol1} is real as expected.
\end{rem}

As Lemma \ref{lemma2} 
indicates, 
%\blue{shows,}
\eqref{perf_generic_sol1} can be expressed in closed-form if the integral in \eqref{L2_scalar} can be evaluated. In what follows, we derive time-domain realizations for the transfer functions $\tilde{h}_{p q}^{(k)}(s)$ in \eqref{TF_block}, which enables the evaluation of this integral. This 
%approach 
leads to 
\defaultColor{our first major contribution, which we state next for
%the computation of performance for the 
single and double-integrator networks \eqref{firstmatrix} and \eqref{secondmatrix}.}
%over arbitrary directed graphs with a globally reachable node. 
%We first focus on the single-integrator network.

%Section3 - Subsection 1
\subsection{Performance of Single-Integrator Networks}

We now present the main result of this section for the single-integrator network \eqref{firstmatrix}. The following result provides a closed-form solution for the performance metric $P$ 
in \eqref{perf_x}. 
%by solving the integral in \eqref{L2_scalar}.
%
%Theorem 1
\begin{thm} \label{thm1}
Consider the single-integrator network \eqref{firstmatrix}. The performance metric $P$ in \eqref{perf_x} for the system $T$ given by \eqref{H_cases_1} is
%
%\begin{equation} \nonumber
%
$
P = \tr{(\Sigma_Q \Psi)}.
$ 
%\end{equation}
%
The elements of $\Psi$ are defined in \eqref{psi_ik} and the scalar product in \eqref{L2_scalar} is given by
%
%\footnotesize
%\begin{equation}  \label{L2_scalar_sol_consensus}
%\left \langle \tilde{h}_{l,d_k}^{(k)}(t), \tilde{h}_{f,a_i}^{(i)}(t) 
%\right \rangle_{\mathcal{L}_2} = 
%\frac{{(-1)}^{\iota_i + \upsilon_k-2} \Xi (\iota_i, \upsilon_k) }
%{ {\left( \overline{\lambda_i} + \lambda_k \right)}^{\iota_i + \upsilon_k - 1} },
%\end{equation} 
%\normalsize
%
\begin{equation}  \label{L2_scalar_sol_consensus}
\left \langle \tilde{h}_{a b}^{(l)}(t), \tilde{h}_{p q}^{(k)}(t) 
\right \rangle_{\mathcal{L}_2} = 
\frac{{(-1)}^{b - a + q - p} \Phi }
{ {\left( \overline{\lambda_k} + \lambda_l \right)}^{b - a + q - p + 1} },
\end{equation} 
where 
$\Phi = 
\frac{(b - a + q - p)!}{(b - a)! (q - p)!} $.
\end{thm}
 \begin{proof}
Using the result of Corollary \ref{cor1} and the notation in \eqref{TF_block} 
 \begin{equation} \nonumber
 \tilde{h}_{p q}^{(k)}(s) = 
 {(-1)}^{q-p}
 \frac{1}
 { {\left( s + \lambda_k \right)}^{q-p+1}}.
 \end{equation}
Here,  
$\frac{1}
 { {\left( s + \lambda_k \right)}^{q-p+1}}$
has the following realization $\left( \mathcal{A}_{k,\delta}, \mathcal{B}_{k,\delta}, \mathcal{C}_{k,\delta} \right)$ in JCF 
\begin{equation} \label{Jordan_realization_3}
\mathcal{A}_{k,\delta} =
\mathcal{J} (-\lambda_k, \delta),
\end{equation}
\begin{equation} \nonumber
\mathcal{B}_{k,\delta} = \bigg[ \underbrace{\begin{matrix} 0 & \dots & 1 \end{matrix}}_{1 \times \delta} \bigg]^\intercal, 
\mathcal{C}_{k,\delta} = \bigg[ \underbrace{\begin{matrix} 1 & \dots & 0 \end{matrix}}_{1 \times \delta} \bigg],
\end{equation}
where $\mathcal{J} (-\lambda_k, \delta)$ denotes the size-$\delta$ Jordan block with the eigenvalue $-\lambda_k$ and $\delta = q - p + 1$. Then, $\tilde{h}_{p q}^{(k)}(t)$ is given by
\begin{equation} \label{h_impResp_sing}
\tilde{h}_{p q}^{(k)}(t) = 
 {(-1)}^{q - p}
C_{k,\delta} e^{A_{k,\delta} t} B_{k,\delta}.
\end{equation}
where 
%$
%{\left( e^{\mathcal{J} (\rho_1^{(i)}, \, \iota_i)} \right)}_{\Delta_i, \chi_i} = 
%\frac{t^{(\chi_i - \Delta_i - 1)}}{(\chi_i - \Delta_i - 1)!} e^{ \rho_1^{(i)}t }$, $\chi_i = 1, \dots, \iota_i$ and $\Delta_i = 1, \dots, \iota_i$.
%
\begin{equation} \label{jordan_exp}
e^{A_{k,\delta} t } = 
e^{\mathcal{J} (-\lambda_k, \delta) t} 
=
e^{ - \lambda_k t}
\begin{bmatrix}
1 & t & \dots & \frac{t^{(\delta - 1)}}{(\delta - 1)!} \\
& \ddots & \ddots & \vdots \\
& & & t \\
& & & 1
\end{bmatrix}.
\end{equation}
Combining \eqref{Jordan_realization_3} and \eqref{h_impResp_sing} leads to
\begin{equation} \nonumber
\tilde{h}_{p q}^{(k)}(t) = 
{(-1)}^{q - p}
e^{ -\lambda_k t}
 \frac{t^{q - p}}{(q - p)!}.
\end{equation}
The proof is completed by evaluating the integral in \eqref{L2_scalar} using the fact that
\iffalse
\blue{
$\int_{0}^{\infty} t^n e^{\lambda t} dt = {(-1)}^{n+1}\frac{n!}{\lambda^{n+1}}$ 
for $\lambda \in \mathbb{C}$, $\R [\lambda] ~<~ 0$.
}
\fi
%
%
$\int_{0}^{\infty} t^n e^{-\lambda t} dt = 
%{(-1)}^{n+1}
\frac{n!}{\lambda^{n+1}}$ 
for $\lambda \in \mathbb{C}$, $\R [\lambda] > 0$.
\end{proof}

The denominator of the right-hand side of \eqref{L2_scalar_sol_consensus} is given by a power of the sum of the graph Laplacian eigenvalues that are associated with possibly distinct Jordan blocks $k$ and $l$. The power of this term depends on the Jordan block sizes $n_k$ and $n_l$ through the indices $q$ and $b$ and it increases as the Jordan block size increases. This indicates that performance is affected not only by the network size, but also \defaultColor{by} the graph Laplacian spectrum and the size of the individual Jordan blocks.

We next present the 
%main result
\defaultColor{analogous result} 
for the double-integrator network \eqref{secondmatrix}.

%Section3 - Subsection 2
\subsection{Performance of Double-Integrator Networks} 
\label{PerfArbDouble}

%In this subsection 
We now provide the closed-form solution for the performance metric $P$ in \eqref{perf_x} for the double-integrator network \eqref{secondmatrix}. A similar approach to the one in Theorem \ref{thm1} 
%will be applied 
\defaultColor{is taken}
but the computation of the impulse response functions $\tilde{h}_{p q}^{(k)}(t)$ is more involved. \defaultColor{We compute these 
%impulse response 
functions through Lemmas \ref{lemma3} and \ref{lemma4} in the Appendix.
%conveniently 
%$\tilde{h}_{p q}^{(k)}(t)$, 
%which are the time-domain representations of the transfer functions in \eqref{TF_block}. 
Then by evaluating the integral in \eqref{L2_scalar}, the 
%main 
result of this 
%section
subsection 
%for the double-integrator network \eqref{secondmatrix} 
is stated as follows.}

\begin{thm} \label{thm2}
%\blue{
Consider the double-integrator network \eqref{secondmatrix}. 
Let $\rho_1^{(k)}$ and $\rho_2^{(k)}$ denote the roots of 
\begin{equation} \label{char_eq_s2}
s^2 + (k_d + \gamma_d \lambda_k) s + k_p + \gamma_p  \lambda_k = 0.
\end{equation}
The performance metric $P$ in \eqref{perf_x} for the system $T$ given by \eqref{H_cases_1} or \eqref{H_cases_2} is
%
%\begin{equation} \label{perf_x_sol1}
$
P = \tr{(\Sigma_Q \Psi)}
$,
%\end{equation}
%
%and
%
%\begin{equation} \label{perf_v_sol1}
%P_{\bold{v}} = \tr{(\Sigma_Q \tilde{\Psi}_{\bold{v}})},
%\end{equation}
%
where $\Psi$ is given element-wise by \eqref{psi_ik} and the scalar product in \eqref{L2_scalar} is as follows:
%
%\begin{itemize}
%
%\item 

\noindent
If $\rho_1^{(k)} \neq \rho_2^{(k)}$ and $\rho_1^{(l)} \neq \rho_2^{(l)}$
%
%\normalsize
%
\begin{align}  \label{L2_scalar_sol_1}
&\left \langle \tilde{h}_{a b}^{(l)}(t), \tilde{h}_{p q}^{(k)}(t) 
\right \rangle_{\mathcal{L}_2} = 
\sum_{\zeta=1}^{\sigma} \sum_{r=1}^{\upsilon}
\frac{ \Phi_{\zeta r} (\sigma, \upsilon) \overline{c_{\zeta}^{(k)}} c_{r}^{(l)} }
{ {\left( \overline{\rho_1^{(k)}} + \rho_1^{(l)} \right)}^{\sigma + \upsilon - \zeta - r + 1} } \nonumber \\
&+ 
\frac{ \Phi_{\zeta r} (\sigma, \upsilon) \overline{c_{\zeta}^{(k)}} c_{r+\upsilon}^{(l)} }
{ {\left( \overline{\rho_1^{(k)}} + \rho_2^{(l)} \right)}^{\sigma + \upsilon - \zeta - r + 1} } \nonumber
\ + \ 
\frac{ \Phi_{\zeta r} (\sigma, \upsilon) \overline{c_{\zeta+\sigma}^{(k)}} c_{r}^{(l)} }
{ {\left( \overline{\rho_2^{(k)}} + \rho_1^{(l)} \right)}^{\sigma + \upsilon - \zeta - r + 1} } \nonumber \\
&+ 
\frac{ \Phi_{\zeta r} (\sigma, \upsilon) \overline{c_{\zeta+\sigma}^{(k)}} c_{r+\upsilon}^{(l)} }
{ {\left( \overline{\rho_2^{(k)}} + \rho_2^{(l)} \right)}^{\sigma + \upsilon - \zeta - r + 1} },
\end{align} 
%\normalsize
%
%\item 
If $\rho_1^{(k)} \neq \rho_2^{(k)}$ and $\rho_1^{(l)} = \rho_2^{(l)} = \rho^{(l)}$
%
%\normalsize
\begin{align}  \label{L2_scalar_sol_2}
\! \! \! \!
\left \langle \tilde{h}_{a b}^{(l)}(t), \tilde{h}_{p q}^{(k)}(t) 
\right \rangle_{\mathcal{L}_2} &= 
\sum_{\zeta=1}^{\sigma} \sum_{r=1}^{2\upsilon}
\frac{{(-1)}^{\upsilon} \Phi_{\zeta r} (\sigma, 2\upsilon) \overline{c_{\zeta}^{(k)}} c_{r}^{(l)} }
{ {\left( \overline{\rho_1^{(k)}} + \rho^{(l)} \right)}^{\sigma + 2\upsilon - \zeta - r + 1} } \nonumber \\
&+ 
\frac{{(-1)}^{\upsilon} \Phi_{\zeta r} (\sigma, 2\upsilon) \overline{c_{\zeta+\sigma}^{(k)}} c_{r}^{(l)} }
{ {\left( \overline{\rho_2^{(k)}} + \rho^{(l)} \right)}^{\sigma + 2\upsilon - \zeta - r + 1} },
\end{align} 
%\normalsize
%
%\item 
If $\rho_1^{(k)} = \rho_2^{(k)} = \rho^{(k)}$ and $\rho_1^{(l)} = \rho_2^{(l)} = \rho^{(l)}$
%
%\footnotesize
\begin{equation}  \label{L2_scalar_sol_3}
\! \!
\left \langle \tilde{h}_{a b}^{(l)}(t), \tilde{h}_{p q}^{(k)}(t) 
\right \rangle_{\mathcal{L}_2} \! = \!
\sum_{\zeta=1}^{2\sigma} \sum_{r=1}^{2\upsilon}
\frac{{(-1)}^{\sigma + \upsilon} \Phi_{\zeta r} (2\sigma, 2\upsilon) \overline{c_{\zeta}^{(k)}} c_{r}^{(l)} }
{ {\left( \overline{\rho^{(k)}} + \rho^{(l)} \right)}^{2\sigma + 2\upsilon - \zeta - r + 1} } \! ,
\end{equation} 
%\normalsize
%
%\end{itemize}
%
where 
$\sigma = q - p + 1$,
$\upsilon = b - a + 1$
and
$\Phi_{\zeta r} (\sigma, \upsilon) ~=~ 
{(-1)}^{ 1 - \zeta - r}
\frac{(\sigma+\upsilon - \zeta - r)!}{(\sigma - \zeta)! (\upsilon - r)!} $. 

The coefficients $c_{\zeta}^{(k)}$ are given in the Appendix by Lemma~\ref{lemma3} if $\rho_1^{(k)} \neq \rho_2^{(k)}$ and by Lemma~\ref{lemma4} if $\rho_1^{(k)} = \rho_2^{(k)}$. 
%
%}
\end{thm}
\defaultColor{
\begin{rem}
For 
%the 
double-integrator networks, 
%it is not straightforward to understand the precise effect of 
the scalar products in 
%\eqref{L2_scalar} 
\eqref{L2_scalar_sol_1} -
%\eqref{L2_scalar_sol_2}, 
\eqref{L2_scalar_sol_3}
depend on both 
the control gains and the eigenvalues of 
%the 
%weighted 
%graph Laplacian 
$L$,
%simultaneously
%through
via 
the roots of \eqref{char_eq_s2} 
%as well as
and 
the coefficients $c_{\zeta}^{(k)}$.
%on performance since the roots of \eqref{char_eq_s2} depend simultaneously on both quantities. 
In contrast, for 
%the 
single-integrator networks, 
%the solution simplifies and 
%the role of 
%the Laplacian 
%the 
eigenvalues of $L$ appear explicitly in the analogous expression
%is more apparent 
%as shown by
in 
\eqref{L2_scalar_sol_consensus}.
\end{rem}
}
\defaultColor{
\begin{proof}[Proof of Theorem \ref{thm2}]
Using the result of Corollary \ref{cor2}, the notation in \eqref{TF_block} and \eqref{phi_TF}, $\tilde{h}_{p q}^{(k)}(t)$ is given by
\begin{equation} \label{h_impResp}
\tilde{h}_{p q}^{(k)}(t) = {(-1)}^{\sigma -1} \Omega_{k, \sigma} (t).
\end{equation}
If $\rho_1^{(k)} \neq \rho_2^{(k)}$, the realization in \eqref{Jordan_realization} can be used to calculate
\begin{equation} \nonumber
\Omega_{k, \sigma} (t) = C_{k,\sigma} e^{A_{k,\sigma} t} B_{k,\sigma}, \\
\end{equation}
where $e^{A_{k,\sigma} t} = 
%
\iffalse
%\left[
\begin{bmatrix}
e^{\mathcal{J} (\rho_1^{(k)}, \, \sigma)t}  & \\
& e^{\mathcal{J} (\rho_2^{(k)}, \, \sigma)t}
\end{bmatrix}
%\right]
\fi
%
\bdiag{ \left(
e^{\mathcal{J} (\rho_i^{(k)}, \, \sigma)t} \right) }_{i=1,2}
$
and 
%\iffalse
$ e^{\mathcal{J} (\rho_i^{(k)}, \, \sigma) t} $ can be expanded as in \eqref{jordan_exp}.
%\fi
%$
%{\left( e^{\mathcal{J} (\rho_1^{(i)}, \, \iota_i)} \right)}_{\Delta_i, \chi_i} = 
%\frac{t^{(\chi_i - \Delta_i - 1)}}{(\chi_i - \Delta_i - 1)!} e^{ \rho_1^{(i)}t }$, $\chi_i = 1, \dots, \iota_i$ and $\Delta_i = 1, \dots, \iota_i$.
%
\iffalse
\begin{equation} \nonumber
 e^{\mathcal{J} (\rho_i^{(k)}, \, \sigma) t} = 
e^{ \rho_i^{(k)} t}
\begin{bmatrix}
1 & t & \dots & \frac{t^{(\sigma - 1)}}{(\sigma - 1)!} \\
& \ddots & \ddots & \vdots \\
& & & t \\
& & & 1
\end{bmatrix}.
\end{equation}
\fi
%
Then, using \eqref{h_impResp} and the definitions of $C_{k,\sigma}$ and $B_{k,\sigma}$ in \eqref{Jordan_realization}
\begin{equation} \nonumber
\tilde{h}_{p q}^{(k)}(t) = 
{(-1)}^{\sigma -1}
\sum_{\zeta=1}^{\sigma} 
\left( c_{\zeta}^{(k)} e^{ \rho_1^{(k)} t} + c_{\zeta + \sigma}^{(k)} e^{ \rho_2^{(k)} t}  \right)
 \frac{t^{\sigma - \zeta}}{(\sigma - \zeta)!}.
\end{equation}
If $\rho_1^{(k)} = \rho_2^{(k)} = \rho^{(k)}$, a similar argument combined with \eqref{Jordan_realization_2} leads to 
\begin{equation} \nonumber
\tilde{h}_{p q}^{(k)}(t) = 
{(-1)}^{\sigma -1}
\sum_{\zeta=1}^{2\sigma} 
c_{\zeta}^{(k)} e^{ \rho^{(k)} t}
 \frac{t^{2\sigma - \zeta}}{(2\sigma - \zeta)!}.
\end{equation}
The proof is completed by evaluating the integral in \eqref{L2_scalar} using the fact that 
$\int_{0}^{\infty} t^n e^{\lambda t} dt = {(-1)}^{n+1}\frac{n!}{\lambda^{n+1}}$ 
for $\lambda \in \mathbb{C}$, $\R [\lambda] ~<~ 0$.
%$\int_{0}^{\infty} t^n e^{a t} dt = {(-1)}^{n+1}\frac{n!}{a^{n+1}}$ 
%[reference] 
%leads to the desired result.
\end{proof}
Theorems \ref{thm1} and \ref{thm2} provide closed-form solutions for the performance metric \eqref{perf_x} which consist of terms that: (a) are geometric, i.e. terms that depend on the input direction, the eigenvalues and the eigenvectors of $M$ in \eqref{matrix_M} and the eigenvectors of $L$ as in \eqref{nu_eta_kappa} and \eqref{Sigma}; and (b) terms that depend on the closed-loop dynamics of the system, as in \eqref{L2_scalar}. Overall, performance is given by a linear combination of the 
%elements 
entries of the matrix $\Psi$ in
\eqref{psi_ik}, weighted by the 
%elements
entries 
of the matrix $\Sigma_Q$ in \eqref{Sigma}. Therefore, in the most general case, it is not straightforward to deduce the individual effect of properties such as network size, graph topology and the 
%spectral properties 
spectrum
of the output matrix 
%from these closed-form solutions.
for an arbitrary system. 
}

In the next section, we study special cases to provide 
%further 
insight on the effect of edge directionality on 
performance.

%SECTION 5 
\section{Digraphs with Diagonalizable Laplacian Matrices} \label{diagL}

\defaultColor{
We now consider 
the class of graphs that emit
%the special case concerning
%the case where the graph 
%special graph structures emitting 
%Laplacian matrix is diagonalizable; 
diagonalizable Laplacian matrices;
and its subclass of
%As an important subclass of it,
%we also study 
normal Laplacian matrices.
%where the latter is 
%an important subclass of the former.
%In this section, we investigate cases that simplify the results of the previous section so that 
%To begin, we will consider diagonalizable Laplacians.
The closed-form solutions derived here will be used in sections \ref{NormalL} and \ref{CentDistFeed}
%in order 
to 
%enhance tractability of 
\defaultColor{provide further insights on}
the results of the last section.
%to study imploding star graphs.
%
\iffalse
As a special case of the diagonalizable Laplacians, we will further investigate 
%how performance scales with network size and graph topology for 
the case of normal Laplacian matrices.
%This enables not only the diagonalization of the Laplacian; 
\fi
Namely, we will show that
performance is determined by the interplay between
%the interplay between 
edge directionality
%the role of directed communication 
%(the role of complex eigenvalues of the Laplacian) 
and control strategy (judicious selection of control gains).
\iffalse
For disturbance inputs that are uniform across the network, 
%the terms that include the scalar products of the Laplacian and the output matrix eigenvectors can be eliminated.
%which significantly simplify the solutions for the performance measures. 
the closed-form computations will reveal how performance is simultaneously determined by 
%the interplay between 
edge directionality
%the role of directed communication 
(the role of complex eigenvalues of the Laplacian) and control strategy (judicious selection of control gains). 
%This class of systems will be studied in detail in sections \ref{NormalL} and \ref{CentDistFeed}. 
\fi 
%
}

\subsection{Single-Integrator Networks}

The following theorem provides the main result of this section for the single-integrator network \eqref{firstmatrix}.
%
%Theorem 3
\begin{thm} [\defaultColor{Single-Integrator, Diagonalizable Laplacian}]
\label{thm3}
Consider the single-integrator network \eqref{firstmatrix} and suppose that $L$ is diagonalizable. Then, the metric $P$ in \eqref{perf_x} for the system $T$ given by \eqref{H_cases_1} is
%
%\begin{equation} \nonumber
%
$
P = \tr{(\Sigma_Q \Psi)}
$,
%\end{equation}
%
where $\bold{j}^2 = -1$ and
\begin{equation} \label{psi_ik_diagable}
\Psi_{kl} =
\nu_{kl} 
\frac{\R[\lambda_k]+\R[\lambda_l] + \bold{j} (\I[\lambda_k]-\I[\lambda_l])}
{(\R[\lambda_k]+\R[\lambda_l])^2 +(\I[\lambda_k]-\I[\lambda_l])^2}.
%
%\quad \bold{j}^2 = -1
\end{equation}
%and $\bold{j}^2 = -1$.
%
\end{thm}
\begin{proof}
The fact that $L$ is diagonalizable leads to $m=n$, i.e. all Jordan blocks are scalars. Then using \eqref{psi_ik} from Lemma~\ref{lemma2}, we have
$
\Psi_{kl}  = 
\nu_{kl} 
\left \langle \tilde{h}_{11}^{(l)}(t), \tilde{h}_{11}^{(k)}(t) 
\right \rangle_{\mathcal{L}_2}
$.
Also,
\eqref{L2_scalar_sol_consensus} from Theorem~\ref{thm1} gives
$
%\begin{equation} 
\left \langle \tilde{h}_{11}^{(l)}(t), \tilde{h}_{11}^{(k)}(t) 
\right \rangle_{\mathcal{L}_2} = 
\frac{1} {\overline{\lambda_k} + \lambda_l }. 
%\end{equation} 
$
Combining these facts and re-arranging terms yields the result. 
\end{proof}
Note that the diagonal terms $\Psi_{kk}$ are real and the cross-terms $\Psi_{kl}$ for $k \neq l$ are possibly imaginary in \eqref{psi_ik_diagable}. However, $P$ is guaranteed to be real due to Remark \ref{rem_perf_real}.

%
%\subsubsection*{Special Case of Normal Laplacian Matrices}
\subsubsection*{Normal Laplacian Matrices}

We next focus on systems over digraphs that emit normal weighted Laplacian matrices. 
%We begin with the single-integrator network \eqref{firstmatrix}.
%We need the following definition.
First recall Definition \ref{obsv_indices}, which introduced the set of observable indices $\mathcal{N}_{obsv}$ in \eqref{obsv_ind_eqn}. If $L$ is normal therefore diagonalizable, we can re-state this set as
%
%\begin{defn}
%The set of observable indices $
%\mathcal{N}_{obsv}$ is given by
\begin{equation*}
\mathcal{N}_{obsv} =
\left\{
k \in \{ 2, \dots, n \}
\mid
C \bold{r}_k \neq 0 
%C \boldsymbol{\theta}_k \neq 0
\right\},
\end{equation*}
%where
recalling that 
$\bold{r}_k$ 
%$\boldsymbol{\theta}_k$ 
denote the right 
%generalized 
eigenvectors of 
%$M$ 
$L$ as defined in \eqref{RandQ}.
%in \eqref{matrix_M}.
%\end{defn}
%
%The following 
%
%\iffalse
We now present two lemmas that will be useful in proving the upcoming results.
%\fi 
%
\iffalse
We now present a lemma that will be useful in proving the upcoming results.
\fi
%
%\iffalse
%
\begin{lemma} \label{prop_obsv}
For $k \in \{ 2, \dots, n \}$, the eigenvalue-eigenvector pair $(\mu_k, \boldsymbol{\theta}_k)$ of $M$ in \eqref{matrix_M} satisfies $\mu_k = 0$ if and only if 
$C \boldsymbol{\theta}_k = 0$.
%$k \notin \mathcal{N}_{obsv}$.
\end{lemma}
\begin{proof} 
%
%
%\iffalse
%arxiv version
Assume for any $k \in \{ 2, \dots, n \}$ that $\mu_k = 0$. Then $0 = M \boldsymbol{\theta}_k = C^T C \boldsymbol{\theta}_k$. This implies that the vector $C \boldsymbol{\theta}_k$ is in the left nullspace of $C$, therefore is orthogonal to the column space of $C$. But $C \boldsymbol{\theta}_k$ also has to be in the column space of $C$ therefore $C \boldsymbol{\theta}_k = 0$. 

Conversely, if $C \boldsymbol{\theta}_k = 0$ for any $k \in \{ 2, \dots, n \}$, then $0 ~=~ M \boldsymbol{\theta}_k = \mu_k \boldsymbol{\theta}_k$ which gives $\mu_k = 0$ since $\boldsymbol{\theta}_k \neq 0$.
%
%\fi
%
\iffalse
%short version
%We refer the reader to 
%\red{
See
\cite{OralMalladaGayme2019_arxiv} for a proof.
%}
\fi
%
\end{proof}
%
%\fi
%

\begin{lemma} \label{prop_obsv_modes}
Suppose that $L$ is normal. For $k \in \{ 2, \dots, n \}$, $\nu_{kk}$ in \eqref{nu_eta_kappa} satisfies
\begin{enumerate}[wide, labelwidth=!, labelindent=0pt]
\item $\nu_{kk} = 0$ if and only if $k \notin \mathcal{N}_{obsv}$.
\item $\nu_{kk} > 0$ if and only if $k \in \mathcal{N}_{obsv}$.
\end{enumerate}

\end{lemma}
\begin{proof}
%
%\iffalse
Normality of $L$ means that it is unitarily diagonalizable, therefore $R^{-1} = {R}^*$.
%we observe that due to Assumption \ref{assum1} and 
%the fact that  $L$ is normal and 
We also recall that $M$ in \eqref{matrix_M} is symmetric, therefore unitarily diagonalizable. Therefore $\bold{r}_1 = \boldsymbol{\theta}_1 = 
\frac{1}{\sqrt{n}} \boldsymbol{1}
%\in \spn \{ \boldsymbol{1} \}
$ and
%and $\mathcal{G}$ has a globally reachable node, 
it holds that $\bold{r}_k, \boldsymbol{\theta}_l \in \spn \{ \boldsymbol{1} \} ^ {\bot} \subset \mathbb{C}^n$ for $k, l \in~ \{ 2,\dots,n \}$. 
So, 
%given any $k \in \{ 2,\dots,n \}$ 
we have
$\bold{r}_k = \sum_{i=2}^{n} \chi_{i}^{k} \boldsymbol{\theta}_i$ with constants $\chi_{i}^{k} \in~ \mathbb{C}$ for
$k \in \{ 2,\dots,n \}$.
%$| \langle \bold{r}_{k}, \boldsymbol{\theta}_l  \rangle |^2 \neq 0$ for some $l \in \{ 2,\dots,n \}$. 
%Then 
%for any $k \in \{ 2,\dots,n \}$, 
%$\nu_{kk} > 0$ for $k = 2, \dots, n$ due to \eqref{nu_eta_kappa}.

Given any $k \in \{ 2, \dots, n \}$, it follows from \eqref{nu_eta_kappa} and Lemma~\ref{prop_obsv} that $\nu_{kk} = 0$ if and only if $\left \langle\boldsymbol{\theta}_l, \bold{r}_k \right \rangle = 0$ for all $l \in~ \{ 2, \dots, n \}$ such that $C \boldsymbol{\theta}_l \neq 0$. 
%In addition, normality of $L$ means that it is unitarily diagonalizable, therefore $R^{-1} = {R}^*$.
%Then $\bold{r}_k^* \bold{1} = 0$, 
%and 
%$\bold{r}_k$ can be represented by $\boldsymbol{\theta}_l$ for $l \in \{ 2, \dots, n \}$ over the field of complex numbers $\mathbb{C}$, 
%we have
%$\bold{r}_k = \sum_{i=2}^{n} \chi_{i}^{k} \boldsymbol{\theta}_i$ with constants $\chi_{i}^{k} \in~ \mathbb{C}$. 
Combining the preceding arguments leads to
$$
\nu_{kk} = 0 \ \Leftrightarrow \
\left( \sum_{i=2}^{n} \overline{\chi_{i}^{k}} \boldsymbol{\theta}_i^* \right)
\boldsymbol{\theta}_l = 0,
\ l \in \{ 2, \dots, n \}, \
C \boldsymbol{\theta}_l \neq 0,
$$
which is equivalent to having $\chi_{l}^{k} = 0$ for such $l$, due to the orthonormality of $\boldsymbol{\theta}_l$. In other words,
%$$
$
\nu_{kk} = 0 \ \Leftrightarrow \
\bold{r}_k \! \! ~=~ \! \! \! \!
\sum\limits_{\substack{C \boldsymbol{\theta}_i = 0, \ i \in \{ 2, \dots, n \} }} 
%\! \! \! \! \! \!
\! \!
\chi_{i}^{k} \boldsymbol{\theta}_i
\ \Leftrightarrow \
C \bold{r}_k = 0,
%\quad \ \ \,
%\qedhere
%\! \! \! \! \!
%\! \! \! \! \!
%\! \! \! \! \!
$
%$$ 
which proves the first result. Since $M$ in \eqref{matrix_M} is postive semi-definite, $\nu_{kk}$ for $k \in \{ 2, \dots, n \}$ is given by a summation in \eqref{nu_eta_kappa} with each summand being non-negative. So, $\nu_{kk} \geq 0$ and the first result implies the second result.
%\fi
%
\iffalse
%We refer the reader to 
\blue{See}
\cite{OralMalladaGayme2019_arxiv} for a proof.
\fi
%
\end{proof}

\defaultColor{
%Next we apply 
Theorem \ref{thm3} 
%to the special case of 
leads to the following result for
normal Laplacian matrices. }

%Corollary 4
\begin{cor} [\defaultColor{Single-Integrator, Normal Laplacian}]
\label{cor4}
Consider the single-integrator network \eqref{firstmatrix}. Suppose that $L$ is normal and the input $\bold{w}_0$ 
%is 
%random and 
\defaultColor{has unit covariance,
%spatially uncorrelated with unit variance,
i.e. $E\left[\Sigma_0\right] = I $}. Then, the expectation of the metric $P$ in \eqref{perf_x} for the system $T$ given by \eqref{H_cases_1} is
\begin{equation} \label{perf_single_normal}
E\left[P\right] = 
\lVert T \rVert_{\mathcal{H}_2}^2
=
\sum_{k \in \mathcal{N}_{obsv}} \nu_{kk}
\frac{1}
{2\R[\lambda_k]}.
\end{equation}
\end{cor}
\begin{proof}
%Since $L$ is normal, it is unitarily diagonalizable therefore $R^{-1} = {R}^*$. 
%This implies that 
Orthonormality of $\bold{r}_j$ 
%are orthonormal 
for $j = 1, \dots, n$ 
%which 
yields $E\left[\Sigma_Q\right] ~=~ I$ and leads to $E\left[P\right] = \sum_{k=2}^{n} \Psi_{kk}$ due to Theorem \ref{thm3}. 
%Finally 
We note that this simplifies to \eqref{perf_single_normal}
by using \eqref{psi_ik_diagable}, Proposition \ref{Prop_H2L2} and Lemma \ref{prop_obsv_modes}. 
%and noting due to \eqref{nu_eta_kappa} that $\nu_{kk} = 0$ for $k \notin \mathcal{N}_{obsv}$.
\end{proof}
Although Corollary~\ref{cor4} is a special case of Theorem~\ref{thm3}, and consequently Theorem~\ref{thm1}, it generalizes 
%the result given in 
\defaultColor{\cite[Proposition 1]{YoungLeonard2010}}
%in the sense that it allows for computing 
to
performance metrics with arbitrary output matrices. 

\subsection{Double-Integrator Networks}

Next we present the main result of this section for the double-integrator network \eqref{secondmatrix}.

%
%Theorem 4
\begin{thm} [\defaultColor{Double-Integrator, Diagonalizable Laplacian}]
\label{thm4}
Consider the double-integrator network \eqref{secondmatrix}. Suppose that $L$ is diagonalizable. The performance metric $P$ in \eqref{perf_x} is
%
%\begin{equation} \nonumber
%
$
P = \tr{(\Sigma_Q \Psi)}
$,
%
%\end{equation}
%
%and
%
%\begin{equation} \nonumber
%P_{\bold{v}} = \tr{(\Sigma_Q \tilde{\Psi}_{\bold{v}})},
%\end{equation}
%
where
\begin{equation} \label{second_pos_closedform}
\Psi_{kk} =
\nu_{kk} 
\frac{\phi_k}
{2 ( \alpha_k \phi_k^2 + \beta_k \xi_k \phi_k - \beta_k^2 ) }
\end{equation}
for the position-based output, i.e. system $T$ given by \eqref{H_cases_1} and
\begin{equation} \label{second_vel_closedform}
\Psi_{kk} = 
\nu_{kk} 
\frac{ \xi_k \beta_k + \phi_k \alpha_k }
{2 ( \alpha_k \phi_k^2 + \beta_k \xi_k \phi_k - \beta_k^2 ) }
\end{equation}
for the velocity-based output, i.e. system $T$ given by \eqref{H_cases_2}; where $\alpha_k = k_p + \gamma_p \R[\lambda_k]$, $\phi_k = k_d + \gamma_d \R[\lambda_k]$, $\beta_k ~=~ \gamma_p \I[\lambda_k]$ and $\xi_k = \gamma_d\I[\lambda_k]$.
\end{thm}
%
%Remark: No cross terms
\begin{rem}
Here, the cross-terms $\Psi_{kl}$ for $k \neq l$ are not given explicitly for brevity. 
%One can perform a 
%
A Gramian computation as in \cite{PaganiniMallada2017, PaganiniMallada2019} 
%to obtain 
would give
$\Psi_{kl}$ in closed-form for $k \neq l$, which is not tractable due to the number of terms involved. To gain some insight from the computation, we focus on the diagonal terms which are the only required ones when $\Sigma_Q$ in \eqref{Sigma} is diagonal.
\end{rem}
\begin{proof}[Proof of Theorem \ref{thm4}]
The fact that $L$ is diagonalizable leads to $m=n$, i.e. all Jordan blocks are scalars. Then, using \eqref{psi_ik} from Lemma~\ref{lemma2} we have
%
%\begin{equation} \nonumber
$
\Psi_{kl}  = 
\nu_{kl} 
\left \langle \tilde{h}_{11}^{(l)}(t), \tilde{h}_{11}^{(k)}(t) 
\right \rangle_{\mathcal{L}_2}.
$
%\end{equation}
%
First consider the position-based performance metric, i.e. the system $T$ given by \eqref{H_cases_1}. Due to \eqref{h_impResp} from the proof of Theorem~\ref{thm2} and \eqref{phi_TF},
%
%\begin{equation} \nonumber
$
\tilde{h}_{11}^{(k)}(s) = 
\frac{1}
{s^2 + (k_d + \gamma_d \lambda_k) s + k_p + \gamma_p  \lambda_k}
$,
%\end{equation}
%
which has the realization $\left( \mathcal{A}_{k}, \mathcal{B}_{k}, \mathcal{C}_{k} \right)$ in controllable canonical form given by
%
%\begin{equation} \label{cont_canon_1}
$\mathcal{A}_{k} =
\left[
\begin{smallmatrix}
0 & 1 \\
-k_p - \gamma_p  \lambda_k & -k_d - \gamma_d  \lambda_k
\end{smallmatrix}
\right]$,
$\mathcal{B}_{k} = \begin{bmatrix} 0 & 1 \end{bmatrix}^\intercal$ and 
$\mathcal{C}_{k} = \begin{bmatrix} 1 & 0 \end{bmatrix}$.
%\end{equation}
%
If $k=l$, performing a standard computation, 
$\left \langle \tilde{h}_{11}^{(k)}(t), \tilde{h}_{11}^{(k)}(t) \right \rangle_{\mathcal{L}_2} = 
%\tr 
\mathcal{B}_{k}^\intercal \mathcal{X}_k \mathcal{B}_{k}$, where $\mathcal{X}_k$ satisfies the Lyapunov equation
%
%\begin{equation} \nonumber
$
\mathcal{A}_{k}^*\mathcal{X}_k + \mathcal{X}_k \mathcal{A}_{k} = -\mathcal{C}_k^*\mathcal{C}_k
$.
%\end{equation}
Then we get
\begin{equation} \nonumber
\left \langle \tilde{h}_{11}^{(k)}(t), \tilde{h}_{11}^{(k)}(t) \right \rangle_{\mathcal{L}_2} = 
\frac{\phi_k}
{2 ( \alpha_k \phi_k^2 + \beta_k \xi_k \phi_k - \beta_k^2 ) }.
\end{equation}
Considering the velocity-based performance metric, i.e. the system $T$ given by \eqref{H_cases_2} and using \eqref{h_impResp} and \eqref{phi_TF} we have 
$
\tilde{h}_{11}^{(k)}(s) = 
\frac{s}
{s^2 + (k_d + \gamma_d \lambda_k) s + k_p + \gamma_p  \lambda_k}
$,
so that $\mathcal{A}_{k}$ and $\mathcal{B}_{k}$ are the same but $\mathcal{C}_{k} = \begin{bmatrix} 0 & 1 \end{bmatrix}$. If $k=l$, solving the Lyapunov equation leads to
$$
\left \langle \tilde{h}_{11}^{(k)}(t), \tilde{h}_{11}^{(k)}(t) \right \rangle_{\mathcal{L}_2} = 
\frac{\xi_k \beta_k + \phi_k \alpha_k}
{2 ( \alpha_k \phi_k^2 + \beta_k \xi_k \phi_k - \beta_k^2 )}.
%\QEDhere
\quad
\mbox{\qedhere}
\! \! \! \! \!
\! \! \! \! \!
\! \!
%\tag*{\qedhere}
$$
\end{proof}
%}
%
If we further assume 
%that the eigenvalues are real, 
real eigenvalues,
we obtain a result similar to the one in \cite{PaganiniMallada2017, PaganiniMallada2019} for diagonalizable 
%graph 
Laplacians.
%
%Corollary 3
\begin{cor} [\defaultColor{Double-Integrator, Diagonalizable Laplacian with Real Eigenvalues}]
\label{cor3}
Consider the double-integrator network \eqref{secondmatrix}. Suppose that $L$ is diagonalizable and has real eigenvalues. Then
\begin{equation} \label{cor3_psi_pos}
\Psi_{kl} =
\nu_{kl} 
\frac{2k_d + \gamma_d (\lambda_k + \lambda_l)}
{\Psi_{kl}^{denom}},
\end{equation}
for the position-based output, i.e. system $T$ given by \eqref{H_cases_1} and
\begin{equation} \label{cor3_psi_vel}
\Psi_{kl} = 
\nu_{kl} 
\frac
{
(k_p+\gamma_p \lambda_l)(k_d+\gamma_d \lambda_k) +
(k_p+\gamma_p \lambda_k)(k_d+\gamma_d \lambda_l)
}
{\Psi_{kl}^{denom}}
\end{equation}
for the velocity-based output, i.e. system $T$ given by \eqref{H_cases_2}, where
\footnotesize
\begin{align*}
\Psi_{kl}^{denom} \, \, &= \, \,
(k_d+\gamma_d \lambda_k)(k_d+\gamma_d \lambda_l)
(2k_p + \gamma_p (\lambda_k + \lambda_l)) \, \, + \\
\gamma_p^2(\lambda_k - \lambda_l)^2 + 
& (k_p+\gamma_p \lambda_k)(k_d+\gamma_d \lambda_l)^2 +
(k_p+\gamma_p \lambda_l)(k_d+\gamma_d \lambda_k)^2.
\end{align*}
\normalsize
\end{cor}
\begin{proof}
By the argument used in the proof of Theorem \ref{thm4} we have
$
\Psi_{kl}  = 
\nu_{kl} 
\left \langle \tilde{h}_{11}^{(l)}(t), \tilde{h}_{11}^{(k)}(t) 
\right \rangle_{\mathcal{L}_2}
$
and
$\left \langle \tilde{h}_{11}^{(l)}(t), \tilde{h}_{11}^{(k)}(t) \right \rangle_{\mathcal{L}_2} = 
%\tr 
\mathcal{B}_{k}^\intercal \mathcal{X}_{kl} \mathcal{B}_{l}$, 
where $\mathcal{X}_{kl}$ satisfies the Sylvester equation 
$
\mathcal{A}_{k}^*\mathcal{X}_{kl} + \mathcal{X}_{kl} \mathcal{A}_{l} = -\mathcal{C}_k^*\mathcal{C}_l
$
\cite{PaganiniMallada2017, PaganiniMallada2019}.
Considering \eqref{H_cases_1} and \eqref{H_cases_2} individually and solving for $\mathcal{X}_{kl}$ in each case leads to respectively \eqref{cor3_psi_pos} and \eqref{cor3_psi_vel}.
\end{proof}

The real and imaginary parts of the Laplacian eigenvalues, and the control gains appear explicitly in the solutions for the performance metrics in Theorem \ref{thm4} and Corollary \ref{cor3}. 
However, these solutions are still given by a weighted linear combination of $\Psi_{kl}$, therefore 
analyzing the dependence 
%scaling properties of 
of the performance metrics on network topological characteristics as well as control strategy will require further simplifying assumptions. 
\iffalse
For this reason, \blue{we now focus on
%turn our attention to the case of 
normal Laplacians.}
\fi
% 

%Subsection 1
%\subsubsection*{Special Case of Normal Laplacian Matrices}
\subsubsection*{\defaultColor{Normal Laplacian Matrices}}

%The following result concerns the position and velocity performance over digraphs with normal Laplacian matrices.
\defaultColor{As in the case of single-integrator network \eqref{firstmatrix}, this class of graphs provides an insightful example for the upcoming sections.}

%In the following, we focus on systems over digraphs that emit normal weighted Laplacian  
%

%Corollary 5
\begin{cor} [\defaultColor{Double-Integrator, Normal Laplacian}]
\label{cor5}
Consider the double-integrator network \eqref{secondmatrix}. Suppose that $L$ is normal and the input $\bold{w}_0$ 
%is random and spatially uncorrelated, 
\defaultColor{has unit covariance,
i.e. $E[\Sigma_0] = I $}. Then, the expectation of the performance metric $P$ in \eqref{perf_x} is
\begin{equation} \label{perf_x_sol_normal}
E\left[P\right] =
\lVert T \rVert_{\mathcal{H}_2}^2
= \! \!
\sum_{k \in \mathcal{N}_{obsv}}
\!
 \nu_{kk} 
\frac{\phi_k}
{2 ( \alpha_k \phi_k^2 + \beta_k \xi_k \phi_k - \beta_k^2 ) },
\!
\end{equation}
for the position-based output, i.e.
system $T$ given by \eqref{H_cases_1} and
\begin{equation} \label{perf_v_sol_normal}
E\left[P\right] = 
\lVert T \rVert_{\mathcal{H}_2}^2
= \! \!
\sum_{k \in \mathcal{N}_{obsv}} 
\!
\nu_{kk} 
\frac{ \xi_k \beta_k + \phi_k \alpha_k }
{2 ( \alpha_k \phi_k^2 + \beta_k \xi_k \phi_k - \beta_k^2 ) },
\!
\end{equation}
for the velocity-based output, i.e.
system $T$ given by \eqref{H_cases_2};
where $\alpha_k = k_p + \gamma_p \R[\lambda_k]$, $\phi_k = k_d + \gamma_d \R[\lambda_k]$, $\beta_k ~=~ \gamma_p \I[\lambda_k]$ and $\xi_k = \gamma_d\I[\lambda_k]$.
\end{cor}
\begin{proof}
By \eqref{second_pos_closedform} and \eqref{second_vel_closedform} and the same argument used in the proof of Corollary \ref{cor4}, we reach the desired result.
\end{proof}
Note that per 
%\eqref{nu_eta_kappa}, 
Lemma \ref{prop_obsv_modes}
all $\nu_{kk}$ in \eqref{perf_x_sol_normal} and \eqref{perf_v_sol_normal} are positive. 
%therefore the performance measures are guaranteed to be real quantities. 
In addition, stability guarantees that the \defaultColor{numerators and the denominators} in \eqref{perf_x_sol_normal} and \eqref{perf_v_sol_normal} are positive due to Proposition~\ref{prop_stab_rh}. Therefore the performance metrics are guaranteed to be positive quantities as expected. \defaultColor{This result generalizes the result given in \cite[Corollary 2]{OralMalladaGayme2017} 
%in the sense that it allows for computing 
to position and velocity based performance metrics with arbitrary output matrices.}
%

%  

%Subsubsection 1
%\subsubsection

\defaultColor{In the next section, we study 
the effect of communication directionality on performance through the example of 
%the special case of 
normal 
%weighted graph 
Laplacian matrices.}

%SECTION 6 EXAMPLES
%\section{Illustrative Examples} \label{sec6}

%Subsection 1: Examples - Normal Laplacian
%\section{Normal Laplacian Matrices: The Role of Communication  Directionality} 
\section{\defaultColor{The Role of Communication  Directionality}}
%Improving Performance Using Relative Feedback
\label{NormalL}

In this section, we focus on a special class of graphs that emit normal weighted Laplacian matrices \defaultColor{and use the respective results from Section \ref{diagL} to investigate the effect of directed feedback}. This class of graphs can for example arise in spatially invariant systems 
%and vehicle platoons 
\cite{BamiehJovanovic2012, teglingbamieh2019}.
Given any normal weighted  
%and possibly non-symmetric 
Laplacian matrix $L$, we 
%can define
\defaultColor{extract its Hermitian part as}
\begin{equation} \label{herm_part_L}
L^{\prime} := \frac{L + L^*}{2}.
\end{equation}
%
%which is the Hermitian part of $L$.
%
\defaultColor{
%For any given 
%normal and non-symmetric 
Since $L$ is weight-balanced \cite[Lemma 4]{YoungLeonard2010}, \eqref{herm_part_L} gives the Laplacian matrix of an undirected graph $\mathcal{G}^{\prime} = \{\mathcal{N}, \mathcal{E}^{\prime}, \mathcal{W}^{\prime}\}$, where 
$\mathcal{E}^{\prime} = \mathcal{E} \cup \{(j,i) \mid (i,j) \in \mathcal{E} \} $ and
$\mathcal{W}^{\prime} = \{\frac{b_{ij} + b_{ji}}{2} \mid b_{ij} \in \mathcal{W} \}$. Put another way, $\mathcal{G}^{\prime}$ is the \textit{undirected counterpart} of $\mathcal{G}$ resulting from creating reverse edges in $\mathcal{G}$ and re-defining edge weights such that both graphs have the same nodal out-degree.}
%
\iffalse
If in addition
$L$ is non-symmetric, i.e. 
%directed paths exist in 
the corresponding graph $\mathcal{G}$ is directed, the relationship given by \eqref{herm_part_L} 
\textit{symmetrizes}
then leads to a symmetric Laplacian $L^{\prime}$ (whose corresponding graph $\mathcal{G}^{\prime}$ has only undirected paths).
\fi
% 

%Since $L$ is normal,
\defaultColor{Normality of $L$ and
\eqref{herm_part_L} imply that 
the spectrum of $L^\prime$,}
\begin{equation} \label{L_prime_spec}
\spec(L^\prime) = \{ \R[\lambda_i] | \lambda_i \in \spec(L), i=1,\dots, n \}.
\end{equation}
%where $\spec(\cdot)$ denotes matrix spectrum.
\defaultColor{
In addition, since $L$ is normal, it has eigenvalues with non-zero imaginary parts if and only if its graph $\mathcal{G}$ is directed.
%has non-symmetric paths.
%Since 
For disturbance inputs that are uniform and uncorrelated across the network, we observe that
both the position and velocity based performance metrics \eqref{perf_x_sol_normal} and \eqref{perf_v_sol_normal} depend on both the real and imaginary parts of the Laplacian eigenvalues.
%\blue{for disturbance inputs that are uniform across the network.}
Therefore,
comparison of 
%performance over 
directed graphs $\mathcal{G}$ and their undirected counterparts $\mathcal{G}^{\prime}$ 
%that are respectively associated with $L$ and $L^{\prime}$ 
can reveal 
%the effect of 
the interplay between
the imaginary parts, i.e. 
%directed communication.
edge directionality
and control strategy (judicious selection of control gains) that determines overall performance.}
\iffalse
For disturbance inputs that are uniform across the network,  
how performance is simultaneously determined by 
%the interplay between 
edge directionality
%the role of directed communication 
(the role of complex eigenvalues of the Laplacian) and control strategy (judicious selection of control gains).
\fi
%
\iffalse
\begin{rem}

\end{rem}
\fi

%example 1

%Position Performance
\subsection{Position based Performance}

\subsubsection{Single-Integrator Networks}

The following \defaultColor{theorem} provides a comparison of 
%the performance metric given in \eqref{perf_single_normal}.
the single-integrator systems with respective Laplacians $L$ and $L^{\prime}$ \defaultColor{in terms of} the performance metric given in \eqref{perf_single_normal}.
%
%Corollary Single Integrator Position Perf
\begin{thm} [\defaultColor{Equal Performance with Directed Networks and 
%Symmetrized 
Undirected Counterparts}] \label{cor_single_normal}
Consider the single-integrator network \eqref{firstmatrix} and the performance metric $P$ in \eqref{perf_x}. Let $T$ and $T^\prime$ be two systems given by \eqref{H_cases_1} with weighted Laplacian matrices $L$ and $L^\prime$. Suppose $L$ is normal and $L^{\prime}$ is given by \eqref{herm_part_L}.
%and the input $\bold{w}_0$ is random and spatially uncorrelated, i.e. $E[\Sigma_0] = I $. 
Then $\lVert T \rVert_{\mathcal{H}_2}^2 ~=~ \lVert T^\prime \rVert_{\mathcal{H}_2}^2$.
\end{thm}
\begin{proof}
The result follows from \eqref{perf_single_normal} and \eqref{L_prime_spec}.
\end{proof}

As Theorem \ref{cor_single_normal} indicates, directed and \defaultColor{associated undirected} single-integrator systems 
%$T$ and $T^{\prime}$ 
perform identically for any output matrix $C$ satisfying Assumption \ref{assum1}. This implies 
%for circulant $L$ and $L^{\prime}$ 
that the same level of performance can be achieved either using 
%non-symmetric 
\defaultColor{directed}
%cycles (including uni-directional cycles) 
paths
in the commmunication graph or using the corresponding 
%symmetrized 
\defaultColor{undirected}
graph per \eqref{herm_part_L}. 
%In this case, 
\defaultColor{The directed system 
%$T$ 
might be preferable in certain cases due to reduced communication requirements (e.g. uni-directional vs.
%compared to 
bi-directional paths).}
%\cite{oralgaymeacc2019}
%If $T$ and $T^{\prime}$ represent systems with respectively uni-directional and bi-directional feedback, uni-directional feedback performs identically to bi-directional feedback, therefore single directional position measurements are sufficient.

%
%As compared to the literature concerning the performance of single-integrator networks with normal Laplacians, 
Theorem \ref{cor_single_normal} also provides a generalization of previous results obtained for this class of directed and \defaultColor{undirected} single-integrator systems. 
For example, performance of directed systems can be bounded by functions of the spectrums of output performance matrices and associated \defaultColor{undirected} system Laplacians (see e.g. \cite[Theorem 5]{SiamiMotee2017}). Here, we provide exact solutions in Corollary \ref{cor4} by additionally accounting for the eigenvectors of these matrices, which lead to the equivalence between directed and \defaultColor{associated undirected} systems as shown by Theorem \ref{cor_single_normal}.

\subsubsection{Double-Integrator Networks}

We now provide a comparison of the double-integrator systems with respective Laplacians $L$ and $L^{\prime}$ for the performance metric given in \eqref{perf_x_sol_normal}.
%
%Remark4
\begin{rem} \label{rem4}
The performance metric in \eqref{perf_x_sol_normal} 
%is independent of $\I[\lambda_k]$ 
simplifies to an expression that does not explicitly depend on 
$\I[\lambda_k]$ if
%if and only if 
$\beta_k \xi_k \phi_k - \beta_k^2 = 0$ for $k \in \mathcal{N}_{obsv}$. This holds if $\I[\lambda_k] = 0$ for $k \in \mathcal{N}_{obsv}$ or $L$ is symmetric or $\gamma_p = 0$.
If $\beta_k \xi_k \phi_k - \beta_k^2 = 0$ for $k \in \mathcal{N}_{obsv}$, \eqref{perf_x_sol_normal} reduces to
\begin{equation} \label{P_normal_x_noImag}
\lVert T \rVert_{\mathcal{H}_2}^2
%P
\! = \! \! \! \!
\sum_{k \in \mathcal{N}_{obsv}} 
\! \!
\nu_{kk} \frac{1}{2(k_p + \gamma_p\R[\lambda_k]) (k_d + \gamma_d\R[\lambda_k])},
\! \!
\end{equation}
%noting that 
\defaultColor{when the}
stability condition \eqref{stab_rh} from Proposition \ref{prop_stab_rh} holds.
\end{rem}
Depending on the values of $k_p, k_d, \gamma_p$ and $\gamma_d$ in \eqref{P_normal_x_noImag}, the denominator in \eqref{perf_x_sol_normal} can be quadratic in $\R[\lambda_k]$, which could indicate a smaller $\mathcal{H}_2$ norm for sufficiently large $\R[\lambda_k]$, hence better performance compared to the performance of the first order system given by \eqref{perf_single_normal}.

The following Lemma shows the effect of the imaginary parts of the weighted Laplacian eigenvalues on the position based performance \eqref{perf_x_sol_normal} of the double-integrator network \eqref{secondmatrix}.
%Corollary 11
\begin{lemma}[Characterization of Position based Performance via the Observable Eigenvalues] \label{thm_imag}
Consider the double-integrator network \eqref{secondmatrix} and the performance metric $P$ in \eqref{perf_x}. Let $T$ and $T^\prime$ be two systems given by \eqref{H_cases_1} with weighted Laplacian matrices $L$ and $L^\prime$. Suppose $L$ is normal and $L^{\prime}$ is given by \eqref{herm_part_L}. Then the following hold:
%and the input $\bold{w}_0$ is random and spatially uncorrelated, i.e. $E[\Sigma_0] = I $. 
\begin{enumerate}[wide, labelwidth=!, labelindent=0pt]

\item 
$\lVert T \rVert_{\mathcal{H}_2}^2 = \lVert T^\prime \rVert_{\mathcal{H}_2}^2$
if $\I[\lambda_k] = 0 \, \, \forall k \in \mathcal{N}_{obsv}$.
\label{thm_Im_1}

\item \label{thm_Im_2}
$\lVert T \rVert_{\mathcal{H}_2}^2 \leq \lVert T^\prime \rVert_{\mathcal{H}_2}^2$
if 
\begin{equation} \label{cor3_cond1}
\!
\gamma_d (k_d + \gamma_d \R[\lambda_k]) - \gamma_p \geq 0, 
%\quad k = 2, \dots, n,
\ \
\forall k \in \mathcal{N}_{obsv}.
%\, \, 
%\text{and}
%\, \, 
%\I[\lambda_k] \neq 0.
%\! \! \! \!
\end{equation}
%
%for $k \in \mathcal{N}_{obsv},$
%then $\lVert T \rVert_{\mathcal{H}_2}^2 \leq \lVert T^\prime \rVert_{\mathcal{H}_2}^2$. 
Furthermore, $\lVert T \rVert_{\mathcal{H}_2}^2 < \lVert T^\prime \rVert_{\mathcal{H}_2}^2$ if in addition at least one of the inequalities in \eqref{cor3_cond1} strictly holds for some $k \in \mathcal{N}_{obsv}$ such that $\I[\lambda_k] \neq 0$
%the inequality in \eqref{cor3_cond1} strictly holds and $\I[\lambda_k] \neq 0$ 
%for some $k \in \mathcal{N}_{obsv}$
%for some $k \in \{ 2,\dots,n \}$ such that the eigenvector $\boldsymbol{\theta}_k$ of $M$ in \eqref{matrix_M} is observable from the output $\bold{y}$ in \eqref{posOutput} 
and relative position feedback is present, i.e. $\gamma_p \neq 0$.

Similarly, 
$\lVert T \rVert_{\mathcal{H}_2}^2 \geq \lVert T^\prime \rVert_{\mathcal{H}_2}^2$
if 
\begin{equation} \label{cor3_cond2}
\!
\gamma_d (k_d + \gamma_d \R[\lambda_k]) - \gamma_p \leq 0, 
%\quad k = 2, \dots, n,
\ \ 
\forall k \in \mathcal{N}_{obsv}.
%\, \, 
%\text{and}
%\, \, 
%\I[\lambda_k] \neq 0.
%\! \! \! \!
\end{equation}
%
%for $k \in \mathcal{N}_{obsv},$ 
%then $\lVert T \rVert_{\mathcal{H}_2}^2 \geq \lVert T^\prime \rVert_{\mathcal{H}_2}^2$. 
Furthermore $\lVert T \rVert_{\mathcal{H}_2}^2 ~>~ \lVert T^\prime \rVert_{\mathcal{H}_2}^2$ if in addition at least one of the inequalities in \eqref{cor3_cond2} strictly holds for some $k \in \mathcal{N}_{obsv}$ such that $\I[\lambda_k] \neq 0$
%the inequality in \eqref{cor3_cond2} strictly holds and $\I[\lambda_k] \neq 0$ 
%for some $k \in \mathcal{N}_{obsv}$
%for some $k \in \{ 2,\dots,n \}$ such that the eigenvector $\boldsymbol{\theta}_k$ of $M$ in \eqref{matrix_M} is observable from the output $\bold{y}$ in \eqref{posOutput} 
and relative position feedback is present, i.e. $\gamma_p \neq 0$.

%Furthermore, suppose that the eigenvalues of L have non-zero imaginary parts,
%(hence $\mathcal{G}$ is directed)
%i.e. $\I[\lambda_k] \neq 0$,
%Furthermore, suppose that $\I[\lambda_k] \neq 0$ and
%the performance measure is such that the eigenvalues $\mu_k$ of $M$ in \eqref{matrix_M} satisfy $\mu_k \neq 0$ for some $k \in \{ 2, \dots, n \}$ and relative position feedback is present, i.e. $\gamma_p \neq 0$.

%Then if \eqref{cor3_cond1} strictly holds, $\lVert T \rVert_{\mathcal{H}_2}^2 < \lVert T^\prime \rVert_{\mathcal{H}_2}^2$. Similarly if \eqref{cor3_cond2} strictly holds then $\lVert T \rVert_{\mathcal{H}_2}^2 > \lVert T^\prime \rVert_{\mathcal{H}_2}^2$.
%
%\begin{equation}
%\lVert G_{\Pi} \rVert_2^2 \leq \lVert G_{\Pi}^\prime \rVert_2^2.
%\end{equation}
%
\end{enumerate}
\end{lemma}
\begin{proof}
Invoking Remark \ref{rem4} and using \eqref{L_prime_spec}, both $\lVert T \rVert_{\mathcal{H}_2}^2$ and $\lVert T^\prime \rVert_{\mathcal{H}_2}^2$ are given by \eqref{P_normal_x_noImag} which leads to Item \textit{\ref{thm_Im_1})}. 
Condition \eqref{cor3_cond1} implies that $\beta_k \xi_k \phi_k - \beta_k^2 \geq 0$ for $k \in \mathcal{N}_{obsv}$ therefore
\begin{equation} \label{cor3_ineq}
%\! 
\frac{\phi_k}{2 ( \alpha_k \phi_k^2 + \beta_k \xi_k \phi_k - \beta_k^2 ) } \leq \frac{1}{2\alpha_k \phi_k }, \ \ 
%k = 2, \dots, n. 
k \in \mathcal{N}_{obsv}.
\!
\end{equation}
Since $\nu_{kk} > 0$ for $k \in \mathcal{N}_{obsv}$ due to Lemma \ref{prop_obsv_modes}, multiplication of both sides of \eqref{cor3_ineq} by $\nu_{kk}$ and summation of the inequalities gives $\lVert T \rVert_{\mathcal{H}_2}^2 \leq \lVert T^\prime \rVert_{\mathcal{H}_2}^2$. 
%To prove the sufficient condition for the strict inequality, 
%we observe that due to Assumption \ref{assum1} and the fact that  $L$ is normal and $M$ is symmetric,
%and $\mathcal{G}$ has a globally reachable node, 
%it holds that $\bold{r}_k, \boldsymbol{\theta}_l \in \spn \{ \boldsymbol{1} \} ^ {\bot}$ for $k, l \in~ \{ 2,\dots,n \}$. Therefore, given $k \in \{ 2,\dots,n \}$, $| \langle \bold{r}_{k}, \boldsymbol{\theta}_l  \rangle |^2 \neq 0$ for some $l \in \{ 2,\dots,n \}$. Then 
%for any $k \in \{ 2,\dots,n \}$, 
%$\nu_{kk} > 0$ for $k = 2, \dots, n$ due to \eqref{nu_eta_kappa}. 
If in addition to \eqref{cor3_cond1} at least one of these inequalities strictly holds for some $k \in \mathcal{N}_{obsv}$ such that $\I[\lambda_k] \neq 0$ and $\gamma_p \neq 0$, then
$\lVert T \rVert_{\mathcal{H}_2}^2 < \lVert T^\prime \rVert_{\mathcal{H}_2}^2$.
The reverse inequalities follow from \eqref{cor3_cond2} using a similar argument.
\end{proof}

Note that the results in Lemma \ref{thm_imag} hold for any output matrix $C$ satisfying Assumption \ref{assum1}. It is necessary that at least one observable eigenvalue does not lie on the real line for the performance of the directed and \defaultColor{undirected} systems 
%$T$ and $T^\prime$ 
to differ, and the gains need to be tuned based on \defaultColor{these} eigenvalues to improve performance. We next use this result to characterize the position-based performance of directed and \defaultColor{undirected} double-integrator systems 
%$T$ and $T^{\prime}$ 
%given in Lemma \ref{thm_imag} 
in terms of relative feedback.
%provides sufficient conditions on the control gains $k_p, k_d, \gamma_p$ and $\gamma_d$ that lead to the possible relationships between the position based performance of the systems $T$ and $T^{\prime}$ given in Lemma \ref{cor11}.

\begin{thm}[Characterization of Position based Performance via Relative Feedback]
\label{thm_gp}
Consider the double-integrator network \eqref{secondmatrix} and the performance metric $P$ in \eqref{perf_x}. Let $T$ and $T^\prime$ be two systems given by \eqref{H_cases_1} with weighted Laplacian matrices $L$ and $L^\prime$. Suppose that $L$ is normal
%its eigenvalues have non-zero imaginary parts,
%(hence $\mathcal{G}$ is directed)
%i.e. 
and $L^{\prime}$ is given by \eqref{herm_part_L}.
%and the input $\bold{w}_0$ is random and spatially uncorrelated, i.e. $E[\Sigma_0] = I $.
Then the following hold:
\begin{enumerate}[wide, labelwidth=!, labelindent=0pt]

\item If relative position feedback is absent, i.e. $\gamma_p = 0$, then $\lVert T \rVert_{\mathcal{H}_2}^2 = \lVert T^\prime \rVert_{\mathcal{H}_2}^2$.
%by Remark \ref{rem4}. 
\label{thm_gp_1}

\item If relative position feedback is present and relative velocity feedback is absent, i.e. $\gamma_p \neq 0$ and $\gamma_d = 0$, and
$\I[\lambda_k] \neq 0$ for some $k \in \mathcal{N}_{obsv}$,
then 
$\lVert T \rVert_{\mathcal{H}_2}^2 ~>~ \lVert T^\prime \rVert_{\mathcal{H}_2}^2$.
%the results of Corollary \ref{cor11} have the following implications:
\label{thm_gp_2}

\item 
%Suppose 
%
If both relative position and velocity feedback are present, i.e. $\gamma_p \neq 0$ and $\gamma_d \neq 0$, and
$\I[\lambda_k] \neq 0$ for some $k \in \mathcal{N}_{obsv}$,
%Then there exists $\tilde{\gamma}_p$ that satisfies
then there exists 
$\underline{\gamma}_p$ and $\overline{\gamma}_p$ that satisfy
$$
\min\limits_{\substack{k \in \mathcal{N}_{obsv}, \\ \I[\lambda_k] \neq 0}}
\! \!
\R[\lambda_k]
\, \,
\leq
\, \,
\frac{\underline{\gamma}_p}{\gamma_d^2} - \frac{k_d}{\gamma_d}
\, \,
\leq
\, \,
\frac{\overline{\gamma}_p}{\gamma_d^2} - \frac{k_d}{\gamma_d}
\, \,
\leq
\, \,
\max\limits_{\substack{k \in \mathcal{N}_{obsv}, \\ \I[\lambda_k] \neq 0}}
\! \!
\R[\lambda_k],
$$
%such that
%$\lVert T \rVert_{\mathcal{H}_2}^2 \! < (>) \, \lVert T^\prime \rVert_{\mathcal{H}_2}^2$ if and only if  
%$\gamma_p \! < (>) \, \tilde{\gamma}_p$.
such that
$\lVert T \rVert_{\mathcal{H}_2}^2 < \lVert T^\prime \rVert_{\mathcal{H}_2}^2$ if 
$\gamma_p < \underline{\gamma}_p$
and
$\lVert T \rVert_{\mathcal{H}_2}^2 > \lVert T^\prime \rVert_{\mathcal{H}_2}^2$ if 
$\gamma_p > \overline{\gamma}_p$.
%where
%$\tilde{\gamma}_{p,k}
%=\gamma_d (k_d + \gamma_d \R[\lambda_k])$.
%Furthermore, 
%$\lVert T \rVert_{\mathcal{H}_2}^2 ~>~ \lVert T^\prime \rVert_{\mathcal{H}_2}^2$ if and only if 
%$\gamma_p > \tilde{\gamma}_p$.
\label{thm_gp_3}

%\item If relative position feedback is %present, 
%$$
%
%$$
%and $k_d$ or $\gamma_d$ are non-zero.
%
\begin{proof}
Invoking Remark \ref{rem4} and using \eqref{L_prime_spec} leads to Item \textit{\ref{thm_gp_1})}. Item \textit{\ref{thm_gp_2})} follows from Lemma \ref{thm_imag} by setting $\gamma_p \neq 0$ and $\gamma_d = 0$ in \eqref{cor3_cond2}. 
To prove Item \textit{\ref{thm_gp_3})} we observe from Lemma \ref{thm_imag} that
\begin{gather*}
\gamma_p >
\max\limits_{\substack{k \in \mathcal{N}_{obsv}, \\ \I[\lambda_k] \neq 0}}
\gamma_d (k_d + \gamma_d \R[\lambda_k])
=: \gamma_{u}
%\quad
\Rightarrow
%\quad
\lVert T \rVert_{\mathcal{H}_2}^2 > \lVert T^\prime \rVert_{\mathcal{H}_2}^2,
\\
\gamma_p <
\min\limits_{\substack{k \in \mathcal{N}_{obsv}, \\ \I[\lambda_k] \neq 0}}
\gamma_d (k_d + \gamma_d \R[\lambda_k])
=: \gamma_{l}
%\quad
\Rightarrow
%\quad
\lVert T \rVert_{\mathcal{H}_2}^2 < \lVert T^\prime \rVert_{\mathcal{H}_2}^2.
\end{gather*}
So 
$
\lVert T \rVert_{\mathcal{H}_2}^2 = \lVert T^\prime \rVert_{\mathcal{H}_2}^2
$ 
if
$\gamma_p = \underline{\gamma}_p$
and 
$
\lVert T \rVert_{\mathcal{H}_2}^2 < \lVert T^\prime \rVert_{\mathcal{H}_2}^2
$
if
$\gamma_p < \underline{\gamma}_p$
for some 
$\underline{\gamma}_p \in [\gamma_l, \gamma_u]$,
since 
$
\lVert T \rVert_{\mathcal{H}_2}^2$
and $\lVert T^\prime \rVert_{\mathcal{H}_2}^2
$ are
continuous in $\gamma_p$.
Similarly,
$
\lVert T \rVert_{\mathcal{H}_2}^2 = \lVert T^\prime \rVert_{\mathcal{H}_2}^2
$ 
if
$\gamma_p = \overline{\gamma}_p$
and 
$
\lVert T \rVert_{\mathcal{H}_2}^2 > \lVert T^\prime \rVert_{\mathcal{H}_2}^2
$
if
$\gamma_p > \overline{\gamma}_p$
for some 
$\overline{\gamma}_p \in [\gamma_l, \gamma_u]$.
%since 
%$
%\lVert T \rVert_{\mathcal{H}_2}^2$
%and $\lVert T^\prime \rVert_{\mathcal{H}_2}^2
%$
%continuous in $\gamma_p$.
Finally we note that 
$\underline{\gamma}_p \leq \overline{\gamma}_p$, because otherwise 
$\gamma_p = \underline{\gamma}_p 
> \overline{\gamma}_p
$
would imply that
$
\lVert T \rVert_{\mathcal{H}_2}^2 = \lVert T^\prime \rVert_{\mathcal{H}_2}^2
$
and
$
\lVert T \rVert_{\mathcal{H}_2}^2 > \lVert T^\prime \rVert_{\mathcal{H}_2}^2
$
must simultaneously hold, which is a contradiction.
\end{proof}

\end{enumerate}
\end{thm}

%Remark position based performance
%old version
%\begin{rem} \label{normal_pos_perf_summary}

%\end{enumerate}

%\end{enumerate}

%\end{rem}

%
Directed communication degrades performance for 
%performance 
metrics that capture some of the modes resulting from the directed 
%cycles 
\defaultColor{paths}
(i.e. $\I[\lambda_k] \neq 0$ for some $k \in \mathcal{N}_{obsv}$) if relative position feedback is used 
%in the absence of 
without
relative velocity feedback. For such metrics, this issue can be addressed in several ways depending on the available feedback.
%then the %symmetrized 
%\blue{bi-directed}
%system $T^\prime$ outperforms $T$. 
For example, omitting relative position feedback (which requires absolute position feedback due to Assumption~\ref{assum_gain}) can mitigate this degradation.
%Unless relative position feedback is used, 
In this case, the directionality of relative velocity feedback does not affect performance 
since 
%in this case 
directed and \defaultColor{undirected} systems 
%$T$ and $T^\prime$ 
perform identically.

%For such measures, 
It is when both types of relative feedback are used that tuning their respective gains properly 
%determines how these systems' performance compare. 
can, not only mitigate the performance degradation, but also lead to the directed system outperforming its \defaultColor{undirected} counterpart. Therefore, it is critical to have relative velocity feedback in addition to relative position feedback.
%as otherwise directed communication degrades performance.
Namely, 
%$T$ 
the directed system  
performs better than 
%$T^\prime$ 
its \defaultColor{undirected} counterpart
for sufficiently small 
%$\gamma_p$ 
relative position gain
(the converse is true for sufficiently large 
%$\gamma_p$
relative position gain). This sufficient magnitude is determined by the velocity gains as well as the magnitude of the real parts of the observable eigenvalues that have non-zero imaginary parts. As a consequence, a judicious control strategy depends on the topological characteristics of the network.
%Furthermore for given velocity gains, the sufficient magnitude of $\gamma_p$ is determined by 
%the minimum (or the maximum) magnitude of the real parts of the observable eigenvalues that have non-zero imaginary parts.

%
% FIGURE
%Deviation from the average 

%}
%\label{DavLoc_PosVel_Normal}
%\end{figure*}
%

\iffalse
%
\begin{figure*}[t]
\centering
%\vspace{-16pt}
\subfloat
{
       \includegraphics[width=0.28\linewidth]{pos_no_gd}
       }
\subfloat
{
       \includegraphics[width=0.28\linewidth]{pos_yes_gd}
       }
\subfloat
{
       \includegraphics[width=0.28\linewidth]{vel_yes_gd}
       }
       
\vspace{-6pt}       

    \caption{ 
    The expectation of the position based performance of the double-integrator system \eqref{secondmatrix} given by $\eqref{H_cases_1}$, for $E\left[ \Sigma_0 \right] = I$ and the gains
    \textbf{(a)}
    $k_p = 3, k_d = 5, \gamma_d = 0$
    \textbf{(b)}
    $k_p = 1, k_d = 2, \gamma_d = 6.5$.
    $\bold{(c)}$
    The velocity based performance of the double-integrator system \eqref{secondmatrix} given by $\eqref{H_cases_2}$, for $E\left[ \Sigma_0 \right] = I$ and the gains
    $k_p = 1, k_d = 2, \gamma_d = 7$. 
    }
    \label{fig_section6}
\vspace{-10pt}
\end{figure*}
%
\fi

%Velocity Performance
%\subsection{Velocity based Performance of Double Integrator Networks}
\subsection{Velocity based Performance}

%The following 
\defaultColor{This subsection}
provides a comparison of the double integrator systems with respective Laplacians $L$ and $L^{\prime}$ \defaultColor{in terms of} the performance metric given in \eqref{perf_v_sol_normal}.
%
%Remark5
\begin{rem} \label{rem5}
The performance metric in \eqref{perf_v_sol_normal} 
%is independent of $\I[\lambda_k]$ 
simplifies to an expression that does not explicitly depend on 
$\I[\lambda_k]$ if
%if and only if 
%$\beta_k \xi_k = 0$ and $\beta_k \xi_k \phi_k - \beta_k^2 = 0$ 
$\beta_k = 0$~
for 
%$k = 2, \dots, n$. 
$k \in \mathcal{N}_{obsv}$.
This holds if $\I[\lambda_k] = 0$ for $k \in \mathcal{N}_{obsv}$ or $L$ is symmetric or $\gamma_p = 0$.
If 
%$\beta_k \xi_k = 0$ and $\beta_k \xi_k \phi_k - \beta_k^2 = 0$ 
$\beta_k = 0$
for $k \in \mathcal{N}_{obsv}$, \eqref{perf_v_sol_normal}~ reduces to
\begin{equation} \label{P_normal_v_noImag}
\lVert T \rVert_{\mathcal{H}_2}^2
=
\sum_{k \in \mathcal{N}_{obsv}} \nu_{kk} \frac{1}{2 (k_d + \gamma_d \R[\lambda_k])},
\end{equation}
%noting that 
\defaultColor{when the}
stability condition \eqref{stab_rh} from Proposition \ref{prop_stab_rh} holds.
\end{rem}

In contrast to the position based performance metric in \eqref{P_normal_x_noImag}, the velocity based performance in \eqref{P_normal_v_noImag} depends only on absolute or relative velocity feedback and 
%has a 
its
denominator 
%that 
is affine in $\R[\lambda_k]$. So, absolute or relative position feedback does not affect velocity based performance if $\mathcal{G}$ is undirected.

The following theorem demonstrates that if the velocity based performance of the system given by $\eqref{H_cases_2}$ is considered and its directed graph emits a normal weighted Laplacian, its $\mathcal{H}_2$ norm is lower bounded by the $\mathcal{H}_2$ norm of the corresponding 
%symmetrized 
%bi-directed
\defaultColor{undirected}
system 
%$\eqref{H_cases_2}$ 
whose interconnection is defined by \eqref{herm_part_L}.
This \defaultColor{result} highlights the inability of standard feedback schemes to mitigate velocity-based performance degradation.

%Corollary 12
\begin{thm} [Characterization of Velocity based Performance]
\label{thm_vel}
Consider the double-integrator network \eqref{secondmatrix} and the performance metric $P$ in \eqref{perf_x}. Let $T$ and $T^\prime$ be two systems given by \eqref{H_cases_2} with weighted Laplacian matrices $L$ and $L^\prime$. Suppose that $L$ is normal and $L^{\prime}$ is given by \eqref{herm_part_L}.
%and the input $\bold{w}_0$ is random and spatially uncorrelated, i.e. $E[\Sigma_0] = I $. 
Then the following hold:
\begin{enumerate}[wide, labelwidth=!, labelindent=0pt]

\item
\label{thm_vel_1}
$\lVert T \rVert_{\mathcal{H}_2}^2 
\geq 
\lVert T^\prime \rVert_{\mathcal{H}_2}^2$.

\item
\label{thm_vel_3}
$\lVert T \rVert_{\mathcal{H}_2}^2 > \lVert T^\prime \rVert_{\mathcal{H}_2}^2$ if and only if $\I[\lambda_k] \neq 0$ for some 
$k \in \mathcal{N}_{obsv}$
%$k \in \{ 2,\dots,n \}$ such that the eigenvector $\boldsymbol{\theta}_k$ of $M$ in \eqref{matrix_M} is observable from the output $\bold{y}$ in \eqref{velOutput} 
and relative position feedback is present, i.e. $\gamma_p ~\neq~ 0$.

\item
\label{thm_vel_2}
$\lVert T \rVert_{\mathcal{H}_2}^2 \!= \lVert T^\prime \rVert_{\mathcal{H}_2}^2$
if and only if $\I[\lambda_k] = 0 \, \, \forall k \in \mathcal{N}_{obsv}$ or relative position feedback is absent, i.e. $\gamma_p = 0$.
%then 
%
%$\lVert T \rVert_{\mathcal{H}_2}^2 ~=~ \lVert T^\prime \rVert_{\mathcal{H}_2}^2$.
%
\end{enumerate}
\end{thm}
\begin{proof}
Since $-\beta_k^2 = - \gamma_p^2 \I[\lambda_k]^2 \leq 0$, it holds that
\begin{equation} \label{cor12_ineq}
\alpha_k \phi_k^2 + \beta_k \xi_k \phi_k - \beta_k^2 \leq \alpha_k \phi_k^2 + \beta_k \xi_k \phi_k, \ \ 
%k=2,\dots,n.
k \in \mathcal{N}_{obsv}.
\end{equation}
Stability condition \eqref{stab_rh} from Proposition \ref{prop_stab_rh} states that 
\begin{equation} \label{stab_vel_perf_thm}
\alpha_k \phi_k^2 + \beta_k \xi_k \phi_k - \beta_k^2 > 0 \ \text{and} \ \phi_k > 0, \ \ 
%k=2,\dots,n.
k \in \mathcal{N}_{obsv}.
\end{equation}
Therefore, \eqref{cor12_ineq} can be re-arranged as 
\begin{equation} \label{vel_perf_ineq}
\frac{ \xi_k \beta_k + \phi_k \alpha_k }
{ \alpha_k \phi_k^2 + \beta_k \xi_k \phi_k - \beta_k^2  }
\geq \frac{1}{\phi_k},
\ \
%k=2,\dots,n.
k \in \mathcal{N}_{obsv}.
\end{equation}
Since $\nu_{kk} > 0$ for $k \in \mathcal{N}_{obsv}$ as shown in Lemma \ref{prop_obsv_modes}, 
\begin{equation} \label{cor12_ineq_2}
\nu_{kk}
\frac{ \xi_k \beta_k + \phi_k \alpha_k }
{2( \alpha_k \phi_k^2 + \beta_k \xi_k \phi_k - \beta_k^2 ) }
\geq 
\nu_{kk}
\frac{1}{2 \phi_k},
\ \
%k=2,\dots,n.
k \in \mathcal{N}_{obsv}.
\end{equation}
Summation of the inequalities given in \eqref{cor12_ineq_2} and using \eqref{perf_v_sol_normal} and \eqref{P_normal_v_noImag} leads to \defaultColor{Item \textit{\ref{thm_vel_1})}}. 

To prove the necessity part of \defaultColor{Item \textit{\ref{thm_vel_3})}}, we observe that 
$-\beta_k^2 = - \gamma_p^2 \I[\lambda_k]^2 < 0$ for some 
$k \in \mathcal{N}_{obsv}$ therefore
\eqref{cor12_ineq} strictly holds for such $k$. Then by a similar argument to the one used above, \eqref{cor12_ineq_2} strictly holds for such $k$ as well, which leads to 
$\lVert T \rVert_{\mathcal{H}_2}^2 > \lVert T^\prime \rVert_{\mathcal{H}_2}^2$. \defaultColor{To prove 
%the 
sufficiency} 
%part 
suppose that
$\lVert T \rVert_{\mathcal{H}_2}^2 > \lVert T^\prime \rVert_{\mathcal{H}_2}^2$. Using \eqref{perf_v_sol_normal} and \eqref{P_normal_v_noImag}, this implies that \eqref{cor12_ineq_2} strictly holds for some $k \in \mathcal{N}_{obsv}$ (otherwise 
$\lVert T \rVert_{\mathcal{H}_2}^2 = \lVert T^\prime \rVert_{\mathcal{H}_2}^2$). Since $\nu_{kk} > 0$ for $k \in \mathcal{N}_{obsv}$, 
%the argument used above 
\eqref{vel_perf_ineq} strictly holds for some $k \in \mathcal{N}_{obsv}$ as well. Using \eqref{stab_vel_perf_thm} and re-arranging terms
leads to $\beta_k^2 = \gamma_p^2 \I[\lambda_k]^2 > 0$ for some $k \in \mathcal{N}_{obsv}$ implying that $\I[\lambda_k] \neq 0$ for some $k \in \mathcal{N}_{obsv}$ and $\gamma_p \neq 0$.
%
%Invoking Remark \ref{rem5} and using \eqref{L_prime_spec}, both $\lVert T \rVert_2^2$ and $\lVert T^\prime \rVert_2^2$ are given by \eqref{P_normal_v_noImag} which leads to Item \ref{thm_vel_2}.
%
Finally we note that \defaultColor{items \textit{\ref{thm_vel_1})} and \textit{\ref{thm_vel_3})} imply Item \textit{\ref{thm_vel_2})}.}
\end{proof}

Unlike position based performance, there does not exist a choice of control gains for the directed system that can result in better velocity based performance compared to 
%system $T^{\prime}$ 
its \defaultColor{undirected} counterpart
for any output matrix $C$ satisfying Assumption \ref{assum1}. 
Furthermore, when relative position feedback is used,
%means that 
the directed system 
%has strictly worse performance 
performs strictly worse
compared to 
%$T^\prime$ 
its \defaultColor{undirected} counterpart
for 
%performance 
metrics capturing the effect of 
%directed relative state measurements. 
the directed interconnection.
They perform identically without relative position feedback or if 
%performance 
metrics do not capture the edge directionality.
%of relative feedback.

When the overall system performance is considered in terms of both position and velocity based metrics, a trade-off emerges.
For systems with observable directed paths, it is possible to have equal performance to that of their
%symmetrized 
\defaultColor{undirected}
counterparts in the case of both position and velocity based metrics by omitting relative position feedback. But this is true only if absolute position feedback is used, as 
it is required for stability (Assumption \ref{assum_gain}).
%otherwise the system becomes unstable due to Proposition \ref{prop_stab_rh}. 
Therefore, unless absolute position measurements are available, the directed system requires well-tuned gains to prevent degradation of the position-based performance (or to possibly improve it) while it will always have worse velocity-based performance compared to the 
%symmetrized 
\defaultColor{undirected}
system. For directed systems with absolute position feedback, improving position-based performance comes at the expense of the velocity-based performance.
\begin{rem}
For the particular metric defined as the variance of the full-state, the $\mathcal{H}_2$ norm of a linear system can be upper bounded by the $\mathcal{H}_2$ norm of  a system whose dynamics emit the Hermitian part of the original state matrix \cite[Theorem~2]{SiamiMotee2017}. In the case of double-integrator networks, this comparison does not 
%necessarily 
\defaultColor{explicitly}
account for \defaultColor{the Laplacian eigenvalues, i.e.} communication directionality. In contrast,
%to this result,
%In this section, 
we have studied communication directionality for general quadratic metrics by comparing directed graphs and their \defaultColor{undirected} counterparts represented by the Hermitian part of the Laplacian \eqref{herm_part_L}. 
%Our results show that judicious control strategy determines  
Our results characterize performance as an aggregate outcome of judicious control strategy and network topology.
%
%In relation to the upper bound obtained for the special case of full-state performance of a linear system (Theorem 2 in \cite{SiamiMotee2017}), For double-integrator systemsThis is different the performance bound provided 
\end{rem}

%
%
%

%
%\begin{proof}

%

%

%
%
\begin{figure*}[t!]
\centering
%\vspace{-16pt}
\subfloat
{
       \includegraphics[width=0.28\linewidth]{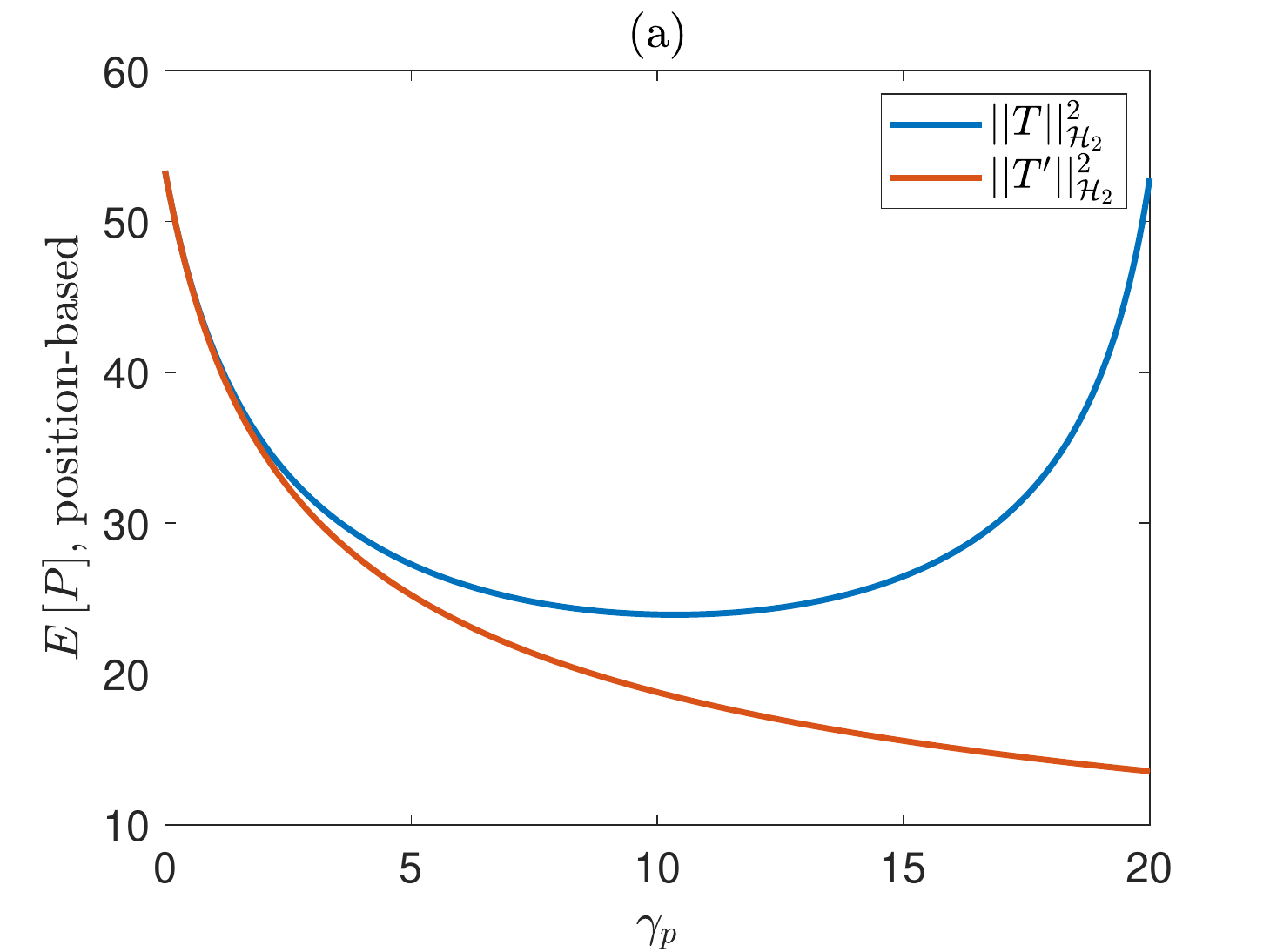}
       }
\subfloat
{
       \includegraphics[width=0.28\linewidth]{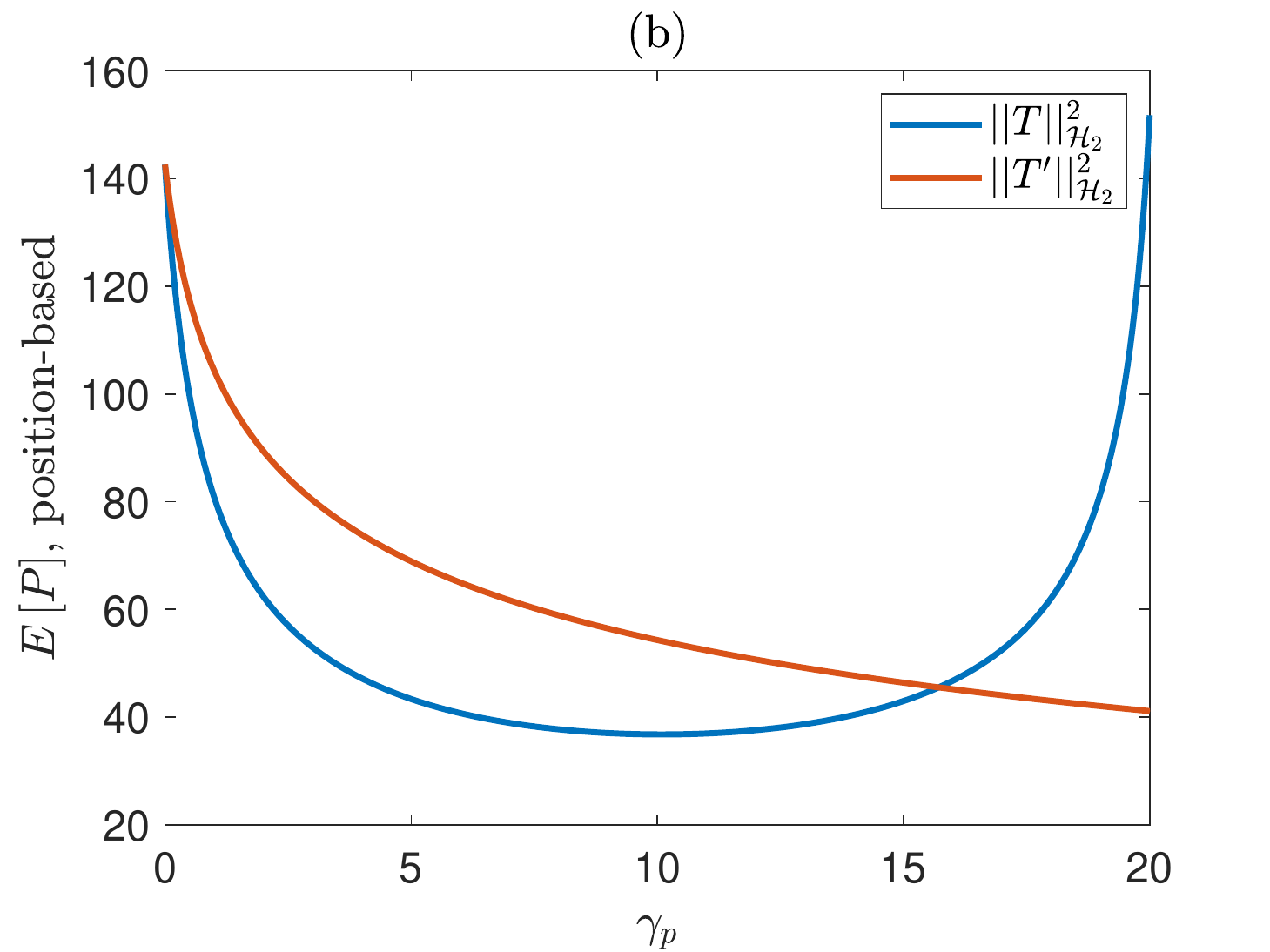}
       }
\subfloat
{
       \includegraphics[width=0.28\linewidth]{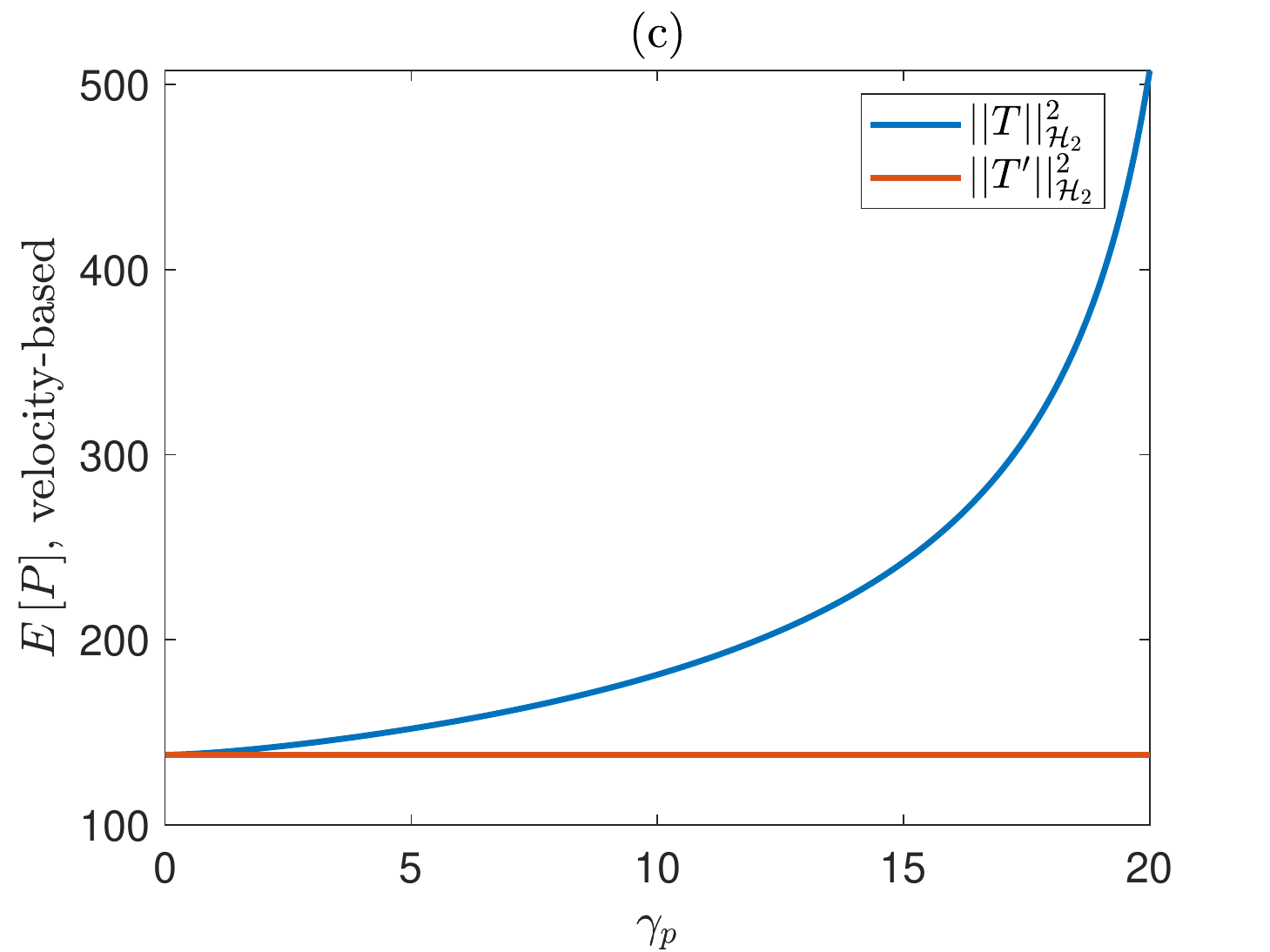}
       }
       
\vspace{-6pt}       

    \caption{ 
    The expectation of the position-based performance of the double-integrator system \eqref{secondmatrix} given by $\eqref{H_cases_1}$, for $E\left[ \Sigma_0 \right] = I$ and the gains
    \mbox{\textbf{(a)}
    $k_p = 3, k_d = 5, \gamma_d = 0$,}
    \textbf{(b)}
    %$\bold{(b)}$
    $k_p = 1, k_d = 2, \gamma_d = 6.5$.
    %$\bold{(c)}$
    \textbf{(c)}
    The expectation of the velocity-based performance of the double-integrator system \eqref{secondmatrix} given by $\eqref{H_cases_2}$, for $E\left[ \Sigma_0 \right] = I$ and the gains
    $k_p = 1, k_d = 2, \gamma_d = 7$. 
    }
    \label{fig_section6}
\vspace{-10pt}
\end{figure*}
%
%

%subsection Numerical Simulations
\subsection{Example: Position and Velocity based Performance with Uni-directional vs. Bi-directional Feedback}

%Now the results stated above will be demonstrated through example of uni-directional and bi-directional feedback and performance measures that quantify the coherence of position and velocity. 

We now consider a cyclic digraph in which each node has uniform out-degree $d$ and the uniformly weighted edges that start at each node reach $\omega$ succeeding nodes. This results in `look-ahead' type state measurements through $\omega$ communication hops. The respective weighted Laplacian is given by
\begin{equation} \label{cycle_L}
%
\iffalse
\small
\! \!
L^{cyc}(d,\omega) \! = \!
d \! \!
\begin{bmatrix*}[r]
1 & -\frac{1}{\omega} & \dots & -\frac{1}{\omega} & 0 & \dots & 0 \\
0 & 1 & -\frac{1}{\omega} & \dots & -\frac{1}{\omega} & \ddots & \vdots \\
\vdots & \ddots & \ddots & \ddots & \ddots & \ddots & 0 \\
0 & \ddots & 0 & 1 &  -\frac{1}{\omega} & \ddots & -\frac{1}{\omega} \\
-\frac{1}{\omega} & \ddots & \ddots & \ddots & \ddots & \ddots & \vdots \\
\vdots & \ddots & \ddots & \ddots & 0 & \ddots & -\frac{1}{\omega}\\
-\frac{1}{\omega} & \dots & -\frac{1}{\omega} & 0 & \dots & 0 & 1 
\end{bmatrix*} \! \!, \! \! \! \! \!
\normalsize
\fi
%
%
\! \,
L^{cyc}(d,\omega) \! = \! 
d \! \times \! \circulant{ \left( 
\begin{bmatrix}
1 & -\frac{1}{\omega} & \dots & -\frac{1}{\omega} & 0 & \dots & 0 
\end{bmatrix}
\right) } \!, \! \! \! \! \!
\end{equation}
where $d \in \mathbb{R}^{+}$, $\omega \in \mathbb{Z}^{+}$, $\omega \leq n-1$ 
and $\circulant(\cdot)$ denotes the circulant matrix generated by permuting the row vector in the argument. The Jordan decomposition of $L = L^{cyc}$ gives \cite{HornJohnson}
\begin{equation} \label{jordan_cycle}
J_k = \lambda_k = 
d \left( 1 - \frac{1}{\omega} \sum_{i=1}^{\omega}
e^{-\bold{j} \frac{2 \pi}{n} i (k-1)} \right),
\end{equation}
for $k = 1, \dots, n$. Choosing $\alpha = \frac{1}{\sqrt{n}}$ in \eqref{RandQ}, the columns of $\tilde{R}$ are given by
\begin{align}
\bold{r}_{l} &= \frac{1}{\sqrt{n}}
\begin{bmatrix} 
1 & e^{\bold{j} \frac{2 \pi}{n} (l-1)} & \dots & e^{\bold{j} \frac{2 \pi}{n} (l-1) (n-1)}
\end{bmatrix} ^*, \label{r_cycle} 
%\\
%\bold{q}_{l} &= \bold{r}_{l}^* \nonumber
\end{align}
for $l=2,\dots,n$. For the special case of uni-directional feedback, we set $d=1$ and $\omega =1$ in \eqref{cycle_L} therefore
%\eqref{uni_Lap}, i.e.
\iffalse 
$$
L = L^{cyc}(1,1).
$$
\fi
$$
L = L^{cyc}(1,1) \quad
\text{and} \quad
L^\prime = 
\frac
{L^{cyc}(1,1) + L^{cyc}(1,1)^*}
{2},
$$  
%
%We refer to the corresponding 'symmetrized' feedback as 
where we have used \eqref{herm_part_L} to also define the corresponding
bi-directional feedback. 
%is determined by \eqref{herm_part_L} which gives
\iffalse
$$
L^\prime = 
\frac
{L^{cyc}(1,1) + L^{cyc}(1,1)^*}
{2}.
$$
\fi
%
%Using this decomposition, the following corollary provides the solution for \eqref{perf_x} for the single-integrator network \eqref{firstmatrix} using Corollary \ref{cor4}.
%
We consider the respective systems $T$ and $T^\prime$ with an 
%randomly selected 
%random
arbitrary
output matrix $C \in \mathbb{R}^{n \times n}$ 
that satisfies Assumption \ref{assum1},
for $n=50$.

For the double-integrator network \eqref{secondmatrix} given by $\eqref{H_cases_1}$ (position based performance),
%and measures $P_{dav}$ and $P_{loc}$ defined by \eqref{C_dav} and \eqref{C_loc}, 
Figure \ref{fig_section6}a shows 
that, as suggested by
\defaultColor{Item \textit{\ref{thm_gp_2})}} of Theorem \ref{thm_gp}, using relative position feedback without relative velocity feedback ($\gamma_p \neq 0$ and $\gamma_d = 0$) 
%$\lVert T \rVert_{\mathcal{H}_2}^2 > \lVert T^\prime \rVert_{\mathcal{H}_2}^2$.
leads to worse performance with directed interconnection. 
It is when both relative position and velocity measurements are used ($\gamma_p \neq 0$ and $\gamma_d \neq 0$) that the directed cycles can be utilized for better performance by tuning the gains. Per \defaultColor{Item \textit{\ref{thm_gp_3})}} of Theorem \ref{thm_gp}, 
%$\lVert T \rVert_{\mathcal{H}_2}^2 < \lVert T^\prime \rVert_{\mathcal{H}_2}^2$
sufficiently small $\gamma_p$ (i.e. sufficiently large velocity gains $k_d$ and $\gamma_d$) improves the performance of the directed interconnection relative to its 
%symmetrized 
\defaultColor{undirected}
counterpart;
but the performance degrades
%$\lVert T \rVert_{\mathcal{H}_2}^2 > \lVert T^\prime \rVert_{\mathcal{H}_2}^2$
%
for sufficiently large $\gamma_p$, as shown in Figure \ref{fig_section6}b. Directed cycles require less communication thus can be preferable, provided the gains are carefully selected.
%
%\ref{DavLoc_PosVel_Normal}b shows 
%that $\lVert T \rVert_{\mathcal{H}_2}^2 = \lVert T^\prime \rVert_{\mathcal{H}_2}^2$ due to Item \ref{remPos_1} in Remark \ref{normal_pos_perf_summary} ($\gamma_p = 0$), Figure \ref{DavLoc_PosVel_Normal}c shows that $\lVert T \rVert_{\mathcal{H}_2}^2 < \lVert T^\prime \rVert_{\mathcal{H}_2}^2$ due to Item \ref{remPos_2a} in Remark \ref{normal_pos_perf_summary} ($\gamma_p \neq 0$ and $k_d$ or $\gamma_d$ are sufficiently large) and Figure \ref{DavLoc_PosVel_Normal}a shows that $\lVert T \rVert_{\mathcal{H}_2}^2 > \lVert T^\prime \rVert_{\mathcal{H}_2}^2$ due to Item \ref{remPos_2b} in Remark \ref{normal_pos_perf_summary} ($\gamma_p \neq 0$ and $\gamma_d = 0$).

For the double-integrator network \eqref{secondmatrix} given by $\eqref{H_cases_2}$ (velocity based performance),
%and measures $P_{dav}$ and $P_{loc}$ defined by \eqref{C_dav} and \eqref{C_loc}, 
Figure \ref{fig_section6}c shows that relative position feedback degrades performance if the cycles are directed. But the performance becomes comparable to that of the 
%symmetrized
\defaultColor{undirected} 
system for sufficiently small $\gamma_p$, equaling it at $\gamma_p = 0$. 
%$\lVert T \rVert_{\mathcal{H}_2}^2 = \lVert T^\prime \rVert_{\mathcal{H}_2}^2$ 
%if $\gamma_p = 0$ and 
%$\lVert T \rVert_{\mathcal{H}_2}^2 > \lVert T^\prime \rVert_{\mathcal{H}_2}^2$ if $\gamma_p \neq 0$, 
This 
%confirms
\defaultColor{supports} the findings of
%items \ref{thm_vel_3} and \ref{thm_vel_2} of 
Theorem \ref{thm_vel}.
%
%due to Corollary \ref{cor12} ($\gamma_p = 0$).  
%
%Corollary 9
%\begin{cor} \label{cor9}

%

%
%
%

%
\iffalse
%
\begin{figure*}[t!]
\centering
%\vspace{-16pt}
\subfloat
{
       \includegraphics[width=0.28\linewidth]{pos_no_gd_ylabel}
       }
\subfloat
{
       \includegraphics[width=0.28\linewidth]{pos_yes_gd_ylabel}
       }
\subfloat
{
       \includegraphics[width=0.28\linewidth]{vel_yes_gd_ylabel}
       }
       
\vspace{-6pt}       

    \caption{ 
    The expectation of the position based performance of the double-integrator system \eqref{secondmatrix} given by $\eqref{H_cases_1}$, for $E\left[ \Sigma_0 \right] = I$ and the gains
    \textbf{(a)}
    $k_p = 3, k_d = 5, \gamma_d = 0$
    \textbf{(b)}
    $k_p = 1, k_d = 2, \gamma_d = 6.5$.
    $\bold{(c)}$
    The velocity based performance of the double-integrator system \eqref{secondmatrix} given by $\eqref{H_cases_2}$, for $E\left[ \Sigma_0 \right] = I$ and the gains
    $k_p = 1, k_d = 2, \gamma_d = 7$. 
    }
    \label{fig_section6}
\vspace{-10pt}
\end{figure*}
%
\fi
%

%Example Section 2
%\section{Centralized vs. Distributed Relative Feedback} 
\section{All-to-One vs. $\omega$-Nearest Neighbor Networks}
\label{CentDistFeed}

In this section, we compare two different relative feedback schemes.
The first one is called 
%centralized relative feedback 
an \mbox{all-to-one} network,
which designates a `leader' node that receives no relative feedback, \defaultColor{where} the remaining nodes \defaultColor{have access to} uniformly weighted uni-directional state measurements relative to the leader only. The second one is referred to as 
%distributed relative feedback, 
an $\omega$-nearest neighbor network,
which is based on uniformly weighted uni-directional state measurements of each node relative to 
%a certain number of 
$\omega$
succeeding nodes. 
%We model these feedback schemes using imploding star graphs and cyclic digraphs, 
%and 
\defaultColor{
We consider performance metrics that have circulant output matrices $C$,
%
%
%Circulant output matrices 
which arise in many 
%common 
applications 
%involving performance metrics 
%that 
such as
quantifying lack of coherence in a system in terms of global 
%as well as 
or local disorder \cite{BamiehJovanovic2012, teglingbamieh2019, oralgaymeacc2019}.}
%
\iffalse
We also consider a special case of such performance metrics 
%which are degrees of coherence and were studied for spatially invariant systems in \cite{BamiehJovanovic2012}. 
which is a global measure of 
%coherence 
disorder
and quantifies the aggregate state deviation from the average through 
%
\begin{equation} \label{C_dav}
C = I - \frac{1}{n} \bold{1} \bold{1}^\intercal
= L^{cyc} \left( \frac{n-1}{n},n-1 \right).
\end{equation}
%
This metric will be denoted by $P_{dav}$. 

%
%
\fi

%Then, the relationship between these two feedback schemes will be revealed for the special case of $P_{dav}$, which is determined by \eqref{C_dav}. 

%
%subsection 1
%\subsection{Imploding Star Graph: Centralized Relative Feedback} 
\subsection{Imploding Star Graph: All-to-One Networks} 

%The centralized relative feedback 
\mbox{All-to-one} networks
can be modeled \defaultColor{as} the
%using an 
imploding star graph whose edge weights are normalized such that the out-degree of each node is $\frac{n}{n-1}$. 
%Then 
The \defaultColor{corresponding} weighted Laplacian is given by
\begin{equation} \label{L_star}
L = 
%
\iffalse
\frac{n}{n-1}
\begin{bmatrix}
1 & 0 & \dots & 0 & -1 \\
0 & 1 & \ddots & \vdots & -1 \\
\vdots & \ddots & \ddots & 0 & \vdots \\
0 & \dots & 0 & 1 &-1 \\
0 & \dots & \dots & 0 & 0
\end{bmatrix},
\fi
%
%
\frac{n}{n-1}
\begin{bmatrix}
I_{n-1} & -\bold{1} \\
\bold{0}^{\intercal} & 0
\end{bmatrix}, 
\end{equation}
%
%where the 
\defaultColor{with
total out-degree 
%is 
$n$.} The Jordan decomposition gives
\begin{equation} \label{jordan_star}
J = 
\frac{n}{n-1}
\begin{bmatrix}
0 & \bold{0}^{\intercal} \\
\bold{0} & I_{n-1}
\end{bmatrix}. 
\end{equation}
Choosing $\alpha = 1$ in \eqref{RandQ}, the matrices $\tilde{R}$ and $\tilde{Q}$ are given by
\begin{equation}
\tilde{R} =
\begin{bmatrix}
I_{n-1} \\ 
\bold{0}^{\intercal}
\end{bmatrix} 
\
\text{and}
\
\tilde{Q} =
\begin{bmatrix}
I_{n-1} & -\bold{1}
\end{bmatrix}.
\label{r_q_star}
\end{equation}
\subsubsection{Single-Integrator Networks}
%Using this decomposition, 
The next theorem provides the solution for \eqref{perf_x} for the single-integrator network \eqref{firstmatrix} \defaultColor{using Theorem \ref{thm3} and the decomposition given by \eqref{jordan_star} and \eqref{r_q_star}.}
%
%Corollary 7
\begin{thm} \label{cor7}
Consider the single-integrator network \eqref{firstmatrix}. Suppose that $\mathcal{G}$ is an imploding star graph with the weighted Laplacian \eqref{L_star}, $C$ is circulant and the disturbance 
%is spatially uncorrelated,
\defaultColor{has unit covariance, 
i.e. $ E [ \Sigma_0 ] = I$}. Then the expectation of the performance metric \eqref{perf_x} for the system $T$ given by \eqref{H_cases_1} is
%
%\begin{equation} \label{perf_star}
%P = \frac{n-1}{n} \sum_{l=2}^{n} \mu_l 
%+ \frac{1}{n} \sum_{k > i}^{n} \sum_{l=2}^{n} 
%\cos \left( \frac{2 \pi}{n} (l-1) (k-i) \right) \mu_l,
%\end{equation}
%
%
%\begin{equation} 
%two lines
%\iffalse
\begin{align}
\label{perf_star}
%\small
%\! \! \! \! \! \! \! \! \! \!
\!
E\left[P\right] \! = 
\lVert T \rVert_{\mathcal{H}_2}^2
= \,
 \frac{n-1}{n^2} 
\sum_{i=2}^{n} \, & \mu_i
\Bigg( \! n-1 \, \, +
\\ 
\nonumber
+ &
\! \! \! \! \! \!
\sum_{ \substack{ l > k, \\ k,l \in \{ 2, \dots, n \} } }
\! \! \! \! \! \! \! \!
\cos \left( \! \frac{2 \pi}{n} (i-1) (l-k) \! \right) \!
\! \Bigg) \!. 
\! \! \!
\end{align}
%\end{equation}
%\fi
%
%one line
\iffalse
\begin{equation}
\label{perf_star}
\footnotesize
%\! \! \! \! \! \! \! \! \! \!
\!
E\left[P\right] \! = 
\lVert T \rVert_{\mathcal{H}_2}^2
= \,
 \frac{n-1}{n^2} 
\sum_{i=2}^{n} \, \mu_i
\Bigg( \! n-1 \, \, +
%\\ 
%\nonumber
%+ &
\! \! \! \! \! \!
\sum_{ \substack{ l > k, \\ k,l \in \{ 2, \dots, n \} } }
\! \! \! \! \! \! \! \!
\cos \left( \! \frac{2 \pi}{n} (i-1) (l-k) \! \right) \!
\! \Bigg) \!. 
\! \! \! 
\end{equation}
\normalsize 
\fi
\end{thm} 
\begin{proof}
Using the fact that $E\left[\Sigma_0 \right]= I$, we have $E\left[P\right] = \tr (\tilde{Q} \tilde{Q}^* \Psi)$. \eqref{r_q_star} leads to $\tilde{Q} \tilde{Q}^* = I_{n-1} + \bold{1} \bold{1}^\intercal$ which gives
\begin{equation} \label{P_psi_star}
E\left[
P\right] = \sum_{k=2}^{n} \Psi_{kk} + \sum_{k=2}^{n} \sum_{l=2}^{n} \Psi_{kl}.
\end{equation}
The matrix $M$ in \eqref{matrix_M} has the  eigenvectors 
\begin{equation} \label{C_circ_evec}
\boldsymbol{\theta}_{l} = \frac{1}{\sqrt{n}}
\begin{bmatrix} 
1 & e^{\bold{j} \frac{2 \pi}{n} (l-1)} & \dots & e^{\bold{j} \frac{2 \pi}{n} (l-1) (n-1)}
\end{bmatrix} ^*
\end{equation}
for $l=2,\dots,n$.
Using \eqref{C_circ_evec} and the columns of $\tilde{R}$ given in \eqref{r_q_star}, the scalar products in \eqref{nu_eta_kappa} are obtained as
\begin{equation} \label{scalar_star_r_theta}
\left \langle \boldsymbol{\theta}_{i}, \bold{r}_{k} \right \rangle
=
\frac{1}{\sqrt{n}}  e^{-\bold{j} \frac{2 \pi}{n} (i-1) (k-2)}, \quad
k = 2,\dots,n.
\end{equation}
By \eqref{psi_ik_diagable} and the fact that $\lambda_i = \frac{n}{n-1}$ for $i=2,\dots,n$ we have $ \Psi_{kl} = \frac{n-1}{2n} \nu_{kl}$, therefore using \eqref{nu_eta_kappa} and \eqref{scalar_star_r_theta} results in
\begin{equation} \label{perf_exp_star}
%\! \! \! \! \! \! \! \!
\small
E\left[P\right] 
= 
\frac{n-1}{2n^2} \!
\left(
\sum_{k=2}^{n} \sum_{i=2}^{n} \mu_i
+ \sum_{k=2}^{n} \sum_{l=2}^{n}  \sum_{i=2}^{n} 
e^{\bold{j} \frac{2 \pi}{n} (i-1) (l-k)} \mu_i
\right) \! \!. 
\! \! 
\end{equation}
Rearranging the terms in \eqref{perf_exp_star} and using Proposition \ref{Prop_H2L2} gives the result.
\end{proof}
\defaultColor{
We now consider a special case of 
%performance metrics with 
circulant output matrices $C$,
%which are degrees of coherence and were studied for spatially invariant systems in \cite{BamiehJovanovic2012}. 
%This is 
which leads to
a global measure of 
%coherence 
disorder
%and
that 
quantifies the aggregate state deviation from the average through 
\begin{equation} \label{C_dav}
C = I - \frac{1}{n} \bold{1} \bold{1}^\intercal
= L^{cyc} \left( \frac{n-1}{n},n-1 \right).
\end{equation}
This metric will be denoted by $P_{dav}$.}
\subsubsection*{Relationship to Previous Results}
\defaultColor{
For $P_{dav}$, the following proposition shows that the result in \cite{YoungLeonard2010} can be reproduced  as a special case of Theorem \ref{cor7}.}
%Proposition 3
\begin{prop} \label{prop_star}
Consider the single-integrator network \eqref{firstmatrix} \defaultColor{and the output matrix \eqref{C_dav}, i.e. the performance metric $P_{dav}$}. Suppose that $\mathcal{G}$ is an imploding star graph with the weighted Laplacian \eqref{L_star}, 
%$C = I - \frac{1}{n} \bold{1} \bold{1}^\intercal$, i.e. the performance metric quantifies the deviation from the state average 
and the disturbance 
%is spatially uncorrelated, 
\defaultColor{has unit covariance,
i.e. $ E [\Sigma_0] = I$}. Then the expectation of the performance metric \eqref{perf_x} for the system $T$ given by \eqref{H_cases_1} is 
\begin{equation} \label{P_star_dav}
E\left[
P_{dav}
\right] =
\lVert T \rVert_{\mathcal{H}_2}^2
=
 \frac{(n-1)^2}{2n}.
\end{equation}
\end{prop}
\begin{proof}
The fact that $\mu_i = 1 \ \forall i$ and \eqref{perf_exp_star} gives
$$
E\left[
P_{dav}
\right] = \frac{n-1}{2n^2} 
\bigg(
2 (n-1)^2
+ \sum_{k \neq l}  \sum_{i=2}^{n} 
e^{\bold{j} \frac{2 \pi}{n} (i-1) (l-k)}
\bigg).
$$
%
%Using the fact that
%
Since 
$\sum_{i=1}^{n} e^{ \bold{j} \frac{2 \pi}{n} (i-1) (l-k)} \! = \! 0$ for $l - k  \! = \! \pm 1, \dots, \pm (n-2)$,
%leads to
%
$$
%\blue{
E\left[
P_{dav}
\right] = \frac{n-1}{2n^2} 
\bigg(
2 (n-1)^2
- \underbrace{\sum_{k \neq l} e^{\bold{j} \frac{2 \pi}{n} 0 (l-k)}}_{ = (n-1)(n-2)}
\bigg). 
\quad 
\mbox{\qedhere}
\! \! \! \! \! \! \! 
\! \! \! \! \! \! 
%}
$$ 
%
\iffalse
$
E\left[
P_{dav}
\right] = \frac{n-1}{2n^2} 
\bigg(
2 (n-1)^2
- \underbrace{\sum_{k \neq l} e^{\bold{j} \frac{2 \pi}{n} 0 (l-k)}}_{ = (n-1)(n-2)}
\bigg)
= \frac{n-1}{2n^2} 
\bigg(
2 (n-1)^2
- (n-1)(n-2)
\bigg)
$
\fi
%Evaluating the expression above yields the result.
\end{proof}

\subsubsection{Double-Integrator Networks}
Using Corollary \ref{cor3} \defaultColor{from Section \ref{diagL}}, the following theorem \defaultColor{characterizes performance metric \eqref{perf_x} for 
all-to-one 
networks with double-integrator dynamics \eqref{secondmatrix}}. 
%
% FIGURE
%Deviation from the average - Position 
%\begin{figure*}[t]
%\centering
%\includegraphics[width=6in]{pdav_pos}

%\label{posDav}
%\end{figure*}
%
%

%Corollary 8
\begin{thm} \label{cor8}
Consider the double-integrator network \eqref{secondmatrix}. Suppose that $\mathcal{G}$ is an imploding star graph with the weighted Laplacian \eqref{L_star}, the output matrix $C$ is circulant and the disturbance 
%is spatially uncorrelated, 
\defaultColor{has unit covariance,
i.e. $ E [\Sigma_0] = I$}. Then the expectation of the performance metric \eqref{perf_x} is
\begin{equation} \label{P_star_secpos}
E\left[P\right]
=
\lVert T \rVert_{\mathcal{H}_2}^2
 = 
P_{0}
%\frac{k_d + \gamma_d}{ \gamma_p^2 + (k_p + \gamma_p)(k_d + \gamma_d)^2 }
\frac{1}{ 2 (k_p + \gamma_p \frac{n}{n-1})(k_d + \gamma_d \frac{n}{n-1}) } 
\end{equation}
for the system $T$ given by \eqref{H_cases_1} and
\begin{equation} \label{P_star_secvel}
E\left[P\right]
=
\lVert T \rVert_{\mathcal{H}_2}^2 = 
P_{0}
%\frac{(k_p + \gamma_p)(k_d + \gamma_d)}{ \gamma_p^2 + (k_p + \gamma_p)(k_d + \gamma_d)^2 }
\frac{1}{ 2(k_d + \gamma_d \frac{n}{n-1} ) } 
\end{equation}
for the system $T$ given by \eqref{H_cases_2}, where 
%$P_{0}$ is given by \eqref{perf_star}.
%
$$
%\blue{
P_{0} = \frac{1}{n}
\left(
\sum_{k=2}^{n} \sum_{i=2}^{n} \mu_i
+ \sum_{k=2}^{n} \sum_{l=2}^{n}  \sum_{i=2}^{n} 
e^{\bold{j} \frac{2 \pi}{n} (i-1) (l-k)} \mu_i
\right).
%}
$$

Furthermore, if 
\defaultColor{the output matrix is given by \eqref{C_dav},}
%$C = I - \frac{1}{n} \bold{1} \bold{1}^\intercal$, i.e. the performance metric quantifies the deviation from the state average, 
then
\begin{equation} \label{P_star_secpos_dav}
E\left[P_{dav} \right]
=
\lVert T \rVert_{\mathcal{H}_2}^2
= 
%\left( \frac{n-1}{2} \right)
%\frac{k_d + \gamma_d}{\gamma_p^2 + (k_p + \gamma_p)(k_d + \gamma_d)^2 }
\frac{n-1}{ 2 (k_p + \gamma_p \frac{n}{n-1})(k_d + \gamma_d \frac{n}{n-1}) }
\end{equation}
for the system $T$ given by \eqref{H_cases_1} and
\begin{equation} \label{P_star_secvel_dav}
E\left[P_{dav} \right]
=
\lVert T \rVert_{\mathcal{H}_2}^2  
= 
%\left( \frac{n-1}{2} \right)
%\frac{(k_p + \gamma_p)(k_d + \gamma_d)}{\gamma_p^2 + (k_p + \gamma_p)(k_d + \gamma_d)^2 }
\frac{n-1}{ 2 (k_d + \gamma_d \frac{n}{n-1} ) }
\end{equation}
for the system $T$ given by \eqref{H_cases_2}.
\end{thm}
\begin{proof}
Substitution of $\lambda_k = \frac{n}{n-1}$ for $k = 2, \dots, n$ into \eqref{cor3_psi_pos} and \eqref{cor3_psi_vel} gives
%$$
%
$
\Psi_{kl} = \nu_{kl}
%\frac{k_d + \gamma_d}{ \gamma_p^2 + (k_p + \gamma_p)(k_d + \gamma_d)^2 }
\frac{1}{ 2 (k_p + \gamma_p \frac{n}{n-1})(k_d + \gamma_d \frac{n}{n-1}) }
$
%$$
for the system $T$ given by \eqref{H_cases_1} and
%$$
$
\Psi_{kl} = \nu_{kl}
%\frac{(k_p + \gamma_p)(k_d + \gamma_d)}{\gamma_p^2 + (k_p + \gamma_p)(k_d + \gamma_d)^2 }
\frac{1}{ 2 (k_d + \gamma_d \frac{n}{n-1}) }
$
%$$
for the system $T$ given by \eqref{H_cases_2}. By the argument given in the proof of Theorem \ref{cor7}, using the expressions above and  \eqref{P_psi_star} leads to \eqref{P_star_secpos} and \eqref{P_star_secvel}. 
%Similarly, 
The argument given in the proof of Proposition~\ref{prop_star} combined with \eqref{P_star_secpos} and \eqref{P_star_secvel} yields \eqref{P_star_secpos_dav} and \eqref{P_star_secvel_dav}.
\end{proof}
%

%Based on corollaries \ref{cor7} and \ref{cor8}  %matrix $M$ in \eqref{matrix_M}. 

When $P_{dav}$ is considered, Proposition \ref{prop_star} and Theorem \ref{cor8} show that the performance metric grows unboundedly with the network size. Next we study 
%the distributed relative feedback scheme.
$\omega$-nearest neighbor networks.

%subsection 2
%\subsection{Cyclic Digraphs: Distributed Relative Feedback}
\subsection{Cyclic Digraphs: $\omega$-Nearest Neighbor Networks} 

%The distributed relative feedback 

The cyclic digraph defined by the weighted Laplacian \eqref{cycle_L} can be used to model $\omega$-nearest neighbor networks. In order to normalize the edge weights of the digraphs with different number of communication hops we choose the out-degree of each node as $d = 1$ \defaultColor{in \eqref{cycle_L}}, which leads to
\begin{equation} \label{dist_feed}
L = L^{cyc} \left( 1, \omega \right)
\end{equation}
so that the total out-degree in the graph is $n$.
Since we consider circulant output matrices $C$, the eigenvectors of $M$ in \eqref{matrix_M} are given by
%
%Considering circulant output matrices, \eqref{nu_mu} is satisfied due to \eqref{th_r_orth}. 
%
\eqref{C_circ_evec}. Combining this with \eqref{r_cycle}, the scalar products in \eqref{nu_eta_kappa} are obtained as
\begin{equation} \label{th_r_orth}
\left \langle \boldsymbol{\theta}_{l}, \bold{r}_{k} \right \rangle
=
\begin{cases}
1 & k = l \\
0 & k \neq l
\end{cases},
\ k = 2, \dots, n,
\end{equation}
therefore \eqref{nu_eta_kappa} leads to 
\begin{equation} \label{nu_mu}
\nu_{kk} = \mu_k. 
\end{equation}
This means that the dependence of \eqref{perf_single_normal}, \eqref{perf_x_sol_normal} and \eqref{perf_v_sol_normal} on the output matrix $C$ is only through the eigenvalues $\mu_k$ of $M$.

Then performance is given by  \eqref{perf_single_normal} for the single-integrator system, and by \eqref{perf_x_sol_normal} or \eqref{perf_v_sol_normal} for the double-integrator system, where due to \eqref{jordan_cycle} the eigenvalues of $L$ satisfy
\begin{equation} \label{lambda_cyc}
\lambda_k = 
1 - \frac{1}{\omega} \sum_{i=1}^{\omega}
e^{-\bold{j} \frac{2 \pi}{n} i (k-1)},
\ k=1,\dots,n.
\end{equation}
%
%due to \eqref{jordan_cycle}. 

%
Next we present two examples 
%in order 
to demonstrate the effect of the number of communication hops $\omega$ on the performance of 
%the distributed relative feedback 
$\omega$-nearest neighbor networks
and to investigate the relationship between 
%the centralized and distributed feedback 
all-to-one and all-to-all communication \defaultColor{structures}.
%$\omega$-nearest neighbor networks.
%schemes.

%
% FIGURE
%Deviation from the average 
% Comm Hops
%Complete vs. Star Graph
%Position and Velocity based performance
\begin{figure*}[t]
\centering
\includegraphics[scale=0.44, trim={5.2cm 0.2cm 4cm 0.2cm},clip] {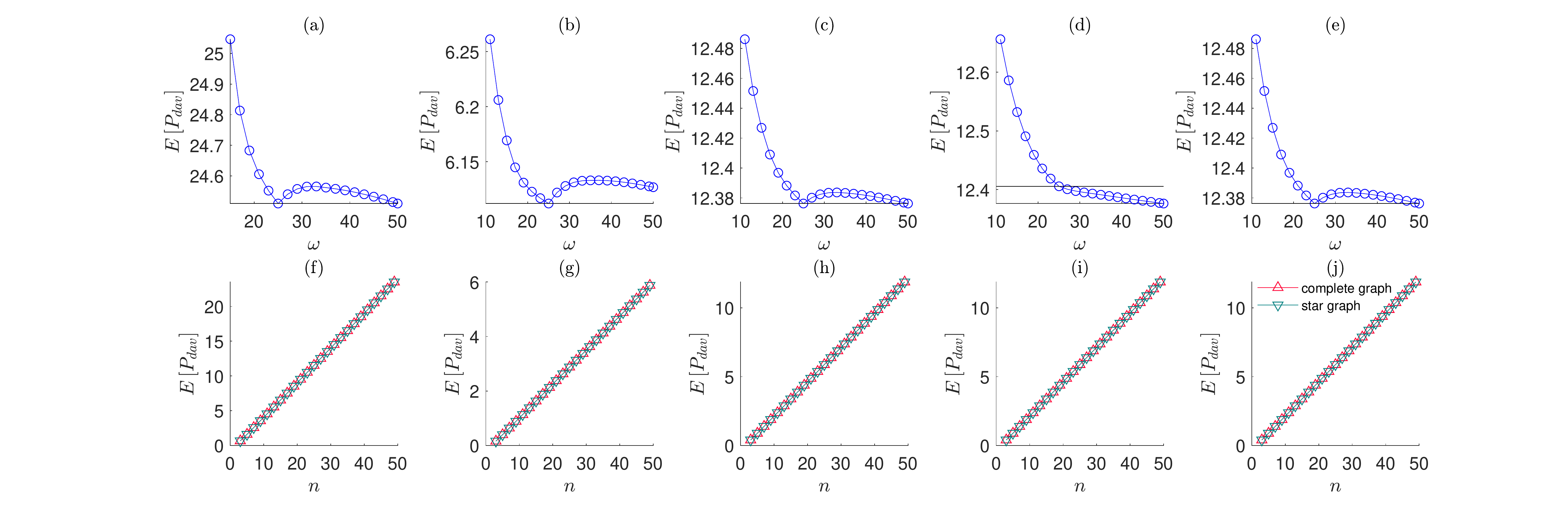}
\vspace{-0.7cm}
\caption{ \textbf{(Top)}
The expectation of $P_{dav}$ defined by \eqref{C_dav} versus the number of communication hops $\omega$ of 
%the distributed feedback 
the $\omega$-nearest neighbor networks
given by \eqref{dist_feed} where the network size is $n=51$. 
\textbf{(Bottom)}
The expectation of $P_{dav}$ versus the network size $n$ for the imploding star graph and the complete graph given by  \eqref{L_star} and \eqref{L_dist_comp}.
The disturbance 
%is spatially uncorrelated, 
\defaultColor{has unit covariance,
i.e. $E[ \Sigma_0 ] = I$}.
Plots respectively illustrate the cases of:
\textbf{(a, f)} single-integrator \eqref{firstmatrix} given by $\eqref{H_cases_1}$, 
\mbox{\textbf{(b, g)} double-integrator \eqref{secondmatrix} given by $\eqref{H_cases_1}$ (position-based performance), $k_p = k_d = \gamma_p = \gamma_d = 1$,}
\textbf{(c, h)} double-integrator \eqref{secondmatrix} given by $\eqref{H_cases_1}$ (position-based performance), $k_p = k_d = \gamma_d = 1, \gamma_p = 0$,
\textbf{(d, i)} double-integrator \eqref{secondmatrix} given by $\eqref{H_cases_2}$ (velocity-based performance), $k_p = k_d = \gamma_p = \gamma_d = 1$,
\mbox{\textbf{(e, j)} double-integrator \eqref{secondmatrix} given by $\eqref{H_cases_2}$ (velocity-based performance), $k_p = k_d = \gamma_d = 1, \gamma_p = 0$.}
}
\label{hops_comp_star}
\vspace{-0.5cm}
\end{figure*}
%

%Example
\subsection{\defaultColor{Example: Number of Communication Hops}}

In the following we investigate how performance changes with respect to $\omega$. We first show that performance does not necessarily improve by increasing $\omega$, i.e. 
%through a larger number of relative feedback measurements in distributed relative feedback.
through communication with a larger 
%range 
\defaultColor{number
of nearest neighbors.}

%\begin{eg} \label{eg_hops}
For convenience suppose that $n$ is odd. Consider the case where $\omega = \frac{n-1}{2}$ such that $L=L^{cyc} (1,\frac{n-1}{2})$. Using the definition given by \eqref{herm_part_L}
\begin{equation} \label{L_dist_comp}
L^{\prime} = \frac{ L^{cyc} (1,\frac{n-1}{2}) + L^{cyc} (1,\frac{n-1}{2})^* } {2}
= L^{cyc} (1,n-1),
\end{equation}
i.e. $L^{\prime}$ is the weighted Laplacian associated with the complete graph with uniform edge weights $\frac{1}{n-1}$. Then the associated systems $T$ and $T^\prime$ have the following properties for any performance metric satisfying Assumption \ref{assum1}:
\begin{itemize}[leftmargin=*]
\item $\lVert T \rVert_{\mathcal{H}_2}^2 = \lVert T^{\prime} \rVert_{\mathcal{H}_2}^2$ for the single-integrator network \eqref{firstmatrix} defined by $\eqref{H_cases_1}$ due to Theorem \ref{cor_single_normal},
\item It is possible due to Theorem \ref{thm_gp} that 
$\lVert T \rVert_{\mathcal{H}_2}^2 <  \lVert T^{\prime} \rVert_{\mathcal{H}_2}^2$,
$\lVert T \rVert_{\mathcal{H}_2}^2 = \lVert T^{\prime} \rVert_{\mathcal{H}_2}^2$ or
$\lVert T \rVert_{\mathcal{H}_2}^2 > \lVert T^{\prime} \rVert_{\mathcal{H}_2}^2$
for the position based performance of the double-integrator network \eqref{secondmatrix} defined by system \eqref{H_cases_1},
\item It can only hold that 
$\lVert T \rVert_{\mathcal{H}_2}^2 = \lVert T^{\prime} \rVert_{\mathcal{H}_2}^2$ or
$\lVert T \rVert_{\mathcal{H}_2}^2 > \lVert T^{\prime} \rVert_{\mathcal{H}_2}^2$
for the velocity based performance of the double-integrator network \eqref{secondmatrix} defined by system \eqref{H_cases_2} due to Theorem \ref{thm_vel}.
\end{itemize}
%
%\end{eg}

As this example suggests, using half the number of communication hops as compared to the complete graph, i.e. the case in which 
%the relative feedback is fully distributed, 
$\omega$ is maximal,
provides identical performance for the single integrator network \eqref{firstmatrix}. It is possible to achieve better performance \defaultColor{using half the number of hops compared to the complete graph} in the case of the position based metrics of the double integrator network \eqref{secondmatrix}; 
%using half the number of hops compared to the complete graph 
but this is not the case for the velocity based metrics.

The dependence of $E\left[P_{dav}\right]$ on $\omega$ is illustrated in figures \mbox{\ref{hops_comp_star}a - \ref{hops_comp_star}e} for a case in which $n=51$ and the disturbance 
%is spatially uncorrelated,
\defaultColor{has unit covariance, 
i.e. $E[\Sigma_0] = I$}. For the single integrator network \eqref{firstmatrix} we observe in Figure \ref{hops_comp_star}a that $\lVert T \rVert_{\mathcal{H}_2}^2 = \lVert T^{\prime} \rVert_{\mathcal{H}_2}^2$. This is also true for the position and velocity based performance of the double-integrator network \eqref{secondmatrix} if relative position feedback is absent ($k_p=k_d=\gamma_d=1$ and $\gamma_p = 0$) as shown in figures \ref{hops_comp_star}c (due to Item \ref{thm_gp_1} in Theorem \ref{thm_gp}) and \ref{hops_comp_star}e (due to Item \ref{thm_vel_2} in Theorem \ref{thm_vel}). Conversely, using relative position feedback ($k_p=k_d=\gamma_p=\gamma_d = 1$) leads to $\lVert T \rVert_{\mathcal{H}_2}^2 < \lVert T^{\prime} \rVert_{\mathcal{H}_2}^2$ as shown in Figure \ref{hops_comp_star}b (due to Item \ref{thm_gp_3} in Theorem \ref{thm_gp}) for the position based performance and to $\lVert T \rVert_{\mathcal{H}_2}^2 > \lVert T^{\prime} \rVert_{\mathcal{H}_2}^2$ as shown in Figure \ref{hops_comp_star}d (due to Item \ref{thm_vel_3} in Theorem \ref{thm_vel}) for the velocity based performance. For all cases, increasing $\omega$ up to $\omega = 25$ monotonically improves performance. Compared to $\omega = 25$, choosing $25 < \omega < 50$ degrades performance, excluding the velocity based performance with relative position feedback ($\gamma_p \neq 0$, Figure \ref{hops_comp_star}d) which improves monotonically as $\omega$ is increased. Therefore at least for $n=51$ and the cases in figures \ref{hops_comp_star}a-\ref{hops_comp_star}c and \ref{hops_comp_star}e, $\omega = \frac{n-1}{2}$ provides the optimal performance.

The next example provides a comparison between 
%centralized 
\mbox{all-to-one}
and 
%fully distributed relative feedback schemes.
all-to-all networks.

%\subsection{Example: Relationship between Centralized and Fully Distributed Relative Feedback}
\subsection{\defaultColor{Example: \mbox{All-to-One} versus \mbox{All-to-All} Networks}}

%\begin{eg} \label{eg_comp_star}
For the special case of $P_{dav}$ which is determined by \eqref{C_dav}, \eqref{nu_mu} holds and we have $\mu_k =1$ for $k=2,\dots,n$. 
If 
%relative feedback is fully distributed,
\mbox{all-to-all} networks are considered, 
i.e. $L$ is given by  \eqref{L_dist_comp}, \eqref{lambda_cyc} reduces to $\lambda_k = \frac{n}{n-1}$ for $k=2,\dots,n$. Then $P_{dav}$ is given by 
\begin{itemize} [leftmargin=*]
\item \eqref{P_star_dav} for the single-integrator network \eqref{firstmatrix} given by \eqref{H_cases_1},
\item \eqref{P_star_secpos_dav} for the double-integrator network \eqref{secondmatrix} given by \eqref{H_cases_1},
\item \eqref{P_star_secvel_dav} for the double-integrator network \eqref{secondmatrix} given by \eqref{H_cases_2},
\end{itemize} 
where we respectively used \eqref{perf_single_normal}, \eqref{P_normal_x_noImag} and \eqref{P_normal_v_noImag}.
%\end{eg}
Therefore, 
%distributed feedback 
$\omega$-nearest neighbor networks
with $\omega = n-1$ (all-to-all) and 
%centralized feedback 
\mbox{all-to-one} networks
perform \emph{identically} if $P_{dav}$ is considered, which is illustrated in figures \ref{hops_comp_star}f - \ref{hops_comp_star}j for up to $n=49$. In conclusion, given that the total out-degree is normalized to be $n$ for each graph, the same $P_{dav}$ is achieved by using $n-1$ directed edges that follow a common leader 
%in comparison to 
\defaultColor{as that of}
using $n(n-1)$ directed edges such that each node follows every other node. The latter feedback scheme can be interpreted as every node being a common leader in the sense of the former feedback scheme. In other words, 
%the fully distributed relative feedback scheme 
the \mbox{all-to-all} network
can be interpreted as the superposition of $n$ 
%centralized relative feedback schemes
\mbox{all-to-one} networks
with edge weights scaled by $\frac{1}{n}$. Thus the same level of deviation from the average state (position or velocity) is achieved by following a single common leader instead of using all-to-all communication, provided the edge weights are sufficiently large. As $n$ grows, the number of edges grow linearly and each edge weight $\frac{n}{n-1}$ remains bounded in 
%the centralized relative feedback scheme. 
\mbox{all-to-one} networks.
In contrast, the number of edges grow quadratically and each edge weight $\frac{1}{n-1}$ decays to zero in 
%the fully distributed relative feedback scheme. 
\mbox{all-to-all} networks.
We note for double-integrator networks \eqref{secondmatrix} given by \eqref{H_cases_1} that
compared to both 
%centralized and fully distributed relative feedback schemes 
\mbox{all-to-one} and \mbox{all-to-all} 
%networks,
communication,
it is possible to achieve better position-based $P_{dav}$
%by choosing 
with
$\omega = \frac{n-1}{2}$ nearest neighbor interactions (odd $n$),
if both relative position and velocity feedback are employed and 
%its respective 
the relative position feedback
gain $\gamma_p$ is sufficiently small (e.g. Figure \ref{hops_comp_star}b).

\section{Conclusions} \label{conclusion}
We studied the performance of single and double integrator networks over arbitrary digraphs that have at least one globally reachable node. Using a unifying framework, closed-form solutions are provided for a general class of output $\mathcal{L}_2$ norm based \defaultColor{quadratic} performance metrics.
The special case of normal weighted Laplacian matrices reveals the importance of judicious control design  for mitigating any performance degradation in directed networks, and possibly improving upon their \defaultColor{undirected} counterparts.
%(e.g. in spatially invariant systems), 
%For example, a spatially invariant and directed communication structure can outperform its symmetrized counterpart (which uses twice the number of relative state measurements) if position based measures are used. Yet, this is not the case for velocity based performance. 
In addition, we have demonstrated that performance is sensitive to the degree of connectivity (e.g. range of communication in $\omega$-nearest neighbor networks), but it does not depend on it monotonically. This non-monotonicity can also be deduced from
%due to the effect of edge directionality, performance does not necessarily improve by increasing the number of communication links in distributed relative feedback.
%It is also worth noting that there exists an 
the equivalence between 
%the feedback schemes which use relative state measurements that are respectively fully distributed or relative to a common leader.
\mbox{all-to-one} and \mbox{all-to-all} networks.
That is, the same level of state deviation from the average is achieved using either \defaultColor{network architecture}.

% or


\appendix

\subsection{Lemmas from Subsection \ref{PerfArbDouble}}

%Before presenting the closed-form solutions for the performance measure $P$ in \eqref{perf_x}, we first provide two lemmas that are concerned with deriving a time-domain realization for a transfer function that is a common factor for all $\tilde{h}_{p q}^{(k)}(s)$ in \eqref{TF_block}. These realizations will be used to compute the scalar products \eqref{L2_scalar} for the double-integrator network \eqref{secondmatrix}.

%Lemma 3
\begin{lemma} \label{lemma3}
Consider the transfer function
\begin{equation} \label{phi_TF}
\Omega_{k,\delta}(s) = \frac
{ r(s) {\left( \gamma_p + s \gamma_d \right)}^{\delta-1}}
{{\left[ s^2 + (k_d + \gamma_d \lambda_k) s + k_p + \gamma_p  \lambda_k \right] }^\delta}
\end{equation}
for some $\delta \in \mathbb{Z}_{+}$. Suppose that $s^2 + (k_d + \gamma_d \lambda_k) s + k_p + \gamma_p  \lambda_k = 0$ has distinct roots $\rho_1^{(k)}$ and $\rho_2^{(k)}$, i.e. $\rho_1^{(k)} \neq \rho_2^{(k)}$. Then, $\Omega_{k,\delta}(s)$ has a realization $\left( \mathcal{A}_{k,\delta}, \mathcal{B}_{k,\delta}, \mathcal{C}_{k,\delta} \right)$ in Jordan canonical form given by
\begin{equation} \label{Jordan_realization}
\mathcal{A}_{k,\delta} =
%
\iffalse
\begin{bmatrix}
\mathcal{J} (\rho_1^{(k)}, \delta) & \\
& \mathcal{J} (\rho_2^{(k)}, \delta)
\end{bmatrix},
\fi
%
%
\bdiag{ \left(
\mathcal{J} (\rho_i^{(k)}, \delta)
\right)_{i=1,2} }, 
\end{equation}
$\mathcal{B}_{k,\delta} = \bigg[ \underbrace{\begin{matrix} 0 & \dots & 1 \end{matrix}}_{1 \times \delta} \underbrace{\begin{matrix} 0 & \dots & 1 \end{matrix}}_{1 \times \delta} \bigg]^\intercal$, 
$\mathcal{C}_{k,\delta} = \begin{bmatrix} c_{1}^{(k)} & \dots & c_{2\delta}^{(k)} \end{bmatrix}$,
where $\mathcal{J} (\rho_1^{(k)}, \delta)$ denotes the size-$\delta$ Jordan block with the eigenvalue $\rho_1^{(k)}$. 

If $r (s) = 1$, i.e. we consider system $T$ given by \eqref{H_cases_1}, the elements of $\mathcal{C}_{k,\delta}$ are given by
%
%\begin{itemize}
%\item If $\gamma_d \neq 0$,
\begin{align} 
c_{l}^{(k)} &= \sum_{\zeta=0}^{l-1} \tau(\zeta,l) {\gamma_d}^{\zeta} \frac{{(\gamma_p + \rho_1^{(k)} \gamma_d)}^{\delta-\zeta-1}}{{(\rho_1^{(k)} - \rho_2^{(k)})}^{\delta+l-\zeta-1}}, \nonumber \\
c_{l+\delta}^{(k)} &= \sum_{\zeta=0}^{l-1} \tau(\zeta,l) {\gamma_d}^{\zeta} \frac{{(\gamma_p + \rho_2^{(k)} \gamma_d)}^{\delta-\zeta-1}}{{(\rho_2^{(k)} - \rho_1^{(k)})}^{\delta+l-\zeta-1}}, \nonumber
\end{align}
if $r (s) = s$, i.e. we consider system $T$ given by \eqref{H_cases_2}, the elements of $\mathcal{C}_{k,\delta}$ are given by
\footnotesize
\begin{align} 
c_{l}^{(k)} &= 
\sum_{\zeta=0}^{l-1} \tau(\zeta,l) {\gamma_d}^{\zeta-1} 
\left( 
\frac{ \zeta \gamma_p + \delta \rho_1^{(k)} \gamma_d } {\delta - \zeta}
\right)
\frac{{(\gamma_p + \rho_1^{(k)} \gamma_d)}^{\delta-\zeta-1}}{{(\rho_1^{(k)} - \rho_2^{(k)})}^{\delta+l-\zeta-1}},
\nonumber \\
c_{l+\delta}^{(k)} &= 
\sum_{\zeta=0}^{l-1} \tau(\zeta,l) {\gamma_d}^{\zeta-1} 
\left( 
\frac{ \zeta \gamma_p + \delta \rho_2^{(k)} \gamma_d } {\delta - \zeta}
\right)
\frac{{(\gamma_p + \rho_2^{(k)} \gamma_d)}^{\delta-\zeta-1}}{{(\rho_2^{(k)} - \rho_1^{(k)})}^{\delta+l-\zeta-1}},
\nonumber
\end{align}
\normalsize
%
%
%
%\item If $\gamma_d = 0$,
%\begin{align} 
%c_{l}^{(i)} &= {\gamma_p}^{\delta-1} \tau(0,l) {(\rho_1^{(i)} - \rho_2^{(i)})}^{1-l-\delta}, \nonumber \\
%c_{l+\delta}^{(i)} &= {\gamma_p}^{\delta-1} \tau(0,l) {(\rho_2^{(i)} - \rho_1^{(i)})}^{1-l-\delta} \nonumber
%\end{align}
%
%\end{itemize}
for $l=1,\dots,\delta$, where $\tau(\zeta,l) = {(-1)}^{l-\zeta-1} \binom{l-1}{\zeta}
\binom{\delta+l-\zeta-2}{l-1}
%\frac{(\delta+l-\zeta-2)!}{(\delta-\zeta-1)!}
$.
\begin{proof}
Using the fact that the denominator of $\Omega_{k,\delta}(s)$ has distinct roots 
\begin{equation} \nonumber
\Omega_{k,\delta}(s) =  
\frac{\Gamma (s)}{{(s - \rho_1^{(k)})}^{\delta} {(s - \rho_2^{(k)})}^{\delta}},
\end{equation}
where $\Gamma (s) = r(s) {\left( \gamma_p + s \gamma_d \right)}^{\delta-1}$. Applying partial fractions, we have
\begin{equation} \label{par_frac}
\Omega_{k,\delta}(s) = \sum_{l=1}^{\delta} 
\frac{c_{l}^{(k)}}{{(s - \rho_1^{(k)})}^{\delta-l+1}} +
\frac{c_{l+\delta}^{(k)}}{{(s - \rho_2^{(k)})}^{\delta-l+1}},
\end{equation}
which can be represented by the Jordan canonical realization \eqref{Jordan_realization}.
%[reference]. 
Here the coefficients $c_{l}^{(k)}$ and $c_{l+\delta}^{(k)}$ are given by
\begin{align} 
c_{l}^{(k)} &= \frac{1}{(l-1)!}
\lim_{s \to  \rho_1^{(k)}} \frac{d^{l-1}}{ds^{l-1}} \left[ {(s - \rho_1^{(k)})}^{\delta} \Omega_{k,\delta}(s) \right], \label{c_l} \\
c_{l+\delta}^{(k)} &= \frac{1}{(l-1)!}
\lim_{s \to  \rho_2^{(k)}} \frac{d^{l-1}}{ds^{l-1}} \left[ {(s - \rho_2^{(k)})}^{\delta} \Omega_{k,\delta}(s) \right]. \label{c_l_beta}
\end{align}

The general Leibniz rule for the derivative of product 
%[reference] 
yields
%
%\footnotesize
\begin{align} 
\! \! \! \! \! \!
c_{l}^{(k)} \! &= \! \! \lim_{s \to  \rho_1^{(k)}} \!
\sum_{\zeta = 0}^{l-1} 
\frac{ \binom{l-1}{\zeta} } {(l-1)!}
\frac{d^{\zeta} \Gamma (s) }{ds^{\zeta}} 
\frac{d^{l-1-\zeta}}{ds^{l-1-\zeta}} \! \!
\left[ {(s - \rho_2^{(k)})}^{-\delta} \right] \! \!, \! \!
\label{c_l_2} \\
\! \! \! \!
c_{l+\delta}^{(k)} \! &= \! \! \lim_{s \to  \rho_2^{(k)}} \!
\sum_{\zeta = 0}^{l-1} 
\frac{ \binom{l-1}{\zeta} } {(l-1)!}
\frac{d^{\zeta}}{ds^{\zeta}} \Gamma (s)
\frac{d^{l-1-\zeta}}{ds^{l-1-\zeta}} \left[ {(s - \rho_1^{(k)})}^{-\delta} \right] \! \!, \!
\label{c_l_beta_2}
\end{align}
%\normalsize
%
For the cases of $r(s) = 1$ or $ r(s) = s $, a direct calculation shows that
\begin{equation}
\!
\frac{d^{\zeta}}{ds^{\zeta}} \left[ {(\gamma_p + s \gamma_d)}^{\delta-1} \right] 
\! \! = \!
{\gamma_d}^{\zeta} \frac{(\delta-1)!}{(\delta-\zeta-1)!} 
{(\gamma_p + s \gamma_d)}^{\delta-\zeta-1} \! , \! \! \! \!
\label{deriv1}
\end{equation}
\begin{align}
\! \! \! \! \! \! \! \! \! \!
\frac{d^{\zeta}}{ds^{\zeta}} \left[ {s(\gamma_p + s \gamma_d)}^{\delta-1} \right] 
\! \! &= \!
s {\gamma_d}^{\zeta} \frac{(\delta-1)!}{(\delta-\zeta-1)!} 
{(\gamma_p + s \gamma_d)}^{\delta-\zeta-1} \nonumber \\
&+ {\gamma_d}^{\zeta - 1} \zeta \frac{(\delta-1)!}{(\delta-\zeta)!} 
{(\gamma_p + s \gamma_d)}^{\delta-\zeta} \! , \! \!
\label{deriv2}
\end{align}
\scriptsize
\begin{equation}
\frac{d^{l-1-\zeta}}{ds^{l-1-\zeta}} \left[ {(s - \rho_2^{(k)})}^{-\delta} \right] =
{(-1)}^{l-1-\zeta} \frac{(\delta+l-\zeta-2)!}{(\delta-1)!} 
{(s - \rho_2^{(k)})}^{-\delta-l+\zeta+1}. \label{deriv3}
\end{equation}
\normalsize
Substituting \eqref{deriv1}, \eqref{deriv2} and \eqref{deriv3} into \eqref{c_l_2} and taking the limit gives the desired result. A similar procedure can be followed to evaluate the expression in \eqref{c_l_beta_2}.
%If $\gamma_d = 0$, using \eqref{c_l}, \eqref{c_l_beta} and \eqref{deriv2} gives
%
%\begin{align} 
%c_{l}^{(i)} &= \lim_{s \to  \rho_1^{(i)}} 
%{\gamma_p}^{\delta-1} {(-1)}^{l-1} \frac{(\delta+l-2)!}{(\delta-1)!} 
%{(s - \rho_2^{(i)})}^{-\delta-l+1}, 
%\nonumber \\
%c_{l+\delta}^{(i)} &= \lim_{s \to  \rho_2^{(i)}} 
%{\gamma_p}^{\delta-1} {(-1)}^{l-1} \frac{(\delta+l-2)!}{(\delta-1)!} 
%{(s - \rho_1^{(i)})}^{-\delta-l+1}, 
%\nonumber
%\end{align}
%
%which leads to the desired result after taking the limit.
\end{proof}
\end{lemma}
%

%\red{
%arxiv version
%The following lemma addresses the case where the roots are repeated.
%Lemma 4
\begin{lemma} \label{lemma4}
Consider the transfer function
$\Omega_{k,\delta}(s)$ in \eqref{phi_TF}
\iffalse
\begin{equation} \nonumber
\Omega_{k,\delta}(s) = \frac
{r(s) {\left( \gamma_p + s \gamma_d \right)}^{\delta-1}}
{{\left[ s^2 + (k_d + \gamma_d \lambda_k) s + k_p + \gamma_p  \lambda_k \right] }^\delta}
\end{equation}
\fi
for some $\delta \in \mathbb{Z}_{+}$. Suppose that $s^2 + (k_d + \gamma_d \lambda_k) s + k_p + \gamma_p  \lambda_k = 0$ has repeated roots $\rho_1^{(k)}$ and $\rho_2^{(k)}$, i.e. $\rho_1^{(k)} = \rho_2^{(k)} = \rho^{(k)}$. Then, $\Omega_{k,\delta}(s)$ has a realization $\left( \mathcal{A}_{k,\delta}, \mathcal{B}_{k,\delta}, \mathcal{C}_{k,\delta} \right)$ in Jordan canonical form given by
\begin{equation} \label{Jordan_realization_2}
\mathcal{A}_{k,\delta} =
\mathcal{J} (\rho^{(k)}, 2\delta),
\end{equation}
\begin{equation} \nonumber
\mathcal{B}_{k,\delta} = \bigg[ \underbrace{\begin{matrix} 0 & \dots & 1 \end{matrix}}_{1 \times 2\delta} \bigg]^\intercal, 
\mathcal{C}_{k,\delta} = \begin{bmatrix} c_{1}^{(k)} & \dots & c_{2\delta}^{(k)} \end{bmatrix}.
\end{equation}

If $r(s) = 1$, i.e. we consider system $T$ given by \eqref{H_cases_1}, the elements of $\mathcal{C}_{k,\delta}$ are given by
\begin{equation} \nonumber
c_{l}^{(k)} = 
\begin{cases} 
{\gamma_d}^{l-1} 
\binom{\delta-1}{l-1}
%\frac{(\delta-1)!}{(\delta-l)!} 
{(\gamma_p + \rho^{(k)} \gamma_d)}^{\delta-l}, & 1 \leq l \leq \delta \\
\hfil 0 , & \delta+1 \leq l \leq 2\delta
\end{cases},
\end{equation}
if $ r(s) = s $, i.e. we consider system $T$ given by \eqref{H_cases_2}, the elements of $\mathcal{C}_{k,\delta}$ are given by
\footnotesize
\begin{equation} \nonumber
c_{l}^{(k)} \! = \! 
\begin{cases} 
\!
\left[ \frac{ (l-1)\gamma_p + \delta \rho^{(k)} \gamma_d } {\delta - l + 1} \right] \!
{\gamma_d}^{l-2}
\binom{\delta-1}{l-1}
%\frac{(\delta-1)!}{(\delta-l+1)!} 
{(\gamma_p + \rho^{(k)} \gamma_d)}^{\delta-l}, & 1 \leq l \leq \delta \\
\hfil {\gamma_d}^{\delta - 1}, & l = \delta+1 \\
\hfil 0, & \delta+2 \leq l \leq 2\delta
\end{cases}.
\end{equation}
\normalsize
\begin{proof}
%
%arxiv version
%\iffalse
Using the fact that $\Omega_{k,\delta}(s)$ has repeated roots leads to 
\begin{equation} \nonumber
\Omega_{k,\delta}(s) =  
\frac{r(s) {\left( \gamma_p + s \gamma_d \right)}^{\delta-1}} 
{{(s - \rho^{(k)})}^{2\delta}}.
\end{equation}
Applying partial fractions, we have
\begin{equation} \nonumber
\Omega_{k,\delta}(s) = \sum_{l=1}^{2\delta} 
\frac{c_{l}^{(k)}}{{(s - \rho^{(k)})}^{2\delta-l+1}},
\end{equation}
which can be represented by the Jordan canonical realization \eqref{Jordan_realization_2}.
%[reference]. 
Here the coefficients $c_{l}^{(k)}$ are given by
\begin{equation} 
c_{l}^{(k)} = 
\frac{1}{(l-1)!}
\lim_{s \to  \rho^{(k)}} \frac{d^{l-1}}{ds^{l-1}} 
\left[
r(s) {\left( \gamma_p + s \gamma_d \right)}^{\delta-1} 
\right]. \label{c_l_2_rep} 
%&= \lim_{s \to  \rho^{(i)}} \frac{d^{l-1}}{ds^{l-1}} 
%\left[ {(\gamma_p + s \gamma_d)}^{\delta-1} \right].
\end{equation}
For the cases of $r(s) = 1$ or $r(s) = s$, using respectively \eqref{deriv1} and \eqref{deriv2} and taking the limit in \eqref{c_l_2_rep} gives the desired result. 
%\fi
%
\iffalse
%short version
%
%We refer the reader to 
See
\cite{OralMalladaGayme2019_arxiv} for a proof.
\fi
\end{proof}
\end{lemma}
%
%}

%
%Lemma 4

%
%\footnotesize
%\begin{align} 

%

%\end{lemma}
%
%\begin{proof}
%See appendix.
%\end{proof}
%

%Lemma 5

%
%

%
%\end{lemma}
%
%\begin{proof}
%See appendix.
%\end{proof}
%

%%%%%%%%%%%%

%\section{}

%\section*{Acknowledgment}


\ifCLASSOPTIONcaptionsoff
  \newpage
\fi



\bibliographystyle{IEEEtran}
\bibliography{IEEEabrv,GirayBib}

\end{document}